\documentclass[11pt]{article}
\usepackage{hyperref}
\hypersetup{
    linktocpage=true,
    colorlinks=true, 
    linkcolor=blue, 
    citecolor=blue, 
    urlcolor=blue 
}
\usepackage{amsmath, amssymb}
\usepackage{amsthm}
\usepackage{MnSymbol}

\usepackage{fullpage}
\usepackage{enumitem}
\usepackage{complexity}
\usepackage{mathdots}

\usepackage{tikz}
\usetikzlibrary{quantikz}
\usetikzlibrary{intersections}
\usepackage{tikzit}

\tikzstyle{box}=[shape=rectangle, text height=1.5ex, text depth=0.25ex, yshift=0.5mm, fill=white, draw=black, minimum height=5mm, yshift=-0.5mm, minimum width=5mm, font={\small}]
\tikzstyle{Z dot}=[inner sep=0mm, minimum size=2mm, shape=circle, draw=black, fill={rgb,255: red,221; green,255; blue,221}]
\tikzstyle{Z phase dot}=[minimum size=5mm, font={\footnotesize\boldmath}, shape=rectangle, rounded corners=2mm, inner sep=0.2mm, outer sep=-2mm, scale=0.8, tikzit shape=circle, draw=black, fill={rgb,255: red,221; green,255; blue,221}, tikzit draw=blue]
\tikzstyle{X dot}=[Z dot, shape=circle, draw=black, fill={rgb,255: red,255; green,136; blue,136}]
\tikzstyle{X phase dot}=[Z phase dot, tikzit shape=circle, tikzit draw=blue, fill={rgb,255: red,255; green,136; blue,136}, font={\footnotesize\boldmath}]
\tikzstyle{hadamard}=[fill=yellow, draw=black, shape=rectangle, inner sep=0.6mm, minimum height=1.5mm, minimum width=1.5mm]
\tikzstyle{vertex}=[inner sep=0mm, minimum size=1mm, shape=circle, draw=black, fill=black]
\tikzstyle{vertex set}=[inner sep=0mm, minimum size=1mm, shape=circle, draw=black, fill=white, font={\footnotesize\boldmath}]

\tikzstyle{hadamard edge}=[-, dashed, dash pattern=on 2pt off 0.5pt, thick, draw={rgb,255: red,68; green,136; blue,255}]
\tikzstyle{brace edge}=[-, tikzit draw=blue, decorate, decoration={brace,amplitude=1mm,raise=-1mm}]
\tikzstyle{diredge}=[->]
\tikzstyle{dashed edge}=[-, dashed, dash pattern=on 2pt off 0.5pt, draw={rgb,255: red,64; green,64; blue,64}]

\usepackage{braket}
\usepackage{framed}
\usepackage{cleveref}
\usepackage[toc]{appendix}

\usepackage[noend]{algpseudocode}

\usepackage[textsize=footnotesize]{todonotes}

\makeatletter
\newtheorem*{rep@theorem}{\rep@title}
\newcommand{\newreptheorem}[2]{%
	\newenvironment{rep#1}[1]{%
		\def\rep@title{#2 \ref{##1}}%
		\begin{rep@theorem}}%
		{\end{rep@theorem}}}
\makeatother

\theoremstyle{plain}
\newtheorem{theorem}{Theorem}
\newreptheorem{theorem}{Theorem}

\newtheorem{lemma}[theorem]{Lemma}
\newreptheorem{lemma}{Lemma}
\newtheorem{proposition}[theorem]{Proposition}
\newtheorem{corollary}[theorem]{Corollary}
\newtheorem{conjecture}[theorem]{Conjecture}
\newtheorem{problem}[theorem]{Problem} 
 
\newtheorem{fact}[theorem]{Fact}

\theoremstyle{definition}
\newtheorem{definition}[theorem]{Definition}
 
\newtheorem{observation}[theorem]{Observation}

\newclass{\Rel}{Rel}
\newclass{\Samp}{Samp}
\newclass{\Inter}{Inter}
\newclass{\SampCliff}{SampCliff}
\newclass{\SampQNC}{SampQNC}
\newclass{\DLOGTIME}{DLOGTIME}
\newclass{\CliffordClass}{Clifford}
\newclass{\BPAC}{BPAC}
\newclass{\BPTC}{BPTC}
\newclass{\ZPAC}{ZPAC}
\newclass{\RPAC}{RAC}

\newlang{\ldagpar}{LDAGParity}
\newlang{\cnotmult}{CNOTMult}
\newlang{\promisecnotmult}{CNOTMult^*}

\providecommand{\cL}{\ComplexityFont{L}}
\newcommand{\parityL}{\parity \cL}

\newcommand{\oracle}{\mathcal O}
\newcommand{\coracle}{\mathcal R}

\newlang{\CliffSim}{CliffSim[2]}
\newlang{\HLF}{HLF}

\newcommand{\CNOT}{\operatorname{CNOT}}
\newcommand{\CXX}{\operatorname{CXX}}
\newcommand{\CZ}{\operatorname{CZ}}
\newcommand{\CSIGN}{\operatorname{CZ}}
\newcommand{\RZ}{R_z}
\newcommand{\RX}{R_x}

\newcommand{\SWAP}{\operatorname{SWAP}}

\renewcommand{\(}{\left(}
\renewcommand{\)}{\right)}
\newcommand{\supp}{\operatorname{supp}}
\newcommand{\rank}{\operatorname{rank}}
\newcommand{\diag}{\operatorname{diag}}

\newcommand{\GL}{\mathrm{GL}}
\newcommand{\F}{\mathbb F}

\newcommand{\Paulis}{\ensuremath\mathcal{P}}
\newcommand{\Clifford}{\ensuremath\mathcal{C}}
\newcommand{\PStr}[1]{\ensuremath\mathcal P_{#1} / \mathcal Z_{#1}}
\newcommand{\normalin}{\trianglelefteq}

\newcommand{\pfont}[1]{\texttt{#1}}

\begin{document}

\title{Interactive shallow Clifford circuits: quantum advantage against $\NC^1$ and beyond}
\author{
Daniel Grier\thanks{This work was also completed at MIT with support from an NSF Graduate Research Fellowship under Grant No.\ 1122374.} \\ University of Waterloo \\ \texttt{dgrier@uwaterloo.ca} \and
Luke Schaeffer \\ University of Waterloo \\ \texttt{lrschaeffer@gmail.com}
}
\date{}
\maketitle


\begin{abstract}
	
Recent work of Bravyi et al.\ and follow-up work by Bene Watts et al.\ demonstrates a quantum advantage for shallow circuits: constant-depth quantum circuits can perform a task which constant-depth classical (i.e., $\AC^0$) circuits cannot. Their results have the advantage that the quantum circuit is fairly practical, and their proofs are free of hardness assumptions (e.g., factoring is classically hard, etc.).   Unfortunately, constant-depth classical circuits are too weak to yield a convincing real-world demonstration of quantum advantage. We attempt to hold on to the advantages of the above results, while increasing the power of the classical model. 
	
Our main result is a two-round interactive task which is solved by a constant-depth quantum circuit (using only Clifford gates, between neighboring qubits of a 2D grid, with Pauli measurements), but such that any classical solution would necessarily solve $\parityL$-hard problems. This implies a more powerful class of constant-depth classical circuits (e.g., $\AC^{0}[p]$ for any prime $p$) unconditionally cannot perform the task.  Furthermore, under standard complexity-theoretic conjectures, log-depth circuits and log-space Turing machines cannot perform the task either. 
	
Using the same techniques, we prove hardness results for weaker complexity classes under more restrictive circuit topologies. Specifically, we give $\QNC^0$ interactive tasks on $2 \times n$ and $1 \times n$ grids which require classical simulations of power $\NC^1$ and $\AC^{0}[6]$, respectively. Moreover, these hardness results are robust to a small constant fraction of error in the classical simulation.

We use ideas and techniques from the theory of branching programs, quantum contextuality, measurement-based quantum computation, and Kilian randomization.
\end{abstract}

\clearpage
\tableofcontents
\clearpage


\section{Introduction}
\label{sec:introduction}

A long-standing goal of quantum computing research is to definitively establish that there are problems quantum computers can solve which classical computers cannot.  Many candidate problems have been proposed, from decision problems such as factoring \cite{Sho97} to sampling problems such as BosonSampling \cite{aaronsonarkhipov:2013}, IQP circuit sampling \cite{bjs:2010_iqp}, or random circuit sampling \cite{boixo:2018_rcs}.

Unfortunately, each proposal suffers from the same unpleasant tradeoff---as the proposal becomes more practical, it also relies on more nonstandard complexity assumptions.  For example, factoring has a long history and many believe it to require super-polynomial classical time.  Nevertheless, a convincing demonstration of quantum supremacy using Shor's factoring algorithm would require thousands of qubits, well beyond what is currently feasible.  In the same vein, to show sampling hardness based on long-accepted complexity assumptions (e.g., the non-collapse of the polynomial hierarchy) requires the quantum device to have little to no error, also eliminating the possibility of actually executing one of these sampling protocols in the lab.  Once errors are incorporated into the model, the hardness result becomes dependent on relatively new and untested conjectures.  We view the selection of conjectures as an extremely important process in establishing quantum supremacy. Indeed, plausible candidates for quantum advantage such as recommendation systems based on low-rank matrix completion \cite{kerenidisprakash:2017} have been refuted later by clever classical algorithms \cite{tang:2019}.  

This raises an obvious question:  can we avoid the above tradeoff when designing a protocol for quantum advantage?  A breakthrough result of Bravyi, Gosset, and K\"{o}nig gives a positive answer to this question \emph{provided} you are willing to restrict to quantum/classical circuits of constant depth \cite{bgk:2018}. They introduce the ``2D Hidden Linear Function Problem'' ($\HLF$), which can be implemented by a quantum circuit of constant depth, using only classically-controlled Clifford gates between adjacent qubits on a grid.  However, the same function \emph{cannot} be computed by any constant-depth classical circuit comprised of bounded fan-in gates.  To emphasize, this impossibility result is unconditional and does not rely on any assumptions or conjectures.  

Although constant-depth bounded fan-in classical circuits (i.e., $\NC^{0}$ circuits) vs.\ constant-depth bounded fan-in quantum circuits (i.e., $\QNC^0$ circuits) is a fair comparison, $\NC^{0}$ is an extremely weak class of circuits, leaving lots of room for improvement. Indeed, there have been several follow-up papers which have strengthened the result by considering average-case versions of the problem \cite{coudronstarkvidick:2018, legall:2019, bwkst:2019}, expanding the class of classical circuits \cite{bwkst:2019}, and adding noise \cite{bgkt:2019}.  Of particular relevance to this work is the result of Bene Watts, Kothari, Schaeffer, and Tal \cite{bwkst:2019} which shows that even classical circuits with \emph{unbounded} fan-in AND and OR gates (i.e., the circuit class $\AC^{0}$) cannot cannot solve $\HLF$.
 
The goal of this paper is to continue to expand the power of the classical models of computation which cannot simulate low-depth quantum computation.  We introduce two new problems solvable by constant-depth classically-controlled Clifford circuits on a grid and prove the problems are hard for complexity classes beyond $\AC^{0}$.  In the first result, the classical model must be able to compute polynomial-size circuits of log depth with bounded fan-in gates (i.e., the class $\NC^1$).  In the second result, the classical model must be able to compute polynomial-size circuits of $\CNOT$ gates (i.e., the class $\parityL$).   The $\NC^{1}$ result may appear strictly weaker than the $\parityL$ result since $\NC^{1} \subseteq \cL \subseteq \parityL$ (as uniform decision classes), but the problem has two distinct advantages:  the classical circuit is allowed to make some errors, and the quantum circuit is embedded on a very narrow grid.

Before discussing these problems in detail, let us first discuss the types of problem we are considering.  First, we adopt the relational view of quantum circuit simulation initiated by Bravyi et al.\ for classically-controlled Clifford circuits \cite{bgk:2018}. That is, a valid classical simulation of the quantum circuit can return \emph{any} measurement outcome a genuine quantum device may output with nonzero probability.\footnote{This is contrasted with sampling problems, e.g., BosonSampling and random circuit sampling, in which the task is to be close to the output distribution of the quantum device.} Unfortunately, we believe\footnote{For more details see \Cref{conj:nc1_vs_qnc0} in \Cref{sec:model}.} that proving a separation against $\NC^1$ under this model will require new nontrivial techniques in circuit complexity.  Instead, we introduce \emph{interactivity} into our model as a way to circumvent these challenges.

In a non-interactive protocol, there is one round where the quantum device or classical simulator receives an input and then returns an output. In an interactive protocol, this may be followed by more rounds of input and output, where each round may depend on the previous rounds. Specifically, all of the problems in this paper follow the same \emph{two}-round interactive protocol:
\begin{enumerate}[itemsep = 0pt]
\item  The quantum device (or classical simulator) is given measurement bases for all but constantly many qubits of a constant-depth Clifford circuit. It performs (or simulates) the circuit, measures the qubits, and returns the outcomes. 
\item In the second round, the device is given measurement bases for the remaining qubits and reports the answer.  
\end{enumerate}
We argue that this interactive protocol is realistic for near-term quantum devices while simultaneously requiring more powerful classical circuits to simulate the protocol.

In terms of practicality, all we are asking is for a tiny portion of the qubits to be measured apart from the rest.  We expect this to be feasible for any near-term quantum device.  One possible objection is that the qubits of the device could have a limited lifespan and might fail if there is too much delay between the rounds of interaction. For example, if the second round input is determined by a lengthy computation on the first round output, then it might delay the second round long enough that the qubits decohere.  Happily, no such processing is required in our interactive protocol.  In fact, our hardness reductions \emph{ignore} the first round output. It is only important that the simulator/device commits to this output, but the particular string does not matter. 

Now let us consider the power of the classical simulator.  By the main theorems, the classical simulator must be able to solve problems in the complexity class $\NC^1$.  This immediately gives the unconditional result that there is no $\AC^0[p]$ circuit ($\AC^0$ circuits with $\MOD_p$ gates\footnote{The $\MOD_k$ gate outputs $1$ if the sum of the input bits is $0$ modulo $k$.}) which can simulate the quantum circuit since $\AC^0[p] \subsetneq \NC^1$ for all prime $p \ge 2$.  Since $\AC^0[p] \supsetneq \AC^0$, this result strictly improves upon the separation of Bene Watts et al., albeit at the cost of this 2-round interactivity.  Of course, assuming standard conjectures in complexity theory, we get much stronger results. The following inclusions are believed to be strict
$$
\AC^0[p] \subset \TC^0 \subseteq \NC^1 \subseteq \cL \subseteq \parityL,
$$
so, for example, there is no $\TC^{0}$ circuit for the $\NC^{1}$-hard task, and no $\cL$ simulator for the $\parityL$-hard task. 

The importance of interactivity in our model is tied to the no-cloning theorem in quantum mechanics.  In particular, we will use the fact that the classical simulator has some \emph{copyable} state at the end of the first round of interaction, whereas a quantum device does not.  After the second round of interaction, we can \emph{rewind} the classical simulator to the earlier state, an idea common in security proofs for (interactive) cryptographic protocols.  This is the difference between classical and quantum devices that makes our result non-blackbox. Intuitively, rewinding gives the classical simulator more power which the quantum device cannot match; the quantum device cannot rewind, and if it resets to the beginning, it has an exponentially small chance of measuring the same first-round outcome. This key observation allows us to prove the main hardness results. 

\subsection{Results}

The purpose of this section is to state the main results of this paper as cleanly as possible by glossing over some of the more tedious details of the model and implementation.  

\begin{theorem}[informal]
\label{thm:informal_main}
For all $n$ and $m$, we define a 2-round interactive task $\mathcal T_{m,n}$ which can be passed by a constant-depth classically-controlled Clifford circuit over an $m \times n$ grid with gates between neighboring qubits.  Suppose a classical machine solves this task for all $n$, and $m$ in one of the three regimes below. Let $\coracle$ be the oracle for its responses. 
\begin{itemize}[itemsep = 0pt]
\item If $m = 1$, then $\AC^{0}[6] \subseteq (\BPAC^{0})^{\coracle}$.
\item If $m = 2$, then $\NC^{1} \subseteq (\BPAC^{0})^{\coracle}$.
\item If $m = \poly(n)$, then $\parityL \subseteq (\BPAC^{0})^{\coracle}$.
\end{itemize}
Here, $\BPAC^0$ is the class of problems solved by random $\AC^0$ circuits with bounded error.\footnote{In fact, all of the reductions are in the slightly smaller class $\ZPAC^0 = \RPAC^0 \cap \co\RPAC^0 \subseteq \BPAC^0$.  Equivalently, this is the class of problems solved by random $\AC^0$ circuits such that the circuit outputs the correct answer (``yes'' or ``no'') with probability at least 1/2 and outputs ``do not know'' with probability at most 1/2.}
\end{theorem}

This theorem leads to two quantum/classical separations, one unconditional and one dependent on complexity theoretic assumptions.

\begin{corollary}
There is an interactive task that $\QNC^0$ circuits can solve that $\AC^0[p]$ circuits cannot.
\end{corollary}

\begin{corollary}
Assuming $\cL/\poly \not\supseteq \parityL/\poly$, there is an interactive task solved by $\QNC^0$ circuits but not logarithmic-space Turing machines.\footnote{The statement of this corollary is just one possible consequence of \Cref{thm:informal_main}.  We can weaken the assumption at the expense of weakening the classical hardness.  For example, assuming $\NC^1/\poly \not \supseteq \parityL/\poly$, we have a separation between $\QNC^0$ and $\NC^1$.}
\end{corollary}

The three subresults of the main theorem are contained in \cref{sec:nc1,sec:parityL} in \cref{thm:ac0mod6_main,thm:nc1_main_result,thm:parityL_main_result}.  The $\NC^1$-hardness result is conceptually easier than the $\parityL$-hardness result, but has the added benefit of allowing the classical simulator to err with some small probability.  In \Cref{sec:model}, we show that the complexity of these types of problems is upper bounded by $\parityL$.  

\begin{theorem}[informal, proof in \Cref{sec:upper_bounds}]
The problem of returning a valid output in round 1 of task $\mathcal T_{m,n}$ is in $\parityL$.  The problem of returning a valid output in round 2 of task $\mathcal T_{m,n}$ conditioned on some round 1 output is in $\parityL$.
\end{theorem}

\subsection{Proof Outline}

All of our hardness results are based on the same general framework:
\begin{enumerate}
	\item Simulating a sequence of Clifford gates is computationally hard, in the sense that the final quantum state contains the solution to a hard decision problem.
	\item Using measurement-based quantum computation, we can collapse a sequence of Clifford gates to a constant-depth circuit, but we introduce Pauli errors on the final quantum state. 
	\item We get an interactive task by measuring the qubits in two rounds, and show how the classical simulator can recover information about the quantum state by rewinding.  
	\item Use random self-reducibility to repeat the process and accumulate enough information about the state to solve the decision problem. 
\end{enumerate} 

We build our hardness results on the difficulty of simulating sequences of classically-controlled Clifford gates. That is, there is a fixed sequence of Clifford gates and a classical input which tells us whether or not to apply each gate in the sequence, much like a branching program. The problem is to simulate (i.e., output a string of measurement outcomes with nonzero probability) the sequence on a fixed initial state, and the complexity depends on the number of qubits. In particular, we show that simulating a sequence of Clifford gates on two qubits is $\NC^{1}$-complete, by first characterizing the group of two-qubit Clifford gates and applying a result of Barrington and Th{\'e}rien \cite{BarringtonTherien}. We also show that simulating a particularly nice sequence of $\CNOT$ gates on $n$-qubits is $\parityL$-hard, by a simple reduction from known $\parityL$-hard problems \cite{damm:1990}. 

However, the simulation problems outlined above are unsuitable because the natural quantum implementation (i.e., literally executing a sequence of classically-controlled Clifford gates, then measuring) is \emph{not} constant depth.\footnote{Interestingly, it is known that any Clifford circuit has a constant depth implementation (using $\tilde{\Theta}(n^2)$ ancillas) using unbounded fan-in parity gates (or equivalently, fan-out gates). This is an easy consequence of Moore and Nilsson \cite{MooreNilsson:2002}, where they show how to implement $\CNOT$ circuits in constant depth (actually $\Theta(\log n)$-depth, but only because they insist on expanding the unbounded fan-in gates as trees of constant fan-in gates), and the Aaronson-Gottesman decomposition of Clifford circuits into constantly many layers of single qubit gates and $\CNOT$ circuits.} We use an observation of Raussendorf, Browne, and Briegel \cite{rbb:2003_mbc} that their procedure for measurement-based quantum computation (MBQC) works without adaptivity for Clifford circuits. Ordinarily, MBQC introduces errors in the form of Pauli operators (they call them \emph{byproduct operators}), which must be fixed adaptively, with a layer of adaptivity (or more) for each layer of the original circuit.\footnote{Such adaptivity is unacceptable for near-term quantum devices with short coherence times.} When all the gates are Clifford, however, the Pauli errors can be pushed (via conjugation relations) to the end. In principle, the final Pauli errors can be computed (in $\parityL$), but since this is too expensive for our hardness reductions, the Pauli error is effectively unknown to us.


The task becomes interactive when we split the MBQC procedure into two measurement rounds: the first round performs a circuit, setting up a quantum state, then the second round measures that state. When the simulator is classical, we gain the power to measure the second-round state multiple times (by copying the internal state). 
The power to measure a state repeatedly, under multiple different bases, would seem to be sufficient to learn the state by tomography. However, recall that we define simulation to be returning measurement outcomes that are \emph{possible} (i.e., occur with nonzero probability), rather than sampling measurement outcomes. In particular, the classical simulator may be designed adversarially to thwart our attempt to learn the state. Standard tomography results depend on measurement outcomes being random, not adversarial, so they break under this model. Instead, we use contextual measurements (i.e., from the magic square game or magic pentagram game \cite{peres_magic_square, mermin_magic_square}) to force the classical simulator to reveal a Pauli string which does not stabilize the state.\footnote{If the measurement outcomes were truly random (instead of being adversarially chosen), we could use conventional tomography to skip this step. However, defining the task as an interactive sampling problem would give a strictly weaker result since any kind of error-free sampling would suffice for relational simulation.}

Although we learn a non-stabilizing Pauli operation for the state, it could be chosen adversarially. In particular, if we need to decide whether a state is, e.g., $\ket{00}$ or $\ket{++}$, then there are many choices of Pauli operators which do not distinguish the two (e.g., $\pfont{YI}$).  In the final step, we randomize the input such that the non-stabilizing Pauli is also random.  This self-reducibility property allows us to collect all non-stabilizing Paulis and deduce the state.  This finishes the reduction since the state contains the solution to the hard problem.

\subsection{Related Work}

It is clear that the original work of Bravyi, Gosset, and K\"{o}nig which separates $\QNC^0$ and $\NC^0$ is highly relevant to the results of this paper.  In fact, these authors, with the addition of Tomamichel, have a follow-up result which also shares many similarities with our own \cite{bgkt:2019}.  There, the authors give a new problem that is solvable by a one-dimensional (i.e., the qubits are arranged in a line) constant-depth quantum circuit, but is impossible for classical constant-depth circuits. They then show that the problem (and others like it) can be transformed to one solvable by \emph{noisy} quantum circuits with high probability, albeit on a 3D array of qubits.  The former result can be broken into two logical steps:  create a one-dimensional constant-depth quantum circuit which can ``play'' the magic square game between arbitrary pairs of input qubits; then, show that any bounded fan-in circuit for this problem requires super-constant depth. In fact, it is this invocation of the magic square game that inspired our own use of nonlocal games and contextual measurements in this paper.

Our result also has a similar flavor to the paper of Shepherd and Bremner \cite{bremnershepherd:2009_interactive} in which they consider interactive protocols for verifying quantum advantage.  They design a task that an IQP circuit can pass that a $\BPP$ machine cannot. In some sense, their protocol is preferable to ours because it verifies a quantum advantage against arbitrary classical polynomial-time computation, rather than against \emph{low-depth} classical computation.  However, their protocol suffers from the fact that it relies on several assumptions, including a conjecture specific to their problem about obfuscating matroids.  Furthermore, their quantum circuit is still harder to implement than our own, requiring long-range Hamiltonians.


\section{Background}
\label{sec:background}

This section serves to introduce the Clifford group, measurement-based quantum computation, as well as the many low-depth circuit classes we reference throughout this paper.  Readers familiar with the Clifford group are still advised to read the relevant section below, as it primarily focuses on a nonstandard notion of Clifford operations modulo Pauli operators.  Readers unfamiliar with the Clifford group can get a more gentle introduction in \Cref{app:clifford}.

\subsection{The Clifford group and its quotients}
\label{sec:background_clifford}

Let $\pfont{I}$, $\pfont{X}$, $\pfont{Y}$, and $\pfont{Z}$ be the four standard Pauli matrices.  We write the \emph{$m$-qubit Pauli group} as $\Paulis_m := \{ \pm 1, \pm i \} \times \{ \pfont{I}, \pfont{X}, \pfont{Y}, \pfont{Z} \}^{\otimes m}$.  We call the $\{ \pm 1, \pm i \}$ component the \emph{phase}, and the $\{ \pfont{I}, \pfont{X}, \pfont{Y}, \pfont{Z}\}^{\otimes m}$ component the \emph{Pauli string}.  We name the group of phases $\mathcal{Z}_m := \{ \pm 1, \pm i \} \times \pfont{I}^{\otimes m}$, and note that $\mathcal{Z}_m$ is a normal subgroup of $\mathcal{P}_m$.  This means the quotient $\PStr{m}$ is well-defined.  Each element of $\PStr{m}$ is a coset $\{+ P, - P, +i P , -iP \}$ for some $P \in \{ \pfont{I}, \pfont{X}, \pfont{Y}, \pfont{Z}\}^{\otimes m}$, but we identify each such element with $P$, its \emph{positive} representative.

We write the \emph{$m$-qubit Clifford group} as
$$
\mathcal{C}_m := \{ U \in \mathrm{U}(2^{m}) : U \Paulis_m U^{\dag} = \Paulis_m \}.
$$
Since conjugation of a Pauli by a Clifford operation is so common, we define the notation $\bullet \colon \Clifford_m \times \Paulis_m \to \Paulis_m$ where $U \bullet P := U P U^{\dag}$ for any $U \in \Clifford_m$ and $P \in \Paulis_m$.  By construction, $\Paulis_m$ is a normal subgroup of $\mathcal{C}_m$, so 
$$
\mathcal{Z}_m \normalin \Paulis_m \normalin \Clifford_m
$$

We will build Clifford circuits from the familiar $\CNOT$, CSIGN ($\CZ$), Hadamard ($H$), and Phase ($\RZ = \RZ(\pi/4) = (\begin{smallmatrix} 1 & 0 \\ 0 & i \end{smallmatrix})$ gates.  A \emph{Clifford state} or \emph{stabilizer state} is any quantum state of the form $U \ket{0}^{\otimes m}$, where $U \in \mathrm{U}(2^m)$ is Clifford. We often define a Clifford state is by its \emph{stabilizer group}:
$$
\mathrm{Stab}_{\ket{\psi}} := \{ P \in \Paulis_m : P \ket{\psi} = \ket{\psi} \}.
$$ 

Due to the limitations of measurement based quantum computation (which we describe in more detail in \Cref{sec:mbqc}), our results will usually implement Clifford operations up to a Pauli correction. That is, instead of implementing $C \in \Clifford_m$, we implement $C P$ for some $P \in \Paulis_m$. 
We must be content to perform any Clifford operation in the same coset as the one we asked for, and we must talk about Clifford operations, Clifford states, and Pauli stabilizers \emph{modulo} $\Paulis_m$. Naturally, this leads us to work with several quotient groups. 

Since $\Paulis_m$ is normal in $\Clifford_m$, the quotient group $\Clifford_m / \Paulis_m$ is well-defined. When we assert that two Clifford operations $C_1, C_2 \in \Clifford_m$ are \emph{equivalent modulo Paulis}, or write $C_1 \equiv C_2 \pmod{\Paulis_m}$, we really mean that the cosets $C_1 \Paulis_m$ and $C_2 \Paulis_m$ are equal, or equivalently, $C_1 C_2^{-1} \in \Paulis_m$. Another way to characterize equivalent Clifford operations is by their action (by conjugation) on the Paulis modulo phase. 

\begin{lemma}
\label{lem:pauli_mod}
For all $C_1, C_2 \in C_m$,  $C_1 \equiv C_2 \pmod{\Paulis_m}$ iff $C_1 \bullet Q \equiv C_2 \bullet Q \pmod{\mathcal Z_m}$ for all $Q \in \Paulis_m$.
\end{lemma}

We prove this lemma in \Cref{app:clifford}.  Next, we want an analogous notion of equivalence for Clifford states modulo $\Paulis_m$. We say $\ket{\psi_1} \equiv \ket{\psi_2} \pmod{\Paulis_m}$ if there exists $P$ such that $\ket{\psi_2} = P \ket{\psi_1}$. It is not hard to check that this is an equivalence relation, and that equivalence is preserved under equivalent Clifford operations. 
\begin{lemma}
Two Clifford states $\ket{\psi_1}$ and $\ket{\psi_2}$ are equivalent if and only if their stabilizer groups contain all the same Pauli operations up to sign. That is, if $P \in \mathrm{Stab}_{\ket{\psi_1}}$ then $\pm P \in \mathrm{Stab}_{\ket{\psi_2}}$. 
\end{lemma}
\begin{proof}
The calculation 
$$
(C \bullet P) C \ket{\psi} = C P C^{\dag} C \ket{\psi} = C P \ket{\psi} = C \ket{\psi}
$$
shows that $C \bullet P$ is a stabilizer of $C \ket{\psi}$ if $P$ is a stabilizer of $\ket{\psi}$, where $C \in \Clifford_m$ and $P \in \Paulis_m$. So, if $\ket{\psi_2} = Q \ket{\psi_1}$ and $P$ is a stabilizer of $\ket{\psi_1}$ then $Q \bullet P = Q P Q^{\dag} = \pm P$ is a stabilizer of $\ket{\psi_2}$. 
\end{proof}


Finally, let us discuss Pauli measurement of some state $\ket{\psi}$. Given a Pauli string $P \in \Paulis_m / \mathcal{Z}_m$, the associated measurement randomly projects $\ket{\psi}$ onto one of the two eigenspaces of $P$ (provided the projection is nonempty) and reports the corresponding eigenvalue, $+1$ or $-1$.  For the vast majority of our applications, the property of Pauli measurement we use is that any Pauli string appearing in the stabilizer group has a deterministic measurement outcome, and all other Pauli measurements have inherently random outcomes. Thus, we will bend terminology and say $P \in \Paulis_2 / \mathcal{Z}_2$ is a \emph{stabilizer} of $\ket{\psi}$ if $P$ has a deterministic outcome on $\ket{\psi}$. Any other element of $\Paulis_2 / \mathcal{Z}_2$ is a \emph{non-stabilizer}, and will have a uniformly random outcomes. We will also abuse notation and write, e.g., $\ket{\psi} = U \ket{++}$ for $U \in \Clifford_2 / \Paulis_2$, even though $\ket{\psi}$ is not a state, since $U$ is a coset of unitaries. 

These notions of equivalence will be important throughout the paper, so let us recap. Clifford operations $C_1, C_2$ are equivalent modulo Paulis if $C_1 C_2^{-1}$ is a Pauli operation, and equivalent Clifford operations induce the same permutation of $\Paulis_m / \mathcal{Z}_m$. Clifford states $\ket{\psi_1}$ and $\ket{\psi_2}$ are equivalent if they have the same stabilizer group up to signs, i.e., for any Pauli $P$ in $\mathrm{Stab}_{\ket{\psi_1}}$, $\pm P \in \mathrm{Stab}_{\ket{\psi_2}}$. 

\subsection{Measurement-based computation}
\label{sec:mbqc}
One of the main techniques used throughout this paper is a model of computation developed by Raussendorf, Browne, and Briegel known as \emph{one-way quantum computation} \cite{rb:2001_1wqc, rbb:2003_mbc}.  For reasons that will soon be clear, this is also sometimes called \emph{measurement-based computation on cluster states}.  At a high level, measurement-based computation allows for the simulation of any quantum computation by performing a sequence of adaptive single-qubit measurements on a certain highly-entangled initial state.

First, let us describe the initial state.  Let $G = (V,E)$ be an undirected graph with $V = [n]$ and $E \subseteq [n]^2$.  Define the \emph{graph state} for $G$ as 
$$
\ket{G} \equiv \sum_{x \in \{0,1\}^n} \prod_{(u,v) \in E} (-1)^{x_u x_v} \ket{x}.
$$
Notice that any graph state can be constructed from the all zeros state by applying a Hadamard gate to each qubit, and then applying a $\CSIGN$ gate between each pair of qubits representing an edge.  Thus, the any graph state can be constructed in depth at most the maximum degree of the graph plus one (by edge coloring arguments).  A \emph{cluster state} is a special case of a graph state where the graph is a 2D grid.  Importantly, the cluster state can be constructed in constant depth, using at most 4 layers of $\CSIGN$.

Measurement-based computation consists of a sequence of measurement operations that result in a gate-by-gate simulation of a quantum circuit.   For a given gate $U$ and an initial state $\ket{\psi}$, there is a measurement procedure that applies $U$ to $\ket{\psi}$ by consuming a cluster state $\ket{G}$:
\begin{enumerate} [itemsep=0pt]
\item Prepare the state $\ket{\psi} \otimes \ket{G}$.
\item Apply $\CSIGN$ between $\ket{\psi}$ and the leftmost grid points of $\ket{G}$.
\item Measure each input qubit of $\ket{\psi}$ and all but the rightmost grid points of $\ket{G}$ in either the $Z$-basis or the $\{\ket{0} \pm e^{i\theta} \ket{1} \}$-basis.
\end{enumerate}
Assuming an appropriate choice of measurement basis in the last step, the unmeasured qubits will be in the state $P U \ket{\psi}$ where $P$ is some Pauli string that depends on the measurement outcomes. To apply multiple gates, we simple string together multiple instances of the above procedure.  Fortunately, each application of the $\CSIGN$ gate commutes with all previous measurements, so we can view the simulation of the entire circuit as a sequence of single-qubit measurements on one sufficiently large cluster state.  Unfortunately, there are Pauli errors between each of the gates.  Therefore, in general one must \emph{adaptively} apply the measurements for one gate to compensate for the Pauli errors made in the application of the previous gate.

In this paper, we will focus only on quantum circuits which are composed of Clifford gates.  Therefore, no such adaptivity is needed since all Pauli errors can be pushed to the end.  We will also use that fact that every Clifford operation can be applied using only $X$, $Y$, and $Z$-basis measurements.  This yields the following theorem:

\begin{theorem}[Raussendorf, Browne, and Briegel \cite{rbb:2003_mbc}]
\label{thm:rbb_mbqc}
For any $m$-qubit Clifford circuit $C$ with $n$ local gates, there exists a pattern of $X$, $Y$, and $Z$ single-qubit measurements on an $O(m) \times O(n)$ cluster state such that the unmeasured rightmost qubits of the cluster are in the state $P C \ket{+}^m$, where Pauli $P$ is a function of the measurement outcomes.
\end{theorem}

In \Cref{app:mbqc_gadgets}, we show how to modify the constructions of Raussendorf, Browne, and Briegel to minimize the dimensions of the grid. Although such constructions only save a constant fraction of qubits, we feel as though it is desirable to get the simplest possible circuit. One of the easiest ways to verify and explain these constructions is to appeal to the ZX-calculus, a graphical language for reasoning about linear maps between qubits.  We refer the reader to Coecke and Duncan \cite{cd:2011_zx} as their paper is quite clear and thorough, but we also introduce (in \Cref{app:zx}) the bare minimum necessary for our results. 

\subsection{Types of Problems}
\label{sec:models_of_computation}
Traditionally, a complexity class is a collection of problems solved by some computational device. This section serves to introduce the various kinds of problems, all of which either appear in this work or relevant literature.  We introduce the models of computation we will use to solve these problems in the next section.

\begin{description}
	\item[Decision problems:] A decision problem is a subset of strings $L \subseteq \{ 0, 1 \}^{*}$, and a computational device solves that problem if it \emph{accepts} precisely the strings in $L$ and rejects all others. In other words, the machine must give a standard yes-or-no answer for all classical inputs. By default, complexity classes are sets of decision problems. For example, $\NC^{1}$ is the collection of languages accepted by a uniform family of $\NC^{1}$ circuits $\{ C_n \}_{n \geq 0}$ where $C_n$ has $n$ input bits and one output bit, and the circuit accepts if and only if the output bit is $1$.
	\item[Relation problems:] A relation problem is defined by a relation $R \subseteq \{ 0, 1 \}^{*} \times \{ 0, 1 \}^{*}$ on strings. A computational device which takes an input string and returns a string solves a relation problem if the input-output pair satisfy the relation. In the interest of clarity, classes of relation problems will begin with $\Rel$.\footnote{Somewhat confusingly, relation problems are sometimes referred to as search or function problems in the literature.  For example, function polynomial time ($\FP$) refers to relation problems where a polynomial-time Turing machine can output any string satisfying the relation.} For example, $\Rel \NC^{1}$ is the class of relation problems such that there exists a uniform family of $\NC^{1}$ circuits which, for any input string, output some string such that the pair satisfies the relation. 
	\item[Sampling problems:] A sampling problem is like a relation problem, but each input defines a distribution over output strings, and the computational device is required to output strings at random from this distribution (sometimes \emph{exactly}, but by default with multiplicative error). Some computational devices are inherently deterministic and thus require a stream of random bits to be provided with the input. Classes of sampling problems will begin with $\Samp$. For example, $\Samp \QNC^{0}$ is the collection of sampling problems where the goal is to sample from distributions obtained by measurements on a uniform family of constant-depth quantum circuits. 
	\item[Interactive problems:] An interactive problem is a generalization from a single round of interaction, e.g., the device gets an input and returns an output, to multiple rounds.
An interactive problem is distinct from a series of relation problems (or sampling problems) in two ways: the computational device is allowed to keep some state from one round to the next (by explicitly passing bits/qubits for circuits, or by not erasing the work tape of a Turing machine between rounds), and the device is evaluated based on the entire transcript (i.e., the sequence of inputs and outputs). 
\end{description}

\subsection{Low-Depth Complexity Classes}

Let us start by defining some standard polynomial-size circuit families. For each of the families below, we adopt the convention that $\mathcal{C}^{i}$ denotes the subfamily of circuits in $\mathcal{C}$ with depth $O(\log^{i} n)$.

\begin{itemize}[itemsep = 0pt]
\item $\NC$: classical circuits of bounded fan-in AND, OR, and NOT gates.
\item $\AC$: classical circuits of unbounded fan-in AND, OR, and NOT gates.
\item $\TC$: classical circuits of unbounded fan-in MAJORITY and NOT gates.
\item $\QNC$: quantum circuits of $\CNOT$ gates and arbitrary single-qubit gates.
\item $\CliffordClass$: quantum circuits of $\CNOT$, $H$, and $\RZ$ gates.
\end{itemize}

The interactive quantum circuits in this paper belong to a restricted set of quantum circuits close to $\CliffordClass$ except each gate may be controlled by a classical input bit. Formally, a \emph{classically-controlled Clifford circuit} consists of a sequence Clifford gates, some of which are controlled by an individual bit of the classical input.  Controlled gates are applied if the input bit is $1$, act as the identity otherwise.  For example, see \Cref{fig:2_round_circuit}.

Finally, let us turn our attention to two important small-space Turing machine complexity classes:  $\cL$ and $\parityL$ (pronounced ``parity $\cL$"). Each class is defined through a log-space Turing machine. Formally, a log-space Turing machine has two tapes:  a read-only tape of length $n$ containing the input, and a read-write workspace tape of length at most $O(\log n)$.  A language is in $\cL$ if there is a deterministic log-space Turing machine which accepts every word in the language and rejects every word not in the language.  A language is in $\parityL$ if there exists a non-deterministic log-space Turing machine such that there are an odd number of accepting paths for every word in the language, and an even number of accepting paths for every word not in the language.

We can also define $\cL$ and $\parityL$ in terms of circuit families. A result of Cobham \cite{cobham:1966} says that non-uniform\footnote{Notice that classical and quantum circuits have a fixed number of inputs, so they can only solve computational problems on inputs of fixed length. When we say $\mathcal{C}$ circuits solve a problem with unbounded size inputs, we mean that there is a collection of circuits in $\mathcal{C}$, $\{ C_n \}_{n \geq 0}$, one for each input length, and the circuit for the appropriate length is applied to the input. In principle, each input length could be solved by a wildly different circuit, i.e., the class we have defined is \emph{non-uniform}. This can be used to solve undecidable problems, which makes it awkward to compare such classes to \emph{uniform} classes defined by a single device (e.g., a Turing machine) which acts on all input lengths. For this reason, we primarily work with \emph{uniform} circuit classes, where the family $\{ C_n \}_{n \geq 0}$ is generated by a Turing machine.  When we require non-uniformity, we give the device classical advice.  This is denoted by the suffix $/\poly$ on the class.} log space (i.e., $\cL/\poly$) is equivalent to polynomial-width branching programs. Combined with the result \cite{lange:2000} that log space is equivalent to reversible log space, we may assume the layers of the branching programs are permutations---essentially a classically-controlled network of swaps. Thus, log space Turing machines are equivalent to uniform classically-controlled circuits of swap gates. Similarly, computing the product of a network of $\CNOT$ gates is complete for $\parityL$ \cite{damm:1990}, so we may think of $\parityL$ as defined by a classically-controlled network of $\CNOT$ gates.

For reference, we have the following relationships between the (uniform) versions of the (decision) complexity classes:
$$
\NC^0 \subset \AC^0 \subset \AC^0[2] \subset \TC^0 \subseteq \NC^1 \subseteq \cL \subseteq \parityL = \CliffordClass \subseteq \AC^1[2]
$$
Note that we write $\mathcal{C}[G]$ for the circuit class or complexity class $\mathcal{C}$ augmented with access to the gate $G$. In the special case of unbounded fan-in MOD gates (i.e., $\MOD_k$ outputs $1$ if the sum of the input bits is $0$ modulo $k$), we may write just the modulus, e.g., $\AC^{0}[2] = \AC^{0}[\MOD_2]$.


\section{Model}
\label{sec:model}

In this section, we discuss the precise interactive model we use for our hardness of simulation results. 
Before we discuss these protocols, let us explain why we need interaction, and why blackbox reductions such as those used in previous separations \cite{bgk:2018, bwkst:2019} may not give us, e.g., $\NC^{1}$-hardness results for simulating constant-depth quantum circuits. 

\subsection{Ruling out blackbox reductions}

Although complexity theorists know how to prove specific problems are not in $\NC^{0}$ (via light cones, for instance), $\AC^{0}$ (e.g., using the switching lemma), or even $\AC^{0}[2]$ (see Razborov-Smolensky \cite{razborov:1987, smolensky:1987}), proving unconditional circuit lower bounds for even slightly larger classes (e.g., $\TC^{0}$) is beyond our current tools. The kind of hardness we hope to prove (i.e., $\NC^{1}$-hardness or $\parityL$-hardness) should therefore be \emph{conditional} on complexity theoretic assumptions. For example, if we could prove by a reduction that a constant-depth quantum circuit solves an $\NC^{1}$-complete problem, then we could say $\TC^{0} \subsetneq \QNC^{0}$ \emph{conditioned} on the (widely assumed) conjecture that $\TC^{0} \subsetneq \NC^{1}$. 

However, we think it is unlikely that (blackbox) oracle access to $\QNC^{0}$ circuits can help solve $\NC^{1}$-hard problems. To formalize this, let $\QNC^{0}_f = \QNC^{0}[\mathsf{fanout}]$ be the family of $\QNC^{0}$ circuits augmented with unbounded fan-out gates. That is, fan-out is the classical reversible gate which XORs a single control bit into any number of target bits (it can be constructed from a deep but straightforward network of $\CNOT$ gates from the control to each target).  We make the following precise conjecture.
\begin{conjecture}
\label{conj:nc1_vs_qnc0} $\NC^{1} \not \subseteq \QNC^{0}_f$.
\end{conjecture}
We can think of this as an extension of the conjecture that $\TC^{0}$ does not contain $\NC^{1}$, because a result of H{\o}yer and \v{S}palek \cite{hoyerspalek:2005} and Takahashi and Tani \cite{tt:2016} shows how to construct majority gates in $\QNC^{0}_f$, and therefore $\TC^{0}$ circuits can be implemented directly in $\QNC^{0}_f$. We make the conjecture on the basis that the H{\o}yer and \v{S}palek result has not been extended to $\NC^{1}$ in the intervening 15 years, perhaps because their work depends on executing sequences of \emph{commuting} gates, but $\NC^{1}$ can compute products over non-abelian groups \cite{BarringtonTherien}. 

The relevant consequence of the conjecture is that an oracle for some relation problem in $\Rel\QNC^{0}$ cannot help solve $\NC^{1}$-hard problems. 
\begin{corollary}
\label{cor:no_black_box}
Suppose $A$ is a relation problem in $\Rel\QNC^{0}$. Under the conjecture, $(\TC^{0})^{A}$ does not contain $\NC^{1}$.  
\end{corollary}
\begin{proof}
Suppose for a contradiction that $(\TC^{0})^{A}$ solves some $\NC^{1}$-hard problem, and therefore $\NC^{1}$ is contained in $(\TC^{0})^{A}$. The $(\TC^{0})^{A}$ circuit can be translated to a $\QNC^{0}_f$ circuit, since majority gates can be constructed, fan-out gates are assumed, and $\QNC^{0}$ gates are part of $\QNC^{0}_f$ by definition. It follows that 
$$\NC^{1} \subseteq (\TC^{0})^{A} \subseteq \QNC^{0}_f,$$
contradicting the conjecture. 
\end{proof}

The key point is that this conjecture and the corollary above rule out a \emph{black box} hardness reduction. We want a task in $A \in \Rel\QNC^{0}$ such that a classical implementation of $A$ solves $\NC^{1}$-hard problems, but the conjecture implies a quantum implementation of $A$ does not solve $\NC^{1}$-hard problems. The only way to have both results is to open up the oracle (i.e., black box) and have the proof depend on the classical internals somehow. For example, the sampling problems separating quantum and classical machines (e.g., Aaronson-Arkhipov \cite{aaronsonarkhipov:2013}) use the fact that the classical machine may be assumed WLOG to be deterministic, with the random bits fed into it as a string. Taking out the randomness allows us to estimate probability amplitudes via approximate counting, where no similar idea is possible with a quantum machine.

However, for \emph{relation} problems, we found no such technique to separate classical and quantum devices.\footnote{One might be tempted to switch to sampling instead of relation problems.  Not only would those arguments suffer from the same limitation on black box reductions, we show in \Cref{app:samp_vs_search} that the sampling and relation versions of Clifford simulation problems such as $\HLF$ are equivalent under low-depth reductions.} Instead, we switch to \emph{interactive} problems. Notice that this does not overcome the black box argument above; we do not expect an oracle solving the interactive task to assist in solving $\NC^{1}$-hard problems. What interactivity gives is an easy way to extract more power (and thus solve hard problems) from a classical simulation of the task than from an actual quantum device performing the same task.

\subsection{The interactive model}

Let us imagine two parties:  a prover (who is supposed to give answers that are consistent with a low-depth quantum computer), and a challenger.  The low-depth quantum computer starts the protocol with the graph state $\ket{G}$ where $G$ is a subgraph of a 2D grid.  Let $\ket{\psi_0} = \ket{G}$. The $i$th round of the protocol consists of the following:
\begin{itemize} [itemsep = 0pt]
\item The challenger chooses a set of non-overlapping, local, Clifford gates on $\ket{\psi_{i-1}}$.
\item The prover returns an outcome consistent with first applying the set of gates and then measuring those qubits in the $Z$-basis.  Let $\ket{\psi_i}$ be the state of the unmeasured qubits consistent with the measurement results.
\end{itemize}
It is worth stressing here that the measurement results of the prover do not need to be uniformly chosen from the possible outcomes.  There simply must be some positive probability that a quantum computer faithfully executing the protocol returns those answers.

We now ask what the computational complexity is for passing such a protocol.  As a first observation, notice that if the protocol has polynomially many rounds, then the challenger can force the prover to simulate measurement-based computation on cluster states, and thus simulate an arbitrary Clifford circuit. The prover would then necessarily need to have the power of $\parityL$.  At the other extreme, if we only ask for a single-round of measurements, then this a relation problem, which we argued previously is unlikely to lead to a separation.  Therefore, this paper will focus on these interactive protocols that have exactly two rounds as shown in \Cref{fig:2_round_circuit}.

\DeclareExpandableDocumentCommand{\ccontrol}{O{}m}{|[phase,inner sep=2.2pt,#1]| {} \cw}
\def\cctrl#1{\ccontrol{}	\vcw{#1}}

\begin{figure}[h!]
\begin{center}
\begin{quantikz}[thin lines]
& \cw\qwbundle{} & \cw & \cw & \cw & \cw  & \cctrl{2}\gategroup[wires=3, steps=4, style={dashed, rounded corners, inner xsep=2pt}, label style={label position=below,anchor= north,yshift=-0.2cm}, background]{Round 2} &  \ \ldots\ \cw & \cctrl{2} & \cw \\
& \cw\qwbundle{} & \cctrl{2}\gategroup[wires=3, steps=4, style={dashed, rounded corners, inner xsep=2pt}, label style={label position=below,anchor=north,yshift=-0.2cm}, background]{Round 1}  &  \ \ldots\ \cw & \cctrl{2} & \cw & \cw & \cw & \cw & \cw \\
\lstick[wires=2]{$\ket{G}$}  & \qw\qwbundle{} & \qw & \qw & \qw & \qw & \gate{D_1} &  \ \ldots\ \qw & \gate{D_{\ell}} & \meter{} \\ 
& \qw\qwbundle{} & \gate{C_1} &  \ \ldots\ \qw & \gate{C_k} & \meter{} & & & &
\end{quantikz}
\vspace{-10pt}
\end{center}
\caption{Constant-depth classically-controlled Clifford circuit for the two round protocol for gates $C_1, \ldots, C_k$ in the first round, and gates $D_1, \ldots, D_\ell$ in the second round.}
\label{fig:2_round_circuit}
\end{figure}
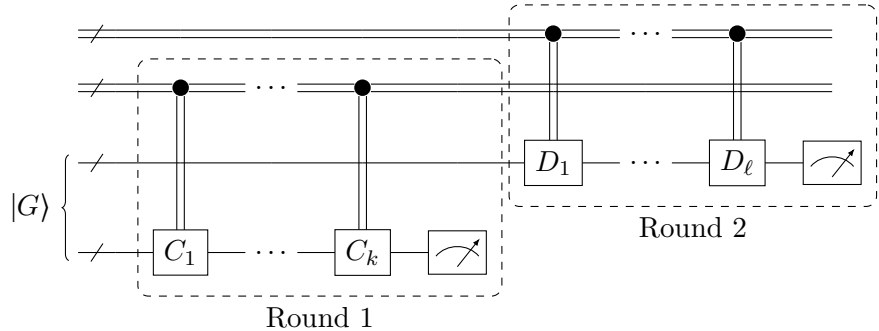

We model the interactive protocol with a quantum computer by an oracle $\oracle$.  The oracle takes the set of gates for the first round and outputs a consistent measurement result, and then the oracle takes a set of gates for the second round and outputs measurement results consistent with both the first round measurements and the given set of gates.  Once the oracle returns an answer for the second phase, you cannot ask it a different question for the second round (justified by the fact that quantum measurements are collapsing).

Now consider a classical circuit which passes the two-round measurement protocol.  After the second round of the protocol, we can \emph{rewind} the classical machine back to the beginning of the second round.  We model this interaction with a rewind oracle $\coracle$.  There are two ways to query the oracle:
\begin{enumerate}[itemsep = 0pt]
\item Given a set of gates for the first round, the oracle returns consistent measurement results for the first round, or
\item Given a set of gates for the second round and a set of gates/measurement results in the first round, the oracle returns a consistent measurement result for the second round. \emph{Such oracle queries are only guaranteed to be correct when the oracle has previously returned those first-round measurements on that particular set of first-round gates.}
\end{enumerate}

Under almost any kind of classical circuit or classical Turing machine, the rewind oracle can be implemented by the same computational device as the original oracle. 
\begin{proposition}
\label{prop:oracleconversion}
Suppose $\oracle$ is an oracle for a $2$-round interactive task. If $\oracle$ is implemented by any of the following circuits or Turing machines $\AC^{0}, \AC^{0}[p], \TC^{0}, \NC^{1}, \cL, $ or $\parityL$, then the rewind oracle, $\coracle$, is implemented by the same class of circuits or Turing machines. 
\end{proposition}

\begin{problem}[2-Round Clifford Simulation (\CliffSim)]
\label{prob:generic_cliff_sim}
Let \emph{2-Round Clifford Simulation} be the task of passing the interactive protocol above in the special case of 2 rounds. Later, we will further specialize this problem to be more specific about 
\begin{itemize}[itemsep = 0pt]
\item the geometry of the starting graph state,
\item the precise encoding of the challenges, and
\item the acceptable rate of error (if any).
\end{itemize}
See \Cref{prob:nc1} and \Cref{prob:parityL}.
\end{problem}


\section{\texorpdfstring{$\NC^1$}{NC1}-hardness}
\label{sec:nc1}

Recall the Clifford simulation problem (\Cref{prob:generic_cliff_sim}). We specialize it to a problem on a width-$2$ grid graph state $\ket{\mathcal H_n}$ which has the following brickwork pattern:
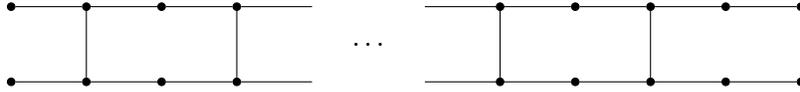
\begin{center}
	\begin{tikzpicture}
	\begin{pgfonlayer}{nodelayer}
		\node [style=vertex] (0) at (0, 0) {};
		\node [style=vertex] (1) at (0, -1) {};
		\node [style=vertex] (2) at (1, 0) {};
		\node [style=vertex] (3) at (1, -1) {};
		\node [style=vertex] (4) at (2, 0) {};
		\node [style=vertex] (5) at (2, -1) {};
		\node [style=vertex] (8) at (3, 0) {};
		\node [style=vertex] (9) at (3, -1) {};
		\node [style=none] (12) at (4.75, -0.5) {$\ldots$};
		\node [style=vertex] (13) at (8.5, 0) {};
		\node [style=vertex] (14) at (8.5, -1) {};
		\node [style=vertex] (15) at (9.5, 0) {};
		\node [style=vertex] (16) at (9.5, -1) {};
		\node [style=vertex] (20) at (10.5, 0) {};
		\node [style=vertex] (21) at (10.5, -1) {};
		\node [style=vertex] (22) at (6.5, 0) {};
		\node [style=vertex] (23) at (6.5, -1) {};
		\node [style=vertex] (24) at (7.5, 0) {};
		\node [style=vertex] (25) at (7.5, -1) {};
		\node [style=none] (26) at (4, 0) {};
		\node [style=none] (27) at (4, -1) {};
		\node [style=none] (28) at (5.5, 0) {};
		\node [style=none] (29) at (5.5, -1) {};
	\end{pgfonlayer}
	\begin{pgfonlayer}{edgelayer}
		\draw (0) to (2);
		\draw (2) to (4);
		\draw (2) to (3);
		\draw (3) to (1);
		\draw (3) to (5);
		\draw (4) to (8);
		\draw (5) to (9);
		\draw (8) to (9);
		\draw (15) to (13);
		\draw (13) to (14);
		\draw (14) to (16);
		\draw (15) to (20);
		\draw (16) to (21);
		\draw (22) to (24);
		\draw (22) to (23);
		\draw (23) to (25);
		\draw (24) to (13);
		\draw (25) to (14);
		\draw (8) to (26.center);
		\draw (9) to (27.center);
		\draw (28.center) to (22);
		\draw (29.center) to (23);
	\end{pgfonlayer}
\end{tikzpicture}
\end{center}
The details of the state $\ket{\mathcal H_n}$ and why it suffices for measurement-based computation are given in \Cref{sec:two_qubit_mbqc}.\footnote{For completeness, we note that a simple width-2 grid would also suffice in place of the graph $\mathcal H_n$.  However, the grid requires 18 qubits per 2-qubit gate, whereas the construction we give only requires 16.}

\begin{problem}[Narrow Cluster Clifford Simulation]
\label{prob:nc1}
Let $A \in \{0,1\}^{2 \times (8 n + 1)}$ and $B \in \{0,1\}^{2 \times 9}$ be binary matrices.  Let \emph{Narrow Cluster Clifford Simulation} be the problem of passing the $\CliffSim$ protocol with initial state $\ket{\mathcal H_n}$ and the following two rounds of challenges:
\begin{itemize} [itemsep = 0pt]
\item Round 1 Challenges:  prover measures qubit $(i,j)$ in $X$-basis if $A_{i,j}=0$; otherwise, prover measures qubit $(i,j)$ in $Y$-basis.
\item Round 2 Challenges:  prover measures qubit $(i , j+ 8n + 1)$ in $X$-basis if $B_{i,j}=0$; otherwise, prover measures qubit $(i,j + 8 n + 1)$ in $Y$-basis.
\end{itemize}
\end{problem}

Our main result for this theorem will be as follows.
\begin{theorem}
\label{thm:nc1_main_result}
Let $\coracle$ be the rewind oracle for the $2$-round Clifford simulation described above (\Cref{prob:nc1}). Then $$\NC^{1} \subseteq (\BPAC^{0})^{\coracle}.$$  
\end{theorem}
As a consequence, if there is, e.g., a $\TC^{0}$ implementation of $\oracle$, then by \Cref{prop:oracleconversion} there is also a $\TC^{0}$ circuit for $\coracle$, implying $\NC^{1} \subseteq (\BPAC^{0})^{\TC^{0}} = \BPTC^{0}$. Now give both sides polynomial advice, i.e., $\NC^{1} / \poly \subseteq \BPTC^{0} / \poly$, and recall a result of Ajtai and Ben-Or \cite{AB84} that $\BPTC^{0} / \poly = \TC^{0} / \poly$. This gives $\NC^{1} / \poly \subseteq \TC^{0} / \poly$, contradicting a standard complexity conjecture that the two (non-uniform) circuit classes are distinct. We take this as evidence that there is no $\TC^{0}$ implementation of $\oracle$ for this task.

The high level outline of the proof is as follows. First, we establish that it is $\NC^{1}$-complete to compute the product of a sequence of $2$-qubit Clifford gates, even up to unknown Pauli corrections. \Cref{prob:nc1} is precisely the task of simulating (via MBQC) a sequence of $2$-qubit Clifford gates (modulo Paulis) on the state $\ket{++}$. In round two, we can use rewinding to measure the state in several bases to perform a kind of tomography. In particular, through the use of non-contextuality, we must learn at least one Pauli string which is not a stabilizer of the state. Finally, by randomizing the reduction we can obfuscate the state and learn a \emph{random} non-stabilizer Pauli string each time, and repetition allows us to learn the state.

\subsection{Hardness and 2-Qubit Clifford Gates}

Recall that the Clifford gates form a discrete group under composition, and the two-qubit Clifford gates are a finite subgroup. Computing the product of a sequence of gates is therefore a special case of the group product problem considered by Barrington and Th\'{e}rien \cite{BarringtonTherien}, so we will use their results to prove hardness as soon as we identify $\Clifford_2$, the group of two-qubit Clifford gates. 

It turns out that there are exactly $11520$ two-qubit Clifford gates (i.e., $|\Clifford_2| = 11520$) in $\Clifford_2 \subseteq \mathrm{SU}(2)$. Among these are the $16$ (again, up to phase) elements of the Pauli group, $\Paulis_2$. As discussed in \Cref{sec:background}, we can only implement a sequence of these Clifford operations modulo the Pauli group, so we are actually interested in the group $\Clifford_2 / \Paulis_2$ of order $720$. Recall from \Cref{sec:background} that the Clifford operations modulo Paulis can be represented by $4 \times 4$ symplectic matrices over $\F_2$. It is known \cite{symmetricInSymplectic} that the symmetric group $S_{4k+2}$ is contained in the symplectic group of $4k \times 4k$ matrices over $\F_2$, and by counting we see that it must be the whole group when $k = 1$, so we have the isomorphism $\Clifford_2 / \Paulis_2 \cong \mathrm{S}_6$. However, to keep this paper as self-contained as possible, we give a proof below, with an explicit description of the isomorphism. 
\begin{lemma}
The group of two-qubit Clifford gates, up to Pauli corrections, is isomorphic to $\mathrm{S}_6$. That is, $\Clifford_2 / \Paulis_2 \cong \mathrm{S}_6$. \Cref{fig:s6_clifford_iso} shows how a permutation of six vertices induces a permutation of the edge labels (in $\Paulis_2 / \mathcal{Z}_2$) and thus specifies an element of $\Clifford_2 / \Paulis_2$.  
\end{lemma}
\begin{proof}
Recall from \Cref{sec:background} that Clifford gates act on the Pauli group by conjugation $U \bullet P = U P U^{\dag}$ for any $U \in \Clifford_m$ and $P \in \Paulis_m$. Furthermore, each Clifford operation induces a permutation of $\Paulis_m / \mathcal{Z}_m$, namely $\tilde{\phi}_{U} \colon  \Paulis_m / \mathcal{Z}_m \to \Paulis_m / \mathcal{Z}_m$, where
$$
\tilde{\phi}_{U}(P \mathcal{Z}_m) = (U \bullet P) \mathcal{Z}_m.
$$
Moreover, recall that the kernel of this homomorphism is $\Paulis_m$, so there is an injective homomorphism from $\Clifford_m / \Paulis_m$ into the symmetric group on $\Paulis_m / \mathcal{Z}_m$. Unfortunately, this symmetric group is much too big (even for the special case $m=2$) since $|\Paulis_2 / \mathcal{Z}_2| = 16$.

Clearly not all permutations of $16$ elements are in the image of $\tilde{\phi}$. For instance, every unitary commutes with $\pfont{I}^{\otimes m}$ so $\pfont{I}^{\otimes m} \mathcal{Z}_m$ is fixed by all permutations. Another constraint is that two Pauli operations $P, Q \in \Paulis_m$ either commute or anti-commute, i.e., $[P,Q] = \pm \pfont{I}^{\otimes m}$, and this is preserved by conjugation:
$$
[U \bullet P, U \bullet Q] = U \bullet [P,Q] = [P,Q].
$$
Also note that the commutator does not depend on the sign of $P$ or $Q$, i.e., $[P, Q] = [\alpha P, \beta Q]$ for $\alpha, \beta \in \{\pm 1, \pm i\}$. These commutation constraints drastically limit the set of possible permutations of $\Paulis_m / \mathcal{Z}_m$. To understand the permutation we define the following six sets of pairwise anti-commuting Pauli operations. 
\begin{align*}
M_1 &= \{ \pfont{XI}, \pfont{YI}, \pfont{ZX}, \pfont{ZY}, \pfont{ZZ} \}, & M_2 &= \{ \pfont{XI}, \pfont{ZI}, \pfont{YX}, \pfont{YY}, \pfont{YZ} \}, & M_3 &= \{ \pfont{YI}, \pfont{ZI}, \pfont{XX}, \pfont{XY}, \pfont{XZ} \}, \\
M_4 &= \{ \pfont{IX}, \pfont{IY}, \pfont{XZ}, \pfont{YZ}, \pfont{ZZ} \}, & M_5 &= \{ \pfont{IX}, \pfont{IZ}, \pfont{XY}, \pfont{YY}, \pfont{ZY} \}, & M_6 &= \{ \pfont{IY}, \pfont{IZ}, \pfont{XX}, \pfont{YX}, \pfont{ZX} \}.
\end{align*}
In \Cref{fig:s6_clifford_iso}, each vertex corresponds to an $M_i$.  The edges incident to that vertex are labeled with the Pauli strings of $M_i$. Not only are the elements in each set pairwise anti-commuting, but they are the \emph{only} pairwise anti-commuting subsets of size five. Since any Clifford operation $U$ preserves commutation relations, $\tilde{\phi}_U$ must permute these six sets/vertices. 

In this way, each element of $\Clifford_2 / \Paulis_2$ maps to a permutation of $\Paulis_2 / \mathcal{Z}_2$, which has an associated permutation in $\mathrm{S}_6$. This map injective because of \Cref{lem:pauli_mod} and the fact that each Pauli string occurs on exactly one edge (so we can recover the permutation of $\Paulis_2 / \mathcal{Z}_2$ from a permutation of vertices). Since $|\Clifford_2 / \Paulis_2| = 720 = |\mathrm{S}_6|$, the homomorphism is an isomorphism.
\end{proof}
\begin{figure}

\begin{center}
	\begin{tikzpicture}[scale=0.3]
	\coordinate[draw] (C) at (10,0);
	\coordinate[] (D) at (5,8.66);
	\coordinate[] (E) at (-5,8.66);
	\coordinate[] (F) at (-10,0);
	\coordinate[] (A) at (-5,-8.66);
	\coordinate[] (B) at (5,-8.66);
	
	\fill[black]  (A) circle [radius=4pt]; 
	\fill[black]  (B) circle [radius=4pt]; 
	\fill[black]  (C) circle [radius=4pt]; 
	\fill[black]  (D) circle [radius=4pt]; 
	\fill[black]  (E) circle [radius=4pt]; 
	\fill[black]  (F) circle [radius=4pt]; 
	
	\draw (A)--(B)--(C)--(D)--(E)--(F)--(A);
	
	\draw (A) -- node[sloped,above right] {$\pfont{ZI}$} (C);
	\draw (C) -- node[sloped,above right] {$\pfont{XI}$} (E);
	\draw (E) -- node[sloped,above right] {$\pfont{YI}$} (A);
	
	\draw (B) -- node[sloped,above left] {$\pfont{IX}$} (F);
	\draw (D) -- node[sloped,above left] {$\pfont{IZ}$} (F);
	\draw (B) -- node[sloped,above left] {$\pfont{IY}$} (D);
	
	\draw (A) -- node[sloped,above=8pt] {$\pfont{XX}$} (D);
	\draw (B) -- node[sloped,above=8pt] {$\pfont{ZZ}$} (E);
	\draw (C) -- node[sloped,below=8pt] {$\pfont{YY}$} (F);
	
	\draw (A) -- node[sloped,below] {$\pfont{XZ}$} (B);
	\draw (B) -- node[sloped,below] {$\pfont{YZ}$} (C);
	\draw (C) -- node[sloped,above] {$\pfont{YX}$} (D);
	\draw (D) -- node[sloped,above] {$\pfont{ZX}$} (E);
	\draw (E) -- node[sloped,above] {$\pfont{ZY}$} (F);
	\draw (F) -- node[sloped,below] {$\pfont{XY}$} (A);
	\end{tikzpicture}
\end{center}
	\caption{$\mathrm{S}_6 \cong \mathcal C_2 / \Paulis_2$ isomorphism}
	\label{fig:s6_clifford_iso}
\end{figure}
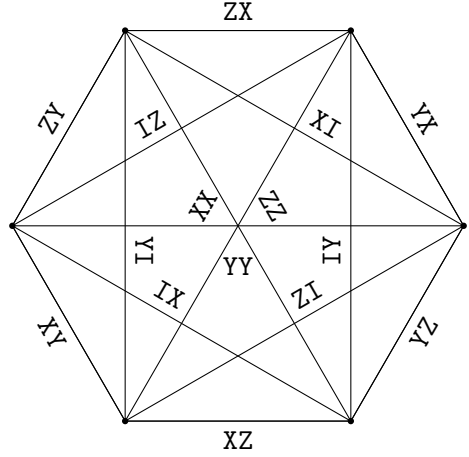

Since $\Clifford_2 / \Paulis_2$ is isomorphic to $\mathrm{S}_6$, there is a normal subgroup within it corresponding to $A_6$. We call the elements of this subgroup \emph{even} since they are isomorphic to \emph{even} permutations within $\mathrm{S}_6$. 

It follows immediately that it is hard to compute the product of many two-qubit Clifford gates.
\begin{corollary}
\label{cor:nc1_hardness}
Given $U_1, \ldots, U_n \in \Clifford_2 / \Paulis_2$, the problem of computing the state $U_1 \cdots U_n \ket{++}$ (modulo Pauli operations) is $\NC^{1}$-complete. The problem remains hard if $U_1, \ldots, U_n$ are promised to be even and if $U_1 \cdots U_n$ is promised to be either $\pfont{II}$ or $H \otimes H \in \Clifford_2 / \Paulis_2$, so $U_1 \cdots U_n \ket{++}$ is either $\ket{++}$ or $\ket{00}$ (modulo Pauli operations).
\end{corollary}
\begin{proof}
By definition, the even elements in $\Clifford_2 / \Paulis_2$ form a subgroup isomorphic to $A_6$. Since $A_6$ is not solvable, a main result of Barrington and Th\'{e}rien \cite{BarringtonTherien} says that computing products in the group is $\NC^{1}$-complete. That is, it is $\NC^{1}$-complete to compute the product $U_1 \cdots U_n$. In fact, the problem remains hard if the product is promised to be either $\pfont{I} \otimes \pfont{I}$ or $U$, for any choice of $U$. Let us take $U = H \otimes H \in \Clifford_2 / \Paulis_2$. Note that we can express $H \otimes H$ as 
$$
H \otimes H = (H \otimes \pfont{I}) \circ \SWAP \circ (H \otimes \pfont{I}) \circ \SWAP,
$$
using two $\SWAP$ gates and two $H \otimes \pfont{I}$ gates, so it must be an even gate. 

Since the product of the elements is either $\pfont{I} \otimes \pfont{I}$ or $H \otimes H$, the state will be either $\ket{++}$ or $\ket{00}$, and it follows that distinguishing these states is $\NC^{1}$-complete. 
\end{proof}

\subsection{Tomography and the Magic Square Game}
\label{sec:magicsquaregame}

Recall that in the context of our task, \Cref{prob:nc1}, the first round essentially computes the state $U_1 \cdots U_n \ket{++}$ with Pauli corrections, in the sense that the state is in the last two unmeasured qubits.\footnote{Recall that the Pauli corrections can be computed from the measurement outcomes, but that computation is already $\NC^{1}$-complete (exercise to the reader), so we will not be able to compute the corrections as part of the $\NC^{1}$-hardness reduction.} The obvious next step is to measure the two-qubit state repeatedly (i.e., rewinding and applying a different Clifford circuit each time) in different bases, and infer the state from the measurements. Under a slightly different definition of \Cref{prob:nc1}, this would be a straightforward application of quantum tomography (in our case, Clifford state tomography \cite{montanaro:2007}). Unfortunately, traditional tomography depends on getting random samples, but a classical algorithm solving \Cref{prob:nc1} is allowed to answer challenges with \emph{any} outcome that occurs with non-zero probability. This makes it impossible to infer even one stabilizer of the state. 

For example, suppose we are promised that the state is either $\ket{00}$ or $\ket{++}$. They have disjoint stabilizers ($\pfont{ZI}$, $\pfont{IZ}$, $\pfont{ZZ}$ vs. $\pfont{XI}$, $\pfont{IX}$, $\pfont{XX}$), so learning even one would determine the entire state. However, the states are not orthogonal, so no measurement can perfectly distinguish the states. For any measurement we make, there must be some outcome that could be observed for both $\ket{00}$ and $\ket{++}$, and an adversarial classical algorithm may choose to answer with that outcome for every measurement. 

Nonetheless, it is possible to learn \emph{something} about the state from measurements. We claim that it is not possible to answer all $2$-qubit Clifford measurements in a way which is consistent with all $2$-qubit stabilizer states. In fact, we have something stronger from the theory of quantum contextuality: there is a collection of $2$-qubit Pauli measurements such that it is impossible to give a consistent answer to all of them that does not depend on \emph{context}, i.e., which other commuting Pauli measurements we apply in the same measurement set. We can use this to force the classical simulator to answer inconsistently on some Pauli measurement, and thus learn that the Pauli string is not a stabilizer. 

Before we go further, let us formalize what kind of measurements we perform. For any stabilizer state $\ket{\psi}$ (which are the only states we consider in this paper), and any $P \in \Paulis_m / \mathcal{Z}_m$, the Pauli measurement associated to $P$ projects onto the $+1$ or $-1$ eigenspace of the canonical (i.e., positive) Pauli operation in $P$, and gives a corresponding $+1$ or $-1$ outcome. As previously discussed, we call $P$ a stabilizer of $\ket{\psi}$ if the outcome of this measurement is deterministic. Given pairwise commuting measurements $P_1, \ldots, P_n \in \Paulis_m / \mathcal{Z}_m$, the order of measurement does not affect the outcomes. Furthermore, the outcome of measuring any product (which commutes with all the individual measurements) will be the product of the outcomes for the individual measurements. 

To measure Paulis $P$ and $Q$ we apply a Clifford operation which maps $P \mapsto \pfont{ZI}$, $Q \mapsto \pfont{IZ}$ and then measure both qubits in the $\pfont{Z}$ basis. Although we cannot directly measure Pauli $PQ$, we can infer its measurement outcome by multiplying the two measurement outcomes from $P$ and $Q$.  We will call any set of three pairwise-commuting Pauli strings $\{ P, Q, R \}$ multiplying to $\pm \pfont{II}$ a \emph{Pauli line}.  By the above, a two-qubit Clifford measurement gives us outcomes for $P$, $Q$, and $R$ on some Pauli line $\{ P, Q, R \}$. 

The geometry of Pauli lines on two-qubits is important in what follows, so let us list a few facts. 
\begin{fact}
\label{fact:pauli_lines}
There are 15 Pauli lines, and they correspond to the perfect matchings in \Cref{fig:s6_clifford_iso}. I.e., the three Pauli strings in the line correspond to the labels of the three edges in the matching. Symmetrically, there are 15 Pauli strings and each one is contained in 3 Pauli lines. It is possible to partition the Pauli strings across $5$ non-intersecting Pauli lines, e.g., as 
$$\{ \pfont{XI}, \pfont{IX}, \pfont{XX} \}, \{ \pfont{YI}, \pfont{IY}, \pfont{YY} \}, \{ \pfont{ZI}, \pfont{IZ}, \pfont{ZZ} \}, \{ \pfont{XY}, \pfont{YZ}, \pfont{ZX} \}, \{ \pfont{YX}, \pfont{ZY}, \pfont{XZ} \}.$$
\end{fact}
A particularly nice structure of Pauli lines and Pauli strings is the \emph{magic square game}, given below. 

\begin{definition}
The \emph{magic square game}, independently discovered by Mermin \cite{mermin_magic_square} and Peres \cite{peres_magic_square} defines a $3 \times 3$ grid (see \Cref{fig:magic_square}) of Pauli measurements on two qubits. Each row and column is a Pauli line.
\begin{figure}
\begin{center}
\begin{tabular}{ccc}
$\pfont{XX}$ & $\pfont{YY}$ & $\pfont{ZZ}$ \\
$\pfont{YZ}$ & $\pfont{ZX}$ & $\pfont{XY}$ \\
$\pfont{ZY}$ & $\pfont{XZ}$ & $\pfont{YX}$
\vspace{-10pt}
\end{tabular}
\end{center}
\caption{The magic square game}
\label{fig:magic_square}
\end{figure}
Moreover, the product of each column is $+\pfont{II}$ and the product of each row is $-\pfont{II}$, so the measurement outcomes multiply to $1$ for a column or $-1$ for a row.
\end{definition}

The reader may be familiar with a different grid under the name ``magic square''. There are many choices for the nine Pauli measurements which work (e.g., simply apply any Clifford operation to all Pauli strings, and the new ones are guaranteed to satisfy the same commutation relations) and most references, including Mermin \cite{mermin_magic_square} and Peres \cite{peres_magic_square}, opt for a different magic square. The reason we choose this one is that there is a particularly nice form for row or column measurements. 
\begin{lemma}
\label{lem:magic_single_only}
The first row of the magic square above (\Cref{fig:magic_square}) corresponds to measurement of the qubits in the Bell basis. Measuring any other row or column is done by applying a single qubit Clifford gate to one of the qubits and then measuring in the Bell basis.  
\end{lemma}

We have already alluded to how we will use the magic square, but let us state it formally:
\begin{theorem}
\label{thm:nonstab_from_magic}
There is a procedure to make six measurements (on six copies) of an unknown two-qubit quantum state $\ket{\psi}$ and learn, with certainty, some Pauli string which does not stabilize $\ket{\psi}$. 
\end{theorem}
\begin{proof}
As one might guess, we measure each copy of the state with a different row or column of the magic square. That is, apply appropriate single qubit gates for the row/column (as in \Cref{lem:magic_single_only}), then measure in the Bell basis. 

Each Pauli string is measured exactly twice: once in a row and once in a column. We can construct two tables of outcomes, one from the row measurements and one from the column measurements. In the column table, the product of all elements is $+1$ (since each column multiplies to $1$) and in the row table, the product of all elements is $-1$. We conclude that the tables must be different, so there exists some Pauli string in the magic square for which we have contradictory measurement outcomes. This Pauli string does not stabilize $\ket{\psi}$, otherwise it would have to measure consistently. 
\end{proof}

It is clear that the above procedure (in \Cref{thm:nonstab_from_magic}) can only output Pauli strings that appear in the magic square. Even if we learn whether each Pauli string of the magic square stabilizes $\ket{\psi}$, it may not be enough to learn $\ket{\psi}$ itself. As discussed, there are several magic squares with the properties we need (e.g., conjugate each Pauli by any $U \in \Clifford_2$), so we can ask about a random magic square. However, we run into the same problem as before: the prover's answers may be adversarial, and it is possible to answer in such a way that we cannot determine $\ket{\psi}$.  Consider the following fact:
\begin{fact}
\label{fact:magic_square_intersect}
Every magic square grid intersects every Pauli line, i.e., there exists some Pauli string in both the line and the grid.
\end{fact}
No matter what magic square we measure, the classical algorithm can arrange for $\pfont{YI}$, $\pfont{IY}$, or $\pfont{YY}$ to be the only Pauli strings with inconsistent answers, and since all of these are non-stabilizers of both $\ket{00}$ and $\ket{++}$ (among other stabilizer states), we cannot deduce $\ket{\psi}$. 

\subsection{Randomization and Self-Reduction}

We have argued that by repeated measurement in round two of \Cref{prob:nc1}, we can force the prover to reveal a non-stabilizer of the two-qubit state $\ket{\psi}$ being measured. Of course, this is not enough to determine the state.\footnote{Even if we were to make all possible Pauli measurements, an adversarial prover need only reveal three non-stabilizers of the state.  By \Cref{fact:magic_square_intersect}, the prover can return non-stabilizers $\pfont{YI}$, $\pfont{IY}$, or $\pfont{YY}$.  However, this information still does not distinguish $\ket{00}$ and $\ket{++}$.} The only way to make progress is to start over in round one with a different instance of the task, constructing a new state $\ket{\psi'}$ which is \emph{related} to $\ket{\psi}$, so that we can carry over what we learn about $\ket{\psi'}$ to $\ket{\psi}$. 

Let us first consider a na\"{i}ve approach on input $U_1, \ldots, U_n$.  By the standard approach, the prover returns a non-stabilizer $P$ of $U_1 \cdots U_n \ket{++}$.  To generate a different non-stabilizer $Q$, we could give the prover the input $U_0, U_1, \ldots, U_n$ for random $U_0 \in \Clifford_2 / \Paulis_2$.  That is, $U_0^{-1} \bullet Q$ is a non-stabilizer $U_1 \cdots U_n \ket{++}$.  If $U_0^{-1} \bullet Q \neq P$, we have learned new information.  Unfortunately, an adversarial prover may anticipate this approach and return $Q = U_0 \bullet P$, so we learn no new information. To force the prover to reveal something new, we need a more thorough randomization procedure, which begins with an idea of Kilian \cite{Kilian}.

\begin{theorem}[Kilian Randomization]
Let $G$ be a group. Given $g_1, \ldots, g_n \in G$, there is a procedure to sample uniformly random $g_1', \ldots, g_n' \in G$  subject to $g_1 \cdots g_n = g_1' \cdots g_n'$.
\end{theorem} 

We will need to extend Kilian's idea slightly below.
\begin{corollary}
\label{cor:kilian}
Let $G$ be a group and let $H \trianglelefteq G$ be a normal subgroup. Given $g_1, \ldots, g_n \in G$, there is a procedure to sample uniformly random $g_1', \ldots, g_n' \in G$ subject to $g_1 \cdots g_n = g_1' \cdots g_n'$ and $g_iH = g_i'H$ for all $i$ as follows:

\begin{algorithmic}
	\Function{Kilian$_{H}$}{$g_1, \ldots, g_n$}
	\State{$h_1, \ldots, h_{n-1} \sim \textrm{Unif}(H)$} 
	\State{$g_1' \gets g_1 h_1$}
	\State{\text{$g_i' \gets h_{i-1}^{-1} g_i h_i$ for $i = 2, \ldots, n-1$}}
	\State{$g_n' \gets h_{n-1}^{-1} g_n$}
	\State{\Return $(g_1', \ldots, g_n')$}
	\EndFunction
\end{algorithmic}

Note that Kilian randomization is the special case $H = G$. 
\end{corollary}
\begin{proof}
It is clear that cancellation gives $g_1 \cdots g_n = g_1' \cdots g_n'$. It is also easy to see that $g_i'H = g_i H$ using the fact that 
$$
g_i' = h_{i-1}^{-1} g_{i} h_i = g_i (g_i^{-1} h_{i-1}^{-1} g_i) h_i
$$
and $g_i^{-1} h_{i-1}^{-1} g_i \in g_i^{-1} H g_i = H$. 

It remains to show that $g_1', \ldots, g_n'$ are uniformly random subject to the constraints. By definition, $h_1$ is a uniformly random element of $H$, so $g_1'$ is a uniformly random element of the coset $g_1 H$. For $i = 2, \ldots, n-1$, we see that $g_i'$ is uniformly random conditioned on $g_1', \ldots, g_{i-1}'$ since $h_i$ is uniformly random and independent of $g_1', \ldots, g_{i-1}'$. Finally, given $g_1', \ldots, g_{n-1}'$, there is a unique choice of $g_n'$ satisfying the constraints, namely $g_n' := (g_{1}' \cdots g_{n-1}')^{-1} g_1 \cdots g_n$. 
\end{proof}

We can use this technique to randomly self-reduce algorithms that take a list of group elements. After applying this randomization step, the na\"ive idea from earlier (i.e., multiplying by a random element and conjugating the result by the inverse) actually works. We will state the theorem in the abstract (with a finite group $(G, \cdot)$ acting on a set $S$ as $\bullet \colon G \times S \to S$ where $(gh) \bullet x = g \bullet (h \bullet x)$), but it is a complicated theorem and it may help to keep an example in mind. For this section, the relevant setting of parameters\footnote{In \Cref{sec:parityL} we use \Cref{thm:symmetrize} in full generality to prove the $\parityL$ result.} is $G = H = F = \Clifford_2 / \Paulis_2$, acting on the set of Pauli strings $S = \Paulis_2 / \mathcal{Z}_2$ by conjugation, $U \bullet P := U P U^{\dag}$. 

\begin{theorem}
	\label{thm:symmetrize}
	Let $A \colon G^{*} \to S$ be a randomized algorithm which takes lists of group elements as input and outputs an element of some set $S$. Suppose $G$ acts on $S$ by $\bullet \colon G \times S \to S$. Let $G$ have subgroups $F, H$ such that $F \leq H \normalin G$. Consider the following randomized algorithm:
	\begin{algorithmic}
		\Function{$B$}{$g_1, \cdots, g_n$}
		\State{$f \sim \textrm{Unif}(F)$}
		\State{$g_1', \ldots, g_n' \gets \textsc{Kilian}_{H}(f g_1, g_2, \ldots, g_n)$}
		\State{\Return $f^{-1} \bullet A(g_1', \ldots, g_n')$.}
		\EndFunction
	\end{algorithmic}
	
	Then the output distribution of $B(g_1, \ldots, g_n)$ is $(g_1 \cdots g_n) \bullet \mathcal{D}(F g_1 \cdots g_n, g_1 H, \ldots, g_n H)$ where $\mathcal{D}(F g_1 \cdots g_n, g_1 H, \ldots, g_n H)$ is the average of $(g_1' \cdots g_n')^{-1} \bullet A(g_1', \ldots, g_n')$ over all $g_1', \ldots, g_n'$ such that $g_1' \cdots g_n' \in Fg_1 \cdots g_n$ and $g_i H = g_i' H$ for all $i$.
\end{theorem}
\begin{proof}
	Kilian randomization (in the form of \Cref{cor:kilian}) samples a uniformly random $g_1', \ldots, g_n'$ subject to $fg_1 g_2 \cdots g_n = g_1' \cdots g_n'$ and $f g_1 H = g_1 H = g_1' H$, $g_2 H = g_2' H$, \ldots, $g_n H = g_n' H$. Since $f$ is uniformly random in $F$, $f g_1 \cdots g_n$ is a uniformly random element of $F g_1 \cdots g_n$, so we are actually sampling a uniformly random $g_1', \ldots, g_n'$ such that $g_1' \cdots g_n' \in Fg_1 \cdots g_n$ and $g_i' H = g_i H$ for all $i$.	
	Instead of returning the result of $A(g_1', \ldots, g_n')$ directly, Algorithm $B$ outputs, 
	\begin{align*}
	f^{-1} \bullet A(g_1', \ldots, g_n') &= (f^{-1} g_1' \cdots g_n') \bullet (g_1' \cdots g_n')^{-1} \bullet A(g_1', \ldots, g_n') \\
	&= (g_1 \cdots g_n) \bullet (g_1' \cdots g_n')^{-1} \bullet A(g_1', \ldots, g_n') \\
	&= (g_1 \cdots g_n) \bullet P
	\end{align*}
	where $P$ is a sample from $(g_1' \cdots g_n')^{-1} \bullet A(g_1', \ldots, g_n')$ where $g_1', \ldots, g_n'$ are uniformly random subject to $g_i H = g_i'H$ and $g_1' \cdots g_n' \in Fg_1 \cdots g_n$. That is, $P$ is drawn from distribution $\mathcal{D}(F g_1 \cdots g_n, g_1 H, \ldots, g_n H)$. 
\end{proof}

Let us interpret the result in the concrete setting we need for this section: $G = H = F = \Clifford_2 / \Paulis_2$ and $S = \Paulis_2 / \mathcal{Z}_2$ under conjugation, $U \bullet P := U P U^{\dag}$. We imagine Algorithm $A$ as adversarially outputting $P \in \Paulis_2 / \mathcal{Z}_2$ which does not stabilize $\ket{\psi} := U_1 \cdots U_n \ket{++}$. Recall that $P$ stabilizes $U_1 \cdots U_n \ket{++}$ if and only if $(U_1 \cdots U_n)^{-1} \bullet P$ stabilizes $\ket{++}$ because
\begin{align*}
U_1 \cdots U_n \ket{++} &= P U_1 \cdots U_n \ket{++} \\
\ket{++} &= (U_1 \cdots U_n)^{-1} P U_1 \cdots U_n \ket{++} \\
&= ((U_1 \cdots U_n)^{-1} \bullet P) \ket{++}.
\end{align*} 
This fact gives meaning to the expression $(U_1' \cdots U_n')^{-1} \bullet A(U_1', \cdots, U_n')$: it is the non-stabilizer for $\ket{++}$ we get by rolling back the unitaries $U_1' \cdots U_n'$ on whichever non-stabilizer $A$ returns for $U_1' \cdots U_n' \ket{++}$. It makes sense to average these distributions (as in, e.g.,  $\mathcal{D}(F g_1 \cdots g_n, g_1 H, \cdots, g_n H)$), since they are all non-stabilizers of the same state. Since $G = H = F$, all the parameters of this distribution have only one value, $G$. Hence, \Cref{thm:symmetrize} defines only one distribution, $\mathcal{D}_n := \mathcal{D}(F g_1 \cdots g_n, g_1 H, \cdots, g_n H)$, which is an average over all lists of $n$ two-qubit Clifford operations. Hence, in our example, Algorithm $B$ is equivalent to sampling $P$ from $\mathcal{D}_n$ and conjugating it by $U_1 \cdots U_n$ to make it a non-stabilizer of $U_1 \cdots U_n \ket{++}$. 

All that being said, we have not attained our original goal of sampling a \emph{random} non-stabilizer $P \in \Paulis_2 / \mathcal{Z}_2$. If Algorithm $A$ is adversarial, it may be that Algorithm $B$ always returns the same answer, e.g., $(U_1 \cdots U_n) \bullet \pfont{YY}$, because the distribution $\mathcal{D}_n$ may have support on only one element (e.g., $\pfont{YY}$). Information theoretically, this kind of distribution actually gives us \emph{more} information about $U_1 \cdots U_n$ than a random non-stabilizer, but since it is more difficult to analyze, we will instead apply another layer of random self-reduction.  
\begin{theorem}
\label{thm:nc1_self_reduce}
Let $\ket{\psi}$ be an $m$-qubit Clifford state. Suppose Algorithm $B$ takes a list of Clifford operations $U_1, \ldots, U_n \in \Clifford_m / \Paulis_m$ and outputs a Pauli string from a distribution $(U_1 \cdots U_n) \bullet \mathcal{D}_n$ where $\mathcal{D}_n$ is a distribution over Pauli strings depending only on $n$. 

Define Algorithm $C$ as below. 
\begin{algorithmic}
	\Function{$C$}{$U_1, \ldots, U_n$}
	\State{$V \sim \textrm{Unif}(\{ V \in \Clifford_m : V \ket{\psi} = \ket{\psi} \})$}
	\State{\Return $B(U_1, \ldots, U_{n-1}, U_n V)$}
	\EndFunction
\end{algorithmic}
Algorithm $C$ outputs an element of $\Paulis_m / \mathcal{Z}_m$ at random from $(U_1 \cdots U_n) \bullet \mathcal{D}_n'$ where $\mathcal{D}_n'$ is a distribution such that 
$$
\Pr[\text{$P$ stabilizes $\ket{\psi}$} | P \sim \mathcal{D}_n'] = \Pr[\text{$P$ stabilizes $\ket{\psi}$} | P \sim \mathcal{D}_n],
$$
and $\mathcal{D}_n'$ is uniform over the stabilizers of $\ket{\psi}$ and uniform over non-stabilizers of $\ket{\psi}$. 
\end{theorem}
\begin{proof}
By symmetry, it suffices to prove the result for any Clifford state $\ket{\psi}$; let us set $\ket{\psi} = \ket{0^m}$. Any Clifford operation stabilizing $\ket{0^m}$ fixes the Pauli subgroup $\{ \pfont{I}, \pfont{Z} \}^{m}$, and therefore the tableau\footnote{We refer the reader to \Cref{app:clifford} for details on Clifford tableaux.} must be of the form 
$$
\begin{bmatrix}
A & B \\
0 & A^{-T}
\end{bmatrix}
$$
for arbitrary $A, B \in \{ 0, 1 \}^{m \times m}$ subject to the conditions that $A$ is full rank and $BA^{T}$ is symmetric. Given an arbitrary Pauli string $P = X^{x} Z^{z}$ expressed as bits $x, z \in \{ 0, 1 \}^{m}$ for the $X$ and $Z$ components respectively, the tableau above will map it to the Pauli string $V P V^{\dag} = X^{x'} Z^{z'}$ where 
$$
\begin{bmatrix}
x' & z' 
\end{bmatrix}
= 
\begin{bmatrix} 
x & z 
\end{bmatrix}
\begin{bmatrix}
A & B \\
0 & A^{-T} 
\end{bmatrix}
=
\begin{bmatrix}
xA & xB + zA^{-T}.
\end{bmatrix}
$$

Clearly $\ket{0^m}$ is stabilized by strings with only $Z$ component, i.e., where $x = 0$. If $P$ is a stabilizer then it follows that $x = 0$ and thus $x' = x A = 0$. Clearly a random invertible transformation $A$ (or $A^{-T}$) will map a non-zero $z \neq 0$ to a uniformly random non-zero $z'$, so $V P V^{\dag}$ is a uniformly random stabilizer of $\ket{0^m}$. On the other hand, if $P$ is not a stabilizer of $\ket{0^m}$ then $x \neq 0$, and by the same argument as above, $x' = xA$ is a uniformly random non-zero vector. Then $z A^{-T}$ is \emph{not} uniformly random conditioned on $x'$, but we assert that $x B$ is uniformly random. Recall that $S := BA^{T}$ is a uniformly random symmetric matrix, and by rearranging, $B = S A^{-T}$. Note that if $x$ is non-zero, then $xS$ is a uniformly random vector, and therefore $xB = xSA^{-T}$ is uniformly random because invertible linear transformations preserve the uniform distribution. We conclude that $z'$ is a uniformly random vector, and hence $V P V^{\dag}$ is uniformly random non-stabilizer of $\ket{0^m}$. 
\end{proof}

We now have all the pieces for this section's main result:
\begin{reptheorem}{thm:nc1_main_result}
Let $\coracle$ be the rewind oracle for the $2$-round Clifford simulation described above (\Cref{prob:nc1}). Then 
$$\NC^{1} \subseteq (\BPAC^{0})^{\coracle}.$$  
\end{reptheorem}
\begin{proof}
Our goal is to use the rewind oracle to determine the state $U_1 \cdots U_n \ket{++}$, modulo Pauli operations, given unitaries $U_1, \ldots, U_n$.

First, we construct Algorithm~$A$ from the rewind oracle. Algorithm~$A$ applies the oracle to $U_1, \ldots, U_n$ in the first round, then measures magic square rows and columns (all of them, by rewinding) in the second round. From these measurements, \Cref{thm:nonstab_from_magic} says we can identify a non-stabilizer of the state $\ket{\psi}$ being measured. If there is more than one measurement with inconsistent results, choose one arbitrarily to return. By \Cref{thm:two_qubit_mbqc}, the initial state is $\ket{++}$ and then $U_1, \ldots, U_n$ are applied ($U_n$ first) to the state by measuring appropriately, so the final two qubits are in state $\ket{\psi} := U_1 \cdots U_n \ket{++}$, modulo the Pauli group.

Next, we use Kilian randomization and related ideas in \Cref{thm:symmetrize} and \Cref{thm:nc1_self_reduce} to construct Algorithm~$C$, which makes calls to Algorithm~$A$ and returns a uniformly random Pauli string $P$ not stabilizing $\ket{\psi}$. Although it is not stated explicitly, Algorithm~$C$ makes exactly one call to Algorithm~$A$ (which is in turn making constantly many calls to the rewind oracle), and the reduction can be computed in $\AC^{0}$ since it only has to sample from the group, multiply constantly many in the group, and translate the two-qubit Clifford elements to measurements. 

Finally, we run Algorithm $C$ sufficiently many times (i.e., $O(\log n)$) to collect \emph{all} non-stabilizers of $\ket{\psi}$ with high probability, and thereby deduce $\ket{\psi}$.\footnote{For any state, there are exactly 12 non-stabilizers.  So, if we collect all 12 after the $O(\log n)$ queries to Algorithm~$C$, we know the answer to the $\NC^1$ problem with certainty; otherwise, we can report ``do not know.''  This places the reduction in $\ZPAC^0 \subseteq \BPAC^0$.} This solves an $\NC^{1}$ decision problem in \Cref{cor:nc1_hardness}.
\end{proof}

This is the simplest version of the result. The remainder of this section is devoted to small improvements which we have avoided until now because they complicate the statement and/or proof of the result. 

\subsection{Error Tolerance}

Suppose the classical simulation is faulty and outputs incorrect (possibly adversarial) responses for some fraction of interactions. That is, for uniformly random inputs (i.e., unitaries $U_1, \ldots, U_n$ chosen randomly from $\Clifford_2$ in the first round, and a uniformly random Clifford measurement in the second round), the classical simulation fails the task with probability $\epsilon > 0$. It turns out that for sufficiently small $\epsilon$, we can still solve an $\NC^{1}$-hard problem, given access to a rewind oracle. 
Let us start with a fact about the maximum number of compatible Pauli measurements.
\begin{fact}
\label{lem:pauli_brute_force}
Let $\nu \colon \Paulis_2 / \mathcal{Z}_2 \to \{ +1, -1 \}$ be an assignment of measurement outcomes to the canonical (i.e., positive) Pauli string in each coset. For any such assignment, at least $3$ of the $15$ Pauli line constraints are violated. 
\end{fact}
\begin{proof}
This may be verified by brute force enumeration of the $2^{15} = 32768$ different assignments. It is tight because, e.g., assigning \emph{all} Pauli strings to $+1$ only violates three lines: $\{ \pfont{XX}, \pfont{YY}, \pfont{ZZ} \}$, $\{ \pfont{XY}, \pfont{YZ}, \pfont{ZX} \}$, and $\{ \pfont{XZ}, \pfont{YX}, \pfont{ZY} \}$. 
\end{proof}

\begin{theorem}
\label{thm:nc1_error}
Let $\coracle$ be the rewind oracle for a $2$-round Clifford simulation which performs the task in \Cref{prob:nc1}, and fails on a uniformly random input with probability less than $\epsilon$ for some $\epsilon < \frac{2}{75}$. Then 
$$\NC^{1} \subseteq (\BPAC^{0})^{\coracle}.$$ 
\end{theorem}
\begin{proof}
We construct an Algorithm $A$ which uses the rewind oracle to return a \emph{Pauli line}, rather than an individual Pauli. In the first round, the algorithm effectively creates the state $\ket{\psi} := U_1 \cdots U_n \ket{++}$ (modulo Paulis), then makes all possible Clifford measurements in the second round, using the rewind oracle. Recall that these measurements give three outcomes for each Pauli string, so we may construct $\nu \colon \Paulis_2 / \mathcal{Z}_2 \to \{ +1, -1 \}$ where $\nu(P)$ is the majority of the three measurement outcomes for $P$. By \Cref{lem:pauli_brute_force}, this assignment of outcomes violates at least three Pauli line constraints, e.g., $PQR = +\pfont{II}$ but $\nu(P) \nu(Q) \nu(R) = -1$. Return one violated Pauli line at random. 

Let $\ell$ be the unique Pauli line such that all Pauli strings stabilize $\ket{\psi}$. If there are no errors, then the outcomes for any $P \in \ell$ are all the same, and the Pauli line constraint for $\ell$ is satisfied, so Algorithm $A$ will output some other line. Moreover, it would require at least two measurement errors involving a stabilizer $P \in \ell$ to change the value of $\nu(P)$, so Algorithm $A$ will not output $\ell$ in case of one error.  

When there are two or more errors, we assume the adversary may arbitrarily change $\nu$ so that the constraint for $\ell$ is violated, in addition to at least two other lines. Hence, the probability of outputting $\ell$ is still at most $\frac{1}{3}$.

Randomly self-reduce Algorithm $A$ (by \Cref{thm:symmetrize} and \Cref{thm:nc1_self_reduce}) to produce an Algorithm~$C$. The self-reductions map stabilizers to stabilizers and non-stabilizers to non-stabilizers, so Algorithm $C$ outputs the line $\ell$ for a particular input with the same average probability as Algorithm~$A$.

Finally, we run Algorithm $C$ many times and count how many times each line is reported. There are three classes of Pauli line:
\begin{align*}
C_1 &= \{ \ell \}, \\
C_2 &= \{ \text{lines incident to $\ell$} \} \backslash \{ \ell \}, \\
C_3 &= \{ \text{lines not incident to $\ell$} \}.
\end{align*}
By the random self-reduction, the probability Algorithm $C$ returns a line is uniform within each class, so the distribution of lines it returns is uniform on three subsets of size $|C_1| = 1$, $|C_2| = 6$, and $|C_3| = 8$. With $O(\log n)$ samples, we recover the distribution to accuracy $O(1/\sqrt{\log n})$, and thus also recover the three uniform subsets \emph{unless} elements from two or more classes occur with the same frequency (to within $O(1/\sqrt{n})$). However, even if \emph{two} of the classes are indistinguishable, the third will give us $\ell$, since either the class is $\ell$ itself, or $\ell$ is the unique Pauli line incident to all lines in the class, or $\ell$ is the unique Pauli line \emph{not} incident to all lines in the class. Thus, the only way $O(n)$ samples fail to recover $\ell$ is if Algorithm $C$ outputs Pauli lines uniformly. We argue this is not possible when $\epsilon < \frac{2}{75}$. 

Let $\delta$ be the probability that Algorithm $C$ makes more than $2$ errors among the $15$ second round measurements, for a uniformly random first round input. Immediately we have $\epsilon > \frac{2}{15} \delta$, so if $\epsilon < \frac{2}{75}$ then $\delta < \frac{1}{5}$. Observe that Algorithm $C$ outputs $\ell$ with probability at most $\frac{\delta}{3}$. Since $\frac{\delta}{3}$ is a constant less than $\frac{1}{15}$, the distribution is noticeably different from uniform with $O(n)$ samples. 
\end{proof}

\subsection{\texorpdfstring{$\AC^{0}[6]$}{AC0[6]}-hardness on a Line}
\label{sec:ac0mod6}

We have tried to make \Cref{prob:nc1} as practical as possible---using a grid of qubits with only nearest neighbor interactions, shrinking the gadgets necessary for two-qubit Clifford gates, and making the grid only two qubits wide. This section examines what is possible on a grid of width one, i.e., a \emph{line} of qubits. 

\begin{problem}[Line Clifford Simulation]
	\label{prob:ac0mod6}
	Let $A, B \in \{0,1\}^{3n+1}$ be binary vectors and let $C \in \Clifford_2$. Let \emph{Line Clifford Simulation} be the problem of passing the $\CliffSim$ protocol with the graph state on a line of $6n+4$ vertices, and the following two rounds of challenges:
	\begin{itemize} [itemsep = 0pt]
		\item Round 1 Challenges: prover measures qubit $i$ from the left in the $X$ or $Y$ basis according to $A_{i}$, and measures qubit $i$ from the right in the $X$ or $Y$ basis according to $B_i$.
		\item Round 2 Challenges: prover applies Clifford operation $C$ to the two middle qubits and measures both in the $X$ basis.
	\end{itemize}
\end{problem}

Given this problem, we prove the following hardness result. 
\begin{theorem}
\label{thm:ac0mod6_main}
Let $\coracle$ be the rewind oracle for the $2$-round Clifford simulation described above (\Cref{prob:ac0mod6}). Then $$\AC^{0}[6] \subseteq (\BPAC^{0})^{\coracle}.$$ 
\end{theorem}
\begin{proof}
We already have all the tools we need to prove this. As we showed in \Cref{thm:single_qubit_mbqc}, a line of $3n+2$ qubits can implement a sequence of $n$ single-qubit Clifford operations. We see that \Cref{prob:ac0mod6} simulates two of these lines, with the two output qubits adjacent in the middle. Thus, round 1 puts the separable state 
$$\ket{\psi} := \ket{\psi_1} \otimes \ket{\psi_2} = (U_1 \cdots U_n \ket{+}) \otimes (V_1 \cdots V_n \ket{+}),$$
in the two middle qubits (where $U_1, \ldots, U_n, V_1, \ldots, V_n$ are single-qubit Clifford operations), and in round 2 we perform an arbitrary entangling Clifford operation and measurement. 

Next, we observe that the group of single-qubit Clifford operations is $\mathrm{S}_4$, but after modding out the Pauli operations it is $\mathrm{S}_3$. If we can compute the product of a sequence of $\mathrm{S}_3$ permutations, then we can immediately compute $\MOD_2$ and $\MOD_3$ gates, whence we get $\MOD_6$ gates by the Chinese remainder theorem. This implies it is at least $\NC^{0}[\MOD_6] = \NC^{0}[6]$ hard to compute the state $U_1 \cdots U_n \ket{+}$ modulo Pauli operations.

As before, a classical oracle for the task implies a classical rewind oracle. The problem explicitly allows us to apply an arbitrary Clifford operation and measure in the second round, so we can make arbitrary measurements on the two-qubit state $\ket{\psi}$. In particular, we can measure the magic square and learn a two-qubit Pauli measurement which does not stabilize $\ket{\psi}$. That is, from the classical rewind oracle we extract an algorithm, $A$, which returns a non-stabilizer Pauli in $\boxplus := \{ \pfont{XX}, \pfont{XY}, \pfont{XZ}, \pfont{YX}, \pfont{YY}, \pfont{YZ}, \pfont{ZX}, \pfont{ZY}, \pfont{ZZ} \}$. Since $\ket{\psi}$ is a tensor product state, note that it will be stabilized by $\pfont{I} \otimes Q$, $P \otimes \pfont{I}$, and $P \otimes Q$, for Paulis $P, Q \in \{ \pfont{X}, \pfont{Y}, \pfont{Z} \}$. Exactly one of these stabilizers appears in the magic square, so there are eight non-stabilizers that $A$ may return.  

To learn the actual state of $\ket{\psi}$, we make several \emph{randomized} queries to algorithm $A$. First, we randomize using \Cref{thm:symmetrize} with $G = H = F = \Clifford_1 / \Paulis_1 \cong \mathrm{S}_3$, on both halves of the line (i.e., $U_1, \ldots, U_n$ and $V_1, \ldots, V_n$), giving an algorithm $B$ with output distribution $(U_1 \cdots U_n \otimes V_1 \cdots V_n) \bullet \mathcal{D}$ where $\mathcal{D}$ is a distribution over the Pauli strings in the magic square. We slightly upgrade this to algorithm $B'$, which reverses the entire line of qubits with probability $\frac{1}{2}$ before calling $B$, and then swaps back the Paulis in the answer. This has the effect of making $\mathcal{D}$ symmetric, i.e., the probability of returning $\pfont{XY}$ is the same as the probability of returning $\pfont{YX}$. 

Finally, we construct an algorithm $C$ which applies $\theta_{yz} = \frac{Y+Z}{\sqrt{2}}$ (the unique non-identity Clifford operation (up to Paulis) which fixes $\ket{+}$) at the ends of the line (with probability $\frac{1}{2}$, independently for each end). Since $\theta_{yz}\ket{+} = \ket{+}$, the operation does not affect the state, and thus it does not affect the stabilizer Pauli (e.g., $\pfont{X}$ for $\ket{+}$) but permutes the other two Paulis ($\pfont{Y} \longleftrightarrow \pfont{Z}$). It follows that the distribution $\mathcal{D}$ puts the same weight on $P \otimes \pfont{Y}$ and $P \otimes \pfont{Z}$ for any $P$, and likewise puts the same weight on $\pfont{Y} \otimes P$ and $\pfont{Z} \otimes P$. This partitions the Pauli strings of the magic square into four subsets,
$$
\{ \pfont{XX} \}, \{ \pfont{XY}, \pfont{XZ} \}, \{ \pfont{YX}, \pfont{ZX} \}, \{ \pfont{YY}, \pfont{YZ}, \pfont{ZY}, \pfont{ZZ} \},
$$
and $\mathcal{D}$ must be uniform on each subset. Furthermore, $\pfont{XX}$ should have weight $0$, and recall $\pfont{XY}$ and $\pfont{YX}$ have the same weight by reversal/swap symmetry, so $\mathcal{D}$ divides mass $p$ between $\{ \pfont{XY}, \pfont{XZ}, \pfont{YX}, \pfont{ZX} \}$ and divides mass $1-p$ between $\{ \pfont{YY}, \pfont{YZ}, \pfont{ZY}, \pfont{ZZ} \}$ for some $p \in [0,1]$.

Now we show how to recover the stabilizer (of $\ket{\psi}$ in $\boxplus$) directly from calls to algorithm $C$. We describe the algorithm under the assumption that $U_1 \cdots U_n \otimes V_1 \cdots V_n$ is the identity, to avoid repeating ``conjugated by $U_1 \cdots U_n \otimes V_1 \cdots V_n$'' throughout, but we take care to verify that the operations behave correctly under conjugation. 

First, we make constantly many calls to algorithm $C$ and learn the distribution $(U_1 \cdots U_n \otimes V_1 \cdots V_n) \bullet \mathcal{D}$ to constant precision. We know the distribution partitions into three sets, $\{ \pfont{XX} \}$, $\{ \pfont{XY}, \pfont{XZ}, \pfont{YX}, \pfont{ZX} \}$, and $\{ \pfont{YY}, \pfont{YZ}, \pfont{ZY}, \pfont{ZZ} \}$, where the distribution is approximately uniform, although $\pfont{XX}$ should have no weight as it is the stabilizer. Depending on the relative weight of the other two sets ($1-p$ for the first and $p$ for the second) there are three cases:
\begin{itemize}
	\item If $p$ is bounded away from $0$ and $1$, then we can collect non-stabilizer Pauli strings until we have all but one. This must be $\pfont{XX}$, the stabilizer.
	\item If $p$ is sufficiently close to $1$, then we learn $\{ \pfont{XY}, \pfont{XZ}, \pfont{YX}, \pfont{ZX} \}$ in $O(1)$ calls. Let us take the product $\pfont{XX} = \pfont{XY} \cdot \pfont{XZ} \cdot \pfont{YX} \cdot \pfont{ZX}$ as the stabilizer. 
	\item If $p$ is sufficiently close to $0$, then we learn $\{ \pfont{YY}, \pfont{YZ}, \pfont{ZY}, \pfont{ZZ} \}$. Among these, two pairs $\{ \pfont{YY}, \pfont{ZZ} \}$ and $\{ \pfont{YZ}, \pfont{ZY} \}$ commute (which also distinguishes it from the previous case), and the product of either is $\pm \pfont{XX}$.
\end{itemize}
In all cases, conjugating the distribution and samples leads to conjugation of the answer. 

To conclude, we have constructed a randomized algorithm from $\coracle$ which learns a magic square stabilizer of $\ket{\psi}$. Since $\ket{\psi}$ is a tensor product Clifford state, this tells us $\ket{\psi}$ exactly. We showed that we can encode the answer to an $\NC^{0}[2]$-hard or $\NC^{0}[3]$-hard problem in $\ket{\psi}$, so with a few calls to this randomized algorithm we can implement $\MOD_6$ gates, and thus compute anything in $\NC^{0}[6]$. However, we require random $\AC^{0}$ circuits to sample for the randomization step, so we must compute at least the class $\AC^{0}[6]$. 
\end{proof}


\section{\texorpdfstring{$\parityL$}{Parity-L}-hardness}
\label{sec:parityL}

The purpose of this section is to obtain a $\parityL$ hardness result for the $\CliffSim$ problem.  First, we will necessarily need to modify the $\CliffSim$ protocol from the previous section to make it more computationally challenging.  Since an $\NC^1$ circuit can pass the $\CliffSim$ protocol on any grid of constant width, we will use larger graph states in this section in order to apply more complicated gates.  For now, let us only roughly describe the initial state as some graph state $\ket{\mathcal G_{n,m}}$ where $\mathcal G_{n,m}$ is some subgraph of an $m \times \Theta(mn)$ grid.\footnote{As in the previous section, we can take the state $\ket{\mathcal G_{n,m}}$ to be a full $\Theta(m) \times \Theta(mn)$ grid.  However, this would require us to introduce $Z$-basis measurement challenges into the model, which essentially carve out those connections which we do not want in the grid.} The exact details of $\mathcal G_{n,m}$ will be given in \Cref{sec:initial state}.  The formal statement is as follows:

\begin{problem}[Wide Cluster Clifford Simulation]
\label{prob:parityL}
Let $A \in \{0,1\}^{m \times c_1mn}$ and $B \in \{0,1\}^{m \times c_2}$ be binary matrices.  Let \emph{Wide Cluster Clifford Simulation} be the problem of passing the $\CliffSim$ protocol with initial state $\ket{\mathcal G_{n,m}}$ and the following two rounds of challenges:
\begin{itemize} [itemsep = 0pt]
\item Round 1 Challenges:  prover measures qubit $(i,j)$ in the $X$-basis if $A_{i,j}=0$; otherwise, prover measures qubit $(i,j)$ in the $Y$-basis.
\item Round 2 Challenges:  prover measures qubit $(i , j+c_1mn)$ in the $X$-basis if $B_{i,j}=0$; otherwise, prover measures qubit $(i,j + c_1 mn)$ in the $Y$-basis.
\end{itemize}
\end{problem}

Our main result is $\parityL$-hardness for this problem.

\begin{theorem}
\label{thm:parityL_main_result}
Let $\coracle$ be the rewind oracle for Wide Cluster Clifford Simulation (\Cref{prob:parityL}). Then 
$$\parityL \subseteq (\BPAC^{0})^{\coracle}.$$  
\end{theorem}

Once again, this implies strong complexity theoretic evidence that certain highly-parallel and low-depth classical models of computation cannot solve this Clifford simulation problem.  For example, if there is an $\cL$ machine which computes $\oracle$, then there is also an $\cL$ machine which computes $\coracle$, implying $\parityL \subseteq (\BPAC^{0})^{\cL}$. Since $\AC^{0}$ circuits can be computed in $\cL$, the entire right-hand side is in $\BPL$. Giving both classes polynomial advice and using a result of Bennett and Gill \cite{bg} that $\BPL/\poly = \cL/\poly$, we get $\parityL/\poly \subseteq \cL/\poly$. Assuming $\parityL/\poly \not\subseteq \cL/\poly$ (a plausible complexity conjecture), this implies there is no $\cL$ solution to the problem. 

The proof of $\parityL$-hardness in \Cref{thm:parityL_main_result} follows the same general outline as the proof of $\NC^1$-hardness in \Cref{thm:nc1_main_result}.  We show that outputs from a conditional distribution of a $\QNC^0$ circuit can be combined by a logarithmic-space machine to solve a $\parityL$-hard problem.  The $\QNC^0$ circuit will be nothing more than a straightforward implementation of measurement-based computation of some Clifford computation. However, because we must choose a larger and more computationally challenging group of Clifford operations to implement, we will need a few key new ideas, especially with respect to the randomization procedure.

\subsection{Formal statement of problem}

The starting point for \Cref{thm:parityL_main_result} is the problem of computing the product of matrices $A_1, \ldots, A_n \in \GL(m,2)$, which is a well-known $\parityL$-complete problem for $m = \poly(n)$ (see e.g., \cite{damm:1990}).  The problem is known to remain $\parityL$-hard when each matrix represents some $\CNOT(i,j)$ gate between bits $i, j \in [m]$ (see e.g., \Cref{lem:cnotmult_reduction} in \Cref{sec:parityL_complete_problems}).  In fact, we will need yet another modification to this problem for the purpose of randomization.

To see this, let us recall the randomization procedure required for the $\NC^1$-hardness result and show why it is insufficient here: given a product of group elements $g_1 \cdots g_n \in \mathcal C_2$, we construct the product $(h_0 g_1 h_1)(h_1^{-1} g_2 h_2) \ldots (h_{n-1}^{-1}g_n)$ with random $h_i \in \mathcal C_2$.  To randomize this product for $\parityL$-hardness, each $h_i$ term would be an arbitrary element of $\GL(m,2)$.  There are two main issues:
\begin{enumerate}[itemsep = 0pt]
\item Local measurement statistics such as those obtained from the magic square game are insufficient to reconstruct the highly entangled state $(h_0 g_1 h_1)(h_1^{-1} g_2 h_2) \ldots (h_{n-1}^{-1} g_n) \ket{+}^{\otimes m}$.
\item Inverting a matrix $A \in \GL(m,2)$ is $\parityL$-complete, so even implementing the randomization seems to make the reduction too powerful.  The simpler procedure of generating a random pair of matrices $A, A^{-1} \in \GL(m,2)$ is also not known to be in any class below $\parityL$ to the authors' knowledge.
\end{enumerate}

Because of these issues, we choose a different approach---namely, we perform the randomization with a group that is normalized by $\GL(m,2)$. Let us now carefully define the groups we will be concerned with in this section.  The largest group we will need is the group $\langle \CNOT, \CZ, \RZ \rangle_m$, which are those transformations generated by $\CNOT$, $\CZ$, and $\RZ$ on $m$-qubits.  More formally, the group contains those operations obtained by composing finitely many transformations on $m$ qubits from the set 
$$
\{\CNOT(i,j) : i\neq j \in [m] \} \cup \{ \CZ(i,j) : i\neq j \in [m]\} \cup \{ \RZ(i) : i \in [m] \}
$$ 
where $\CNOT(i,j)$ denotes the $\CNOT$ gate from qubit $i$ to qubit $j$, $\CZ(i,j)$ denotes a controlled-$Z$ gate from qubit $i$ to qubit $j$, and $\RZ(i)$ denotes a $\RZ$ gate on qubit $i$.  We will only be concerned with transformations modulo the Pauli operations.  For this reason, we define the group $G_m := \langle \CNOT, \CZ, \RZ \rangle_m / \Paulis_m$.  Notationally, when we refer to an element $g \in G_m$, we are referring to an actual transformation generated by $\CNOT$, $\CZ$, and $\RZ$.   However, when discussing equality,\footnote{These equivalence classes are well defined since $G_m$ normalizes $\Paulis_m$.} we say that $g = h$ for some $h \in G_m$ if the coset $g \Paulis_m$ equals the coset $h \Paulis_m$.

Let $H_m$ be the normal subgroup\footnote{Formally, it may be easier to verify that $\Paulis_m \trianglelefteq \langle \CZ, \RZ \rangle_m \trianglelefteq \langle \CNOT, \CZ, \RZ \rangle_m$.  Then, by the third isomorphism theorem, we have that $H_m \trianglelefteq G_m$ and
$$
\frac{G_m}{H_m} = \frac{\langle \CNOT, \CZ, \RZ \rangle_m / \Paulis_m}{\langle \CZ, \RZ \rangle_m / \Paulis_m} \cong \frac{\langle \CNOT, \CZ, \RZ \rangle_m}{\langle \CZ, \RZ \rangle_m}.
$$
} of $G_m$ consisting of those transformations generated by $\CZ$ and $\RZ$ only.  That is, $H_m := \langle \CZ, \RZ \rangle_m / \Paulis_m$.  The fact that $H_m$ is a normal subgroup of $G_m$ can be checked by straightforward calculation using the identities in \Cref{fig:commutation_relations_CNOT_Rz_CZ}.  

Another useful property of $H_m$ is that it is abelian, so an arbitrary element of $H_m$ can be represented by a layer of $\RZ$ gates followed by a layer of $\CZ$ gates, where no $\CZ(i,j)$ gate appears more than once.  In fact, since $\RZ^2 = \pfont{Z} \equiv \pfont{I}$ modulo Pauli operations, we can assume that there is at most one $\RZ$ gate per qubit as well.  Thus, for each $h \in H_m$, we associate a symmetric $m \times m$ binary matrix $A$ such that
$$
h = \prod_{i < j} \CZ(i,j)^{A_{i,j}} \prod_i \RZ(i)^{A_{i,i}}
$$
where the equality is modulo Pauli operations.  We will call this representation its \emph{canonical decomposition}.

Notationally, we will be somewhat sloppy when writing out products of elements in $G_m$ and $H_m$.  We write $g h$ for $g \in G_m$ and $h \in H_m$ to refer to the product of $g$ and $h$ in the group $G_m$. We will write $\CNOT_m$ for the set consisting of single $\CNOT$ gates on $m$ qubits (and also the identity)---i.e., $\CNOT_m = \{ \CNOT(i,j) : i,j \in [m] \text{ and } i \neq j \} \cup \{I_m\}$.

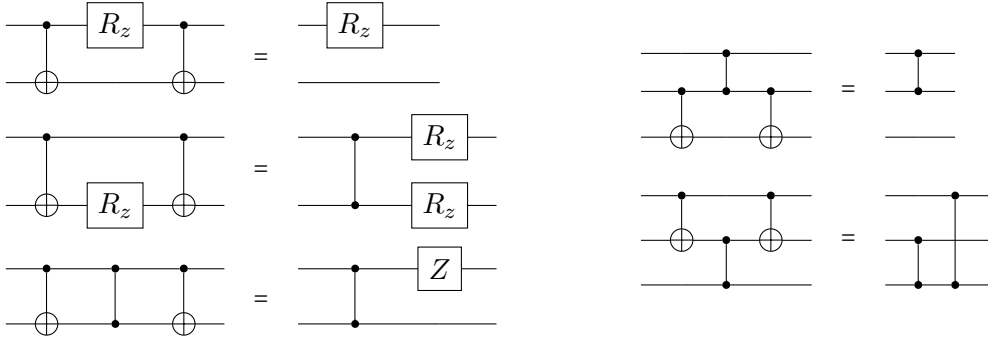
\begin{figure}
\begin{center}
\begin{quantikz}[row sep=1em, column sep=1em, thin lines]
& \ctrl{2} & \gate{\RZ} & \ctrl{2} & \qw &&& \gate{\RZ} & \qw
\\[-0.3cm] &&&&&=&\\[-0.3cm]
& \targ{} & \qw & \targ{} & \qw &&& \qw & \qw
\\[-.5cm] \\
& \ctrl{2} & \qw & \ctrl{2} & \qw &&& \ctrl{2} & \gate{\RZ} & \qw
\\[-0.3cm] &&&&&=&\\[-0.3cm]
& \targ{} & \gate{\RZ} & \targ{} & \qw &&& \ctrl{} & \gate{\RZ}& \qw
\\[-.5cm] \\
& \ctrl{2} & \ctrl{2} & \ctrl{2} & \qw &&& \ctrl{2} & \gate{Z} & \qw
\\[-0.3cm] &&&&&=&\\[-0.3cm]
& \targ{} & \ctrl{} & \targ{} & \qw &&& \ctrl{} & \qw & \qw
\end{quantikz}
\hspace{40pt}
\begin{quantikz}[row sep=1em, column sep=1em, thin lines]
& \qw & \ctrl{1} & \qw & \qw &&& \ctrl{1} & \qw \\
& \ctrl{1} & \ctrl{} & \ctrl{1} & \qw &=&& \ctrl{} & \qw \\
& \targ{} & \qw & \targ{} & \qw &&& \qw & \qw
\\[-.2cm] \\
& \ctrl{1} & \qw & \ctrl{1} & \qw &&& \qw & \ctrl{2} & \qw \\
& \targ{} & \ctrl{1} & \targ{} & \qw &=&& \ctrl{1} & \qw & \qw \\
& \qw & \ctrl{} & \qw & \qw &&& \ctrl{} & \ctrl{} & \qw
\end{quantikz}
\end{center}
\caption{Conjugating $\CZ$ and $\RZ$ by $\CNOT$.}
\label{fig:commutation_relations_CNOT_Rz_CZ}
\end{figure}

Finally, we base our problem on the hardness of evaluating CNOT circuits.  This is captured in the following theorem, the proof of which appears in \Cref{sec:parityL_complete_problems}.

\begin{reptheorem}{thm:CNOT_mult_hardness}
Given $g_1, \ldots, g_n \in \CNOT_m$ and promised that  $g_1 \cdots g_n$ is either a cycle on the first three qubits ($C_3$) or the identity transformation ($I$), it is $\parityL$-hard to decide which.  
\end{reptheorem}

\subsection{Tomography and the Magic Pentagram Game}

We would like to distinguish the cycle from the identity transformation by applying the permutation to a state. The $\ket{+}^{\otimes m}$ state is fixed by both permutations, so we introduce some $h \in H_m$ to change the state in such a way that after applying the permutation, the two possibilities may be distinguished. For example, if $h = \RZ(1)$ then we are attempting to distinguish the states $h_1 \ket{+}^{\otimes m}$ and $h_2 \ket{+}^{\otimes m}$ where $h_1 = \RZ(1)$ and $h_2 = \RZ(2)$.  Since the resulting states are both product states, it suffices to learn the state of the first two qubits---that is, are the stabilizer generators $\{ \pfont{XI}, \pfont{IY} \}$ or $\{ \pfont{YI}, \pfont{IX} \}$?

One might then be tempted to try to repeat the magic square measurements from the previous section to eventually learn this state.  Suppose, however, that the magic square measurements always reveal some element from the Pauli line $\{ \pfont{ZI}, \pfont{IZ}, \pfont{ZZ} \}$.  Since each element of the Pauli line is a non-stabilizer of both possible states, we learn nothing about our state.  On the other hand, we gain nothing from our randomization procedure since the Pauli line is fixed under conjugation by $\CZ$ and $\RZ$ gates.  Finally, every magic square game must contain at least one of \pfont{ZI}, \pfont{IZ}, or \pfont{ZZ}, ensuring that every magic square may return a result which is useless for distinguishing our two states.

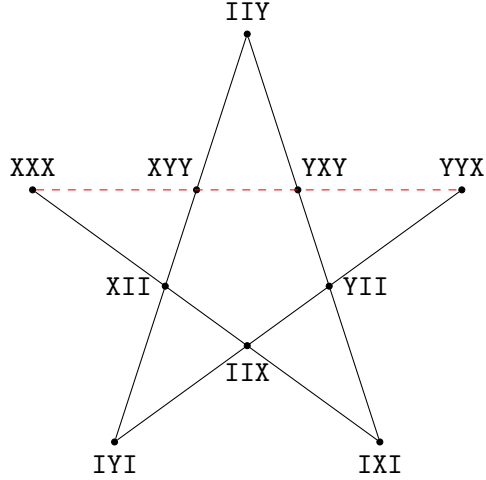
\begin{figure}
\centering

\begin{tikzpicture}[dot/.style={circle, fill=black, draw, inner sep=.8pt}]
\pgfmathsetmacro{\radius}{3}

    \node[dot] at ({\radius * cos((pi/2 + 0*2*pi/5) r)}, {\radius * sin((pi/2 + 0*2*pi/5) r)}) (1) [label={\pfont{IIY}}] {};
    \node[dot] at ({\radius * cos((pi/2 + 1*2*pi/5) r)}, {\radius * sin((pi/2 + 1*2*pi/5) r)}) (2) [label={\pfont{XXX}}] {};
    \node[dot] at ({\radius * cos((pi/2 + 2*2*pi/5) r)}, {\radius * sin((pi/2 + 2*2*pi/5) r)}) (3) [label=below:{\pfont{IYI}}] {};
    \node[dot] at ({\radius * cos((pi/2 + 3*2*pi/5) r)}, {\radius * sin((pi/2 + 3*2*pi/5) r)}) (4) [label=below:{\pfont{IXI}}] {};
    \node[dot] at ({\radius * cos((pi/2 + 4*2*pi/5) r)}, {\radius * sin((pi/2 + 4*2*pi/5) r)}) (5) [label={\pfont{YYX}}] {};
    
    \node[dot] at (intersection of 1--3 and 2--5) [label={[above, xshift=-10pt]\pfont{XYY}}] {};
    \node[dot] at (intersection of 1--3 and 2--4) [label=left:{\pfont{XII}}] {};
    \node[dot] at (intersection of 1--4 and 2--5) [label={[above, xshift=10pt]\pfont{YXY}}] {};
    \node[dot] at (intersection of 1--4 and 3--5) [label=right:{\pfont{YII}}] {};
    \node[dot] at (intersection of 2--4 and 3--5) [label={[below, yshift=-4pt]\pfont{IIX}}] {};

    \draw (2) -- (4) -- (1) -- (3) -- (5);
    \draw[dashed, red] (2) -- (5);
\end{tikzpicture}
\caption{Magic pentagram.}
\label{fig:magic_pentagram}
\end{figure}

This argument generalizes, and for this reason, we switch to a different set of contextual measurements shown in \Cref{fig:magic_pentagram} known as the magic pentagram \cite{mermin_magic_square}.  Each vertex in the magic pentagram is labeled by a Pauli operator in the set 
$$
\pentagram := \{\pfont{YII}, \pfont{IYI}, \pfont{IIY}, \pfont{XII}, \pfont{IXI}, \pfont{IIX}, \pfont{YYX}, \pfont{YXY}, \pfont{XYY}, \pfont{XXX} \}.
$$
Each line consists of 4 commuting operators which multiply to the identity, except for the dashed red line which multiplies to minus identity.  One can easily verify that there is no $\pm 1$ assignment to the vertices that multiply to 1 along the solid black lines and to $-1$ along the dashed red line.  Thus, we have the following:
\begin{theorem}
\label{thm:nonstab_from_pentagram}
There is a procedure to make five measurements (on five copies) of an unknown three-qubit quantum state $\ket{\psi}$ and learn, with certainty, some Pauli string which does not stabilize $\ket{\psi}$. Furthermore, the Pauli string is in the set $\pentagram$.
\end{theorem}
The proof is virtually identical to that of \Cref{thm:nonstab_from_magic}, so we omit it.  

\subsection{Randomization and Self-Reducibility}

In this section, we will slightly abuse notation and let $H_3 \le H_m$ be the subgroup of $H_m$ of $\CZ$ and $\RZ$ gates on the first three qubits.  Let $H_3^{\oplus} < H_3$ be the subgroup of $H_3$ that only contain those transformations with an even number of $\RZ$ and $\CZ$ gates.  Let $S := H_3^{\oplus} \bullet \pentagram = \{ h \bullet P  : h \in H_3^{\oplus} \text{ and } P \in \pentagram \}$ represent the set of possible Pauli strings obtained by measuring the first three qubits using magic pentagram measurements (conjugated by some element in $H_3^{\oplus}$).

Suppose we have some randomized algorithm $A$ which takes $g_1, \ldots, g_n \in G_m$ as input, and outputs a Pauli string $P \in S$ such that $P$ does not stabilize $\ket{\psi} = g_1 \cdots g_n \ket{+}^{\otimes m}$.  Let algorithm $B$ be obtained from \Cref{thm:symmetrize} where $F = H_3^{\oplus}$, $H = H_m$, $G = G_m$, and $A$ is the algorithm defined above.  That is, on input $g_1, \ldots, g_m \in G_m$, Algorithm $B$ returns a Pauli $(g_1 \cdots g_n) \bullet P$, where $P \in \Paulis_m$ is drawn from a distribution of Pauli strings which do not stabilize $\ket{+}^m$.  

We would now like to use algorithm $B$ to solve the $\CNOT$ multiplication problem (see \Cref{thm:CNOT_mult_hardness}). Recall that the product of $\CNOT$ gates is either the three-cycle or the identity.  Since any Pauli $P$ which non-stabilizes $\ket{+}^{\otimes 3}$ also non-stabilizes $C_3 \ket{+}^{\otimes 3}$, the output of algorithm $B$ is not very meaningful.  On the other hand, in the following theorem, we will show how an $\cL$ machine can smuggle an element of $H_3^{\oplus}$ into the product of the $\CNOT$ gates.  Since the two permutations act differently on elements of $H_3^{\oplus}$, the $\cL$ machine will be able to extract meaningful statistics.

\begin{theorem}
\label{thm:parityL_more_randomness}
Let $g_1, \ldots, g_n \in \CNOT_m$ for $m \ge 3$ be such that $g_1 \cdots g_n = \pi$ is some permutation on the first three qubits.  Let $f \in H_3^{\oplus}$.  Define Algorithm $C$ as below.
\begin{algorithmic}
	\Function{$C$}{$f, g_1, \ldots, g_n$}
 	\State{\Return $B(g_1, \ldots, g_n, f, g_n, \ldots, g_1)$}
	\EndFunction
\end{algorithmic}
Then the output distribution of $C(f, g_1, \ldots, g_n)$ is $(\pi f \pi^{-1}) \bullet \mathcal D(g_1 H_m, \ldots, g_n H_m)$ where the distribution $\mathcal D(g_1 H_m, \ldots, g_n H_m)$ is the average of $(g_1' \cdots g_{2n}')^{-1} \bullet A(g_1', \ldots, g_{2n}')$ over all $g_1', \ldots, g_{2n}'$ such that $g_i' H_m = g_i H_m$ and $g_{n+i}' H_m= g_{n-i+1} H_m$ for all $i \in [n]$ and $g_1' \cdots g_{2n}' \in H_3^{\oplus}$.
\end{theorem}
\begin{proof}
By \Cref{thm:symmetrize}, we have that $C$ will sample from the distribution
$$
(g_1 \cdots g_n f g_n \cdots g_1) \bullet \mathcal D(H_3^{\oplus}g_1 \cdots g_n g_n \cdots g_1, g_1 H_m, \ldots, g_n H_m, g_n H_m, \ldots, g_1 H_m)
$$
where $\mathcal{D}(\cdots)$ is the average of $(g_1', \ldots g_{2n}')^{-1} \bullet A(g_1', \ldots, g_{2n}')$ over all $g_1', \ldots, g_{2n}'$ such that  $g_1' \cdots g_{2n}' \in H_3^{\oplus} g_1 \cdots g_n f g_n \cdots g_1$ and $g_i' H_m = g_i H_m$ and $g_{n+i}' H_m = g_{n-i+1} H_m$ for all $i \in \{1, \ldots, n\}$.  Observe that $g = g^{-1}$ for all $g \in \CNOT_m$, so
$$
H_3^{\oplus} g_1 \cdots g_n f g_n \cdots g_1 = H_3^{\oplus} g_1 \cdots g_n f (g_1^{-1} \cdots g_n^{-1})^{-1} = H_3^{\oplus} \pi f \pi^{-1} = H_3^{\oplus}
$$
for any $f \in H_3^{\oplus}$.  That is, for fixed $g_1, \ldots, g_n$, Algorithm $C$ samples from the same distribution $(\pi f \pi^{-1}) \bullet \mathcal D$ for some fixed distribution $\mathcal D$ \emph{regardless of the choice of $f \in H_3^{\oplus}$}.  
\end{proof}

\begin{theorem}
\label{thm:detect_permutation}
Let $A$ be a randomized algorithm which takes $g_1, \ldots, g_n \in G_m$ as input, and outputs a Pauli string $P \in S$ such that $P$ does not stabilize $g_1 \cdots g_n \ket{+}^{\otimes m}$. Then, given $g_1, \ldots, g_n$, there exists a randomized algorithm in $(\BPAC^{0})^{A}$ which computes $g_1 \cdots g_n$, promised that it is either the three-cycle on the first three qubits or the identity transformation.
\end{theorem}
\begin{proof}
Fix $g_1, \ldots, g_n$ such that $g_1 \cdots g_n = \pi$, and let Algorithm $C$ be obtained by composing \Cref{thm:symmetrize} and \Cref{thm:parityL_more_randomness}.  By construction, $C$ outputs a non-stabilizer Pauli $(\pi f \pi^{-1}) \bullet P \in S$ of $\pi f \pi^{-1} \ket{+}^{\otimes m}$ where $P \sim \mathcal D:= \mathcal D(g_1H_m, \ldots, g_nH_m)$.  In other words, $\mathcal D(P)$ is a probability distribution over non-stabilizers of $\ket{+}^{\otimes m}$, and $\mathcal D_{\pi,f} := (\pi f \pi^{-1}) \bullet \mathcal D$ is the conjugated probability distribution over non-stabilizers of $\pi f \pi^{-1} \ket{+}^{\otimes m}$.  The algorithm works in two steps:
\begin{enumerate}[itemsep = 0pt]
\item Query algorithm $C(\pfont{III}, g_1, \ldots, g_n)$ to sample Pauli $P$ that has non-zero support in $\mathcal D$.
\item Choose $f \in H_3^{\oplus}$ based on query results, and query $C(f, g_1, \ldots, g_n)$ to draw $O(\log n)$-many samples from $\mathcal D_{\pi, f}$.  This will reveal $\pi$ with overwhelming probability.
\end{enumerate}

Let us now show how we should choose $f$ in the second step.  In particular, we will show that there exists $f \in H_3^{\oplus}$ such that the distribution $\mathcal D_{I, f}$ is not equal to $\mathcal D_{C_3, f}$.  We will argue by contradiction:  if there is no such $f$, then the support of $\mathcal D$ is empty.  Our starting point will be the fact that $\mathcal D$ does not have any support on the stabilizer group of $\ket{+}^{\otimes 3}$, which generated by \pfont{XII}, \pfont{IXI}, and \pfont{IIX}.  That is, we know a priori that $\mathcal D(\pfont{XII}) = \mathcal D(\pfont{IXI}) = \mathcal D(\pfont{IIX}) = 0$.  

Suppose we have some $f \in H_3^{\oplus}$ which does not distinguish the two permutations.  To be concrete, let $f = \RZ \otimes \RZ \otimes \pfont{I}$.  We have $\mathcal D_{I, f}(\pfont{YII}) = \mathcal D(\pfont{XII}) = 0$ and  $\mathcal D_{C_3, f}(\pfont{YII}) = \mathcal D(\pfont{YII})$.  Since $f$ does not distinguish the two permutations, it must be that $ \mathcal D_{I,f}(\pfont{YII}) = \mathcal D_{C_3, f}(\pfont{YII})$, so
$$
0 = \mathcal D_{I,f}(\pfont{YII}) = \mathcal D_{C_3, f}(\pfont{YII}) = \mathcal D(\pfont{YII}).
$$
That is, $\mathcal D$ cannot have any support on Pauli $\pfont{YII}$. More generally, consider any Pauli $P \in S$.  If $f \in H_3^{\oplus}$ does not distinguish $\mathcal D_{I, f}$ from $\mathcal D_{C_3, f}$, then $\mathcal D_{I, f}(P) = \mathcal D_{C_3, f}(P)$ and so
$$
\mathcal D (f \bullet P) = \mathcal D_{I, f}(P) = \mathcal D_{C_3, f}(P) = \mathcal D (C_3 f C_3^{-1} \bullet P).
$$
In other words, $\mathcal D(P) = \mathcal D(C_3 f C_3^{-1} f \bullet P)$ for all $P \in S$ (recall that $f = f^{-1}$ for all $f \in H_m$).  If there is no distinguishing $f$, we can now systematically show that $\mathcal D(P) = 0$ for every $P \in S$.  More precisely, for every $P \in S$, there exists a Clifford $f \in H_3^{\oplus}$ such that $D(C_3 f C_3^{-1} f \bullet P)$ is a stabilizer of $\ket{+}^{\otimes 3}$.  We show the complete enumeration in \Cref{table:pauli_support}.  Therefore, there must be some distinguishing $f$.

To conclude, we simply observe that there must exist (possibly multiple) Pauli $P \in S$ such that $\mathcal D(P) > \frac{1}{|S|}$.  We sample such $P$ in the first step of the protocol.  Let $f \in H_3^{\oplus}$ be the Clifford from \Cref{table:pauli_support} such that $C_3 f C_3^{-1} f \bullet P$ only has Pauli $\pfont X$ terms.  Let $Q := f \bullet P$.  We claim that whether or not we see $Q$ in the second step of the protocol reveals the permutation $\pi$.  Suppose $\pi = I$, then
$$
\mathcal D_{\pi, f}(Q) = \mathcal D_{I, f}(Q) = \mathcal D(f \bullet Q) = \mathcal D(P),
$$
and after $O(\log n)$ queries to $\mathcal D_{\pi, f}$, we will see $Q$ with high probability.  On the other hand, if $\pi = C_3$, then
$$
\mathcal D_{\pi, f}(Q) = \mathcal D_{C_3, f}(Q) = \mathcal D(C_3 f C_3^{-1} \bullet Q) = \mathcal D(C_3 f C_3^{-1} f \bullet P) = 0,
$$
so we will never see $Q$ in the second step, so we conclude that $\pi = C_3$.\footnote{Notice that this gives us one-sided error.  Suppose our task is to identify whether or not the permutation is the identity.  On ``yes'' instances, we output ``yes'' with high probability, but on ``no'' instances, we always output ``no.''  This places the reduction in $\RPAC^0$.  It's not hard to see that it is also in $\co\RPAC^0$:  for any given input $g_1, \ldots, g_n \in \CNOT_m$, append the permutation $C_3^{-1}$.  Thus, ``no'' instances now multiply to the identity, and ``yes'' instances multiply to $C_3^{-1}$.  The reduction is the same (simply replacing $C_3$ with $C_3^{-1}$ in the arguments above), and so the reduction is in $\co\RPAC^0$ by inverting the answer.  Since the reduction is in both $\co\RPAC^0$ and $\RPAC^0$, it is in $\ZPAC^0$.}
\end{proof}

\begin{table}
\centering
\begin{tabular}{ l | c | c | c }
$P \in S$ & $f \in H_3^{\oplus}$ & $C_3 f C_3^{-1} f$ & $C_3 f C_3^{-1} f \bullet P$ \\ \hline
\pfont{YII} & $\RZ(1) \cdot \RZ(2) $ & $\RZ(1) \cdot \RZ(3)$ &  \pfont{XII} \\
\pfont{XZI} & $\CZ(1,3) \cdot \CZ(2,3)$ & $\CZ(1,2) \cdot \CZ(2,3)$ &  \pfont{XII} \\
\pfont{YZI} & $\CZ(1,3) \cdot \CZ(2,3) \cdot \RZ(1) \cdot \RZ(2)$ & $\CZ(1,2) \cdot \CZ(2,3) \cdot \RZ(1) \cdot \RZ(3)$ &  \pfont{XII} \\
\pfont{XZZ} & $\CZ(1,2) \cdot \CZ(2,3)$ & $\CZ(1,2) \cdot \CZ(1,3)$ &  \pfont{XII} \\
\pfont{YZZ} & $\CZ(1,2) \cdot \CZ(2,3)\cdot \RZ(1) \cdot \RZ(2)$ & $\CZ(1,2) \cdot \CZ(1,3)\cdot \RZ(1) \cdot \RZ(3)$ &  \pfont{XII} \\
\pfont{YYX} & $\RZ(2) \cdot \RZ(3)$ & $\RZ(1) \cdot \RZ(2)$ &  \pfont{XXX} \\
\end{tabular}
\caption{Operations $f \in H_3^{\oplus}$ which distinguish $\mathcal D_{I,f}$ from $\mathcal D_{C_3, f}$. We omit $P \in S$ that are equivalent up to permutation.}
\label{table:pauli_support}
\end{table}

\subsection{Initial state details}
\label{sec:initial state}

Our initial graph state $\ket{\mathcal G_{n,m}}$ will be arranged for the measurement-based computation of the product of $\CNOT$ gates $g_1, \ldots, g_n \in \CNOT_m$.  In fact, recall that due to the randomization procedure, we will actually need to compute the product $g_1 h_1 g_2 h_2 \cdots g_n h_n$ where $h_i \in H_m$.  For elements of $H_m$, we can apply the corresponding layer of $\RZ$ gates in constant depth by \Cref{thm:single_qubit_mbqc}, so we will omit these from the discussion below.

Let us focus now on the application of a single $\CNOT$ gate from qubit $i$ to $j$.  In \Cref{sec:two_qubit_mbqc}, we give a general procedure to implement any transformation in $\mathcal C_2$ on a $2 \times 16$ grid.  If we wished to apply the same construction in a circuit of $m$ qubits, then we would need to know $i$ and $j$ in advance.  That is, we want to construct some fixed graph such that there exists some measurement pattern that implements $\CNOT(i,j)$ for all possible values of $i$ and $j$.  Importantly, we would also like the graph to be embeddable in the grid.

The most straightforward way to ensure that the graph state is grid-like is to only apply local operations.  We accomplish this by swapping far-away qubits until they become adjacent.  For instance, any $\CNOT(i,j)$ gate can be applied by swapping the $i$th bit until it adjacent to bit $j$, applying the gate, and then reversing the swaps.
There are two downsides to this procedure:
\begin{itemize} [itemsep = 0pt]
\item An element of $H_m$ could consists of $\Omega(m^2)$ many $\CZ$ gates.  So if we extended this swap architecture to apply an element of $H_m$, we would needed $\Omega(m^2)$ copies, adding a nontrivial overhead to the entire procedure.
\item Which $\SWAP$ gates are applied depends on both $i$ and $j$.  This complicates the choice of measurement basis for some qubit in the measurement-based application of some $\CNOT(i,j)$ gate.  Ideally, this choice of basis is a simple local operation.
\end{itemize}

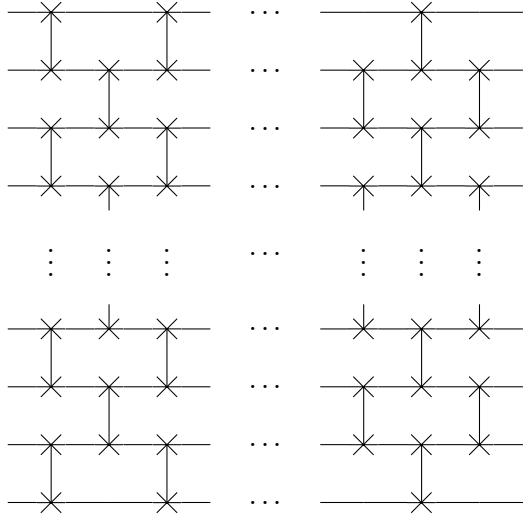
\begin{figure}
\centering
\scalebox{1}{
\begin{quantikz}[row sep=1em, column sep=1em, thin lines]
& \swap{1} & \qw  & \swap{1} & \qw & \ \ldots \ && \qw & \swap{1} & \qw & \qw \\
& \targX{} & \swap{1} & \targX{} & \qw & \ \ldots \ && \swap{1} & \targX{} & \swap{1} & \qw \\
& \swap{1} & \targX{} & \swap{1} & \qw & \ \ldots \ && \targX{} & \swap{1} & \targX{} & \qw \\
& \targX{} & \swap{1} & \targX{} & \qw & \ \ldots \ && \swap{1} & \targX{} & \swap{1} & \qw \\ [-.25cm]
&&&&&&&&& \\ [-.1cm]
& \ \vdots \ & \ \vdots \ & \ \vdots \ & & \ \ldots \ && \ \vdots \ & \ \vdots \ & \ \vdots \  \\ 
&&&&&&&&& \\ [-.25cm]
& \swap{1} & \swap{-1} & \swap{1} & \qw & \ \ldots \ && \swap{-1} & \swap{1} & \swap{-1} & \qw \\
& \targX{} & \targX{} & \targX{} & \qw & \ \ldots \ && \targX{} & \targX{} & \targX{} & \qw  \\
& \swap{1} & \swap{-1} & \swap{1} & \qw & \ \ldots \ && \swap{-1} & \swap{1} & \swap{-1} & \qw \\
& \targX{} & \qw & \targX{} & \qw & \ \ldots \ && \qw & \targX{} & \qw & \qw \\
\end{quantikz}}
\caption{Even-Odd sorting network: given any permutation, there is some subset of the SWAP gates above that implements that permutation.}
\label{fig:swap_networks}
\end{figure}

We solve both issues with the same construction.  The key idea will be to use a sorting network, which is a fixed network of SWAP gates on $m$ qubits that can be toggled to implement every permutation of those qubits.  That is, for every permutation, applying some subset of the SWAP gates in the network yields the permutation.  Remarkably, there are $O(m)$-depth sorting networks such that every SWAP gate is applied locally. 

We will use the even-odd sorting network \cite{habermann:1972_even_odd} of depth-$m$ shown in \Cref{fig:swap_networks}.  That is, on the $k$th layer, there are SWAP gates between qubits $(2i, 2i+1)$ for $k$ even, and qubits $(2i-1, 2i)$ for $k$ odd.  We now claim that the even-odd sorting network implies a $m \times O(m)$ grid implementing any transformation in $\CNOT_m$ or $H_m$.  First, we use the sorting network to reverse the order of the list of qubits, so qubit 1 becomes qubit $m$, qubit 2 becomes qubit $m-1$, and so on.  Since all swaps are local, at some point every pair of qubits must have been swapped.  At that point, we can apply any two-qubit gate using \Cref{thm:two_qubit_mbqc}.  Conveniently, our specific choice for encoding elements of $\CNOT_m$ and $H_m$ (given below) are such that the choice of measurement basis for each qubit depends only on a single bit of the input.  Once the qubits are in reverse order, we can simply reverse the order again to complete the transformation.  

In conclusion, we let $\mathcal G_{n,m}$ be the $m \times \Theta(nm)$ graph obtained by combining the above observations. This yields the following theorem:

\begin{theorem}
\label{thm:parityL_mbqc}
Let $g_1, \ldots, g_n \in \CNOT_m$ and $h_1, \ldots, h_n \in H_m$ be given.  Each element $h_i \in H_m$ is encoded by a matrix representing its canonical decomposition.  Each element $g_i \in \CNOT_m$ is encoded by a binary $m \times m$ matrix whose only non-zero entry $(i,j)$ represents the $\CNOT$ gate from qubit $i$ to qubit $j$.  There exist projections $P_1, \ldots, P_{\Theta(nm^2)}$ such that 
\begin{itemize}[itemsep = 0pt]
\item $P_i$ is either the projection $\frac{\pfont{I} + \pfont{X}}{2}$ or $\frac{\pfont{I} + \pfont{Y}}{2}$ onto the $i$th qubit,
\item $P_1 \otimes \ldots \otimes P_{\Theta(nm^2)} \otimes \pfont{I}^{\otimes m} \ket{\mathcal G_{n, m}} \propto \ket{\psi} \otimes g_1 h_1 g_2 h_2 \cdots g_n h_n \ket{+}^{\otimes m}$,
\item $P_1, \ldots, P_m$ are $X$-projections, and 
\item For all $i$, $P_i$ depends on at most a single bit of the input.
\end{itemize}
\end{theorem}

\subsection{Main theorem}

We now have all the ingredients to prove the main theorem.

\begin{reptheorem}{thm:parityL_main_result}
Let $\coracle$ be the rewind oracle for Wide Cluster Clifford Simulation (\Cref{prob:parityL}). Then 
$$\parityL \subseteq (\BPAC^{0})^{\coracle}.$$  
\end{reptheorem}
\begin{proof}
Our goal is to use the rewind oracle to determine the product $g_1 \cdots g_n$, modulo Pauli operations, given unitaries $g_1, \ldots, g_n \in \CNOT_m$.  By  \Cref{thm:CNOT_mult_hardness}, we will use that it is $\parityL$-hard to determine this product promised that it is the three-cycle or the identity transformation.

First, we construct Algorithm $A$ from the rewind oracle. Algorithm $A$ applies the oracle to $g_1', \ldots, g_n' \in G_m$ in the first round, then measures the five lines of the pentagram (all of them, by rewinding) in the second round. The first round measurements are specified by \Cref{thm:parityL_mbqc}. Recall that the measurements in the second round may require some entangling operations on the first three qubits (not simply $X$ and $Y$-basis measurments).  However, we can simulate these gates by yet more measurement-based computation using \Cref{thm:two_qubit_mbqc}.  Since the depth of these computations is constant, an $\AC^0$ verifier can apply the appropriate recovery operation to the oracle's answer in the second round.  From these measurements, we identify a non-stabilizer $P \in \pentagram$ of the state by \Cref{thm:nonstab_from_pentagram}. 

From Algorithm $A$, we construct Algorithm $B$ (\Cref{thm:symmetrize}), and then Algorithm $C$ (\Cref{thm:parityL_more_randomness}).  We must now check that an $\AC^0$ circuit can perform these randomization steps.  Algorithm $B$ queries Algorithm $A$ with elements $h_1 g h_2$ for $h_1,h_2 \in H_m$ and $g \in \CNOT_m$.  However, since \Cref{thm:parityL_mbqc} requires that each input be specified by a local input encoding, we must show that we can calculate the encoding of the product $h_1 g h_2$ given the individual encodings of $h_1$, $g$, and $h_2$.  However, since $H_m$ is abelian and conjugation by $\CNOT$ gates is locally computable (see \Cref{fig:commutation_relations_CNOT_Rz_CZ}), preparing the final encoding is done by $\AC^{0}$ circuits with randomness.

Finally, we use Algorithm $C$ to decide if $g_1 \cdots g_n$ is either $C_3$ or $I$, using \Cref{thm:detect_permutation}, which completes the proof.
\end{proof}


\section{Open Problems}

There are a few ways in which one could hope to improve our result---either by weakening the quantum circuit or strengthening the classical one.  First, one could weaken the quantum circuit by allowing for noise.  For example, Bravyi et al.\ \cite{bgkt:2019} give a separation between $\QNC^0$ circuits with local stochastic noise and $\NC^0$ circuits.  It is possible (or even likely) that a similar approach could work for our interactive model. Second, one could allow the classical circuit to err.  In \Cref{thm:nc1_error}, we show that a classical device which errs with probability less than $2/75$ must still be solving $\NC^1$-hard problems.  However, we have no such error guarantees for the $\parityL$-hardness result.  What can be said about the maximum allowable error rate in both cases?

We also ask whether or not interactivity is necessary for our results.  Assuming \Cref{conj:nc1_vs_qnc0}, we would require some non-blackbox feature to replace the role of rewinding in our reduction.  On the other hand, the conjecture does not rule out a separation between $\QNC^0$ and $\AC^0[p]$.  A natural question is whether or not such a separation exists.  It seems unlikely that $\MOD_3$ gates, say, help in solving the Parity Halving Problem of Bene Watts et al.\ \cite{bwkst:2019}.

\section*{Acknowledgments}

The authors would like to thank Scott Aaronson, Alex Lombardi, and Daochen Wang for helpful discussions.

\bibliographystyle{plain}
\phantomsection\addcontentsline{toc}{section}{References} 
\bibliography{bibliography}

\begin{appendices}
\addtocontents{toc}{\protect\setcounter{tocdepth}{2}}
\makeatletter
\addtocontents{toc}{%
  \begingroup
  \let\protect\l@section\protect\l@subsection
  \let\protect\l@subsection\protect\l@subsubsection
}
\makeatother


\section{Clifford Group}
\label{app:clifford}

In this appendix, we give some basic background on the Clifford group.  This includes of the tableau representation of a Clifford operation, which is used as the basis for efficient ($\parityL$) simulations of Clifford circuits as well as the final randomization step (\Cref{thm:nc1_self_reduce}) as part of the $\NC_1$-hardness result.

\subsection{Pauli Group}

First, recall the \emph{Pauli matrices}, a set of four $2 \times 2$ unitary matrices:
\begin{align*}
\pfont{I} &= \begin{pmatrix} 1 & 0 \\ 0 & 1 \end{pmatrix},&
\pfont{X} &= \begin{pmatrix} 0 & 1 \\ 1 & 0 \end{pmatrix},& 
\pfont{Y} &= \begin{pmatrix} 0 & -i \\ i & 0 \end{pmatrix},&
\pfont{Z} &= \begin{pmatrix} 1 & 0 \\ 0 & -1 \end{pmatrix}.
\end{align*}
Since the Pauli matrices satisfy the relations 
\begin{align*}
\pfont{X} \pfont{Y} &= i \pfont{Z},  & \pfont{Y} \pfont{Z} &= i \pfont{X}, & \pfont{Z} \pfont{X} &= i \pfont{Y}, \\
\pfont{Y} \pfont{X} &= -i \pfont{Z}, & \pfont{Z}\pfont{Y} &= -i \pfont{X}, & \pfont{X} \pfont{Z} &= -i \pfont{Y}, \\
& & \pfont{X}^2 = \pfont{Y}^2 &= \pfont{Z}^2 = \pfont{I} & &
\end{align*}
and $\pfont{I}$ is an identity element, the set $\Paulis_1 := \{ \pm 1, \pm i \} \times \{ \pfont{I}, \pfont{X}, \pfont{Y}, \pfont{Z} \}$ is a group under multiplication. This is the \emph{one-qubit Pauli group}, and it generalizes to the \emph{$m$-qubit Pauli group} $\Paulis_m := \{ \pm 1, \pm i \} \times \{ \pfont{I}, \pfont{X}, \pfont{Y}, \pfont{Z} \}^{\otimes m}$. We call the $\{ \pm 1, \pm i \}$ component the \emph{phase}, and the $\{ \pfont{I}, \pfont{X}, \pfont{Y}, \pfont{Z}\}^{\otimes m}$ component the \emph{Pauli string}. Let us name the group of signs $\mathcal{Z}_m := \{ \pm 1, \pm i \} \times \pfont{I}^{\otimes m}$, and note that $\mathcal{Z}_m$ is a normal subgroup of $\mathcal{P}_m$.  This means the quotient $\PStr{m}$ is well-defined.  Each element of $\PStr{m}$ is a coset $\{+ P, - P, +i P , -iP \}$ for some $P \in \{ \pfont{I}, \pfont{X}, \pfont{Y}, \pfont{Z}\}^{\otimes m}$, but we identify each such element with $P$, its \emph{positive} representative.

A useful property of the Pauli group is that it is a basis for all matrices.
\begin{fact}
Any matrix $A \in \mathbb C^{2^m \times 2^m}$ can be written as a complex linear combination of $\{ \pfont{I}, \pfont{X}, \pfont{Y}, \pfont{Z} \}^{\otimes m}$. 
\end{fact}
We introduce the notation $X^{(b_1, \ldots, b_m)} := \pfont{X}^{b_1} \otimes \cdots \otimes \pfont{X}^{b_m}$ for a bit vector $(b_1, \ldots, b_m) \in \F_2^{m}$, and similarly for $\pfont{Y}^{b}$ and $\pfont{Z}^b$. Another useful fact is that there are subsets of $2n$ Pauli elements which generate the whole group up to sign. 
\begin{fact}
Any $P \in \mathcal{P}_m$ can be written in the form $P := \alpha \pfont{X}^{a} \pfont{Z}^{b}$ where $\alpha \in \{ \pm 1, \pm i \}$ and $a, b \in \F_2^{m}$. Since $\pfont{X}^{e_1}, \ldots, \pfont{X}^{e_m}$ generates all $\pfont{X}^{a}$ and $\pfont{Z}^{e_1}, \ldots, \pfont{Z}^{e_m}$ generates all $\pfont{Z}^b$, together they generate all of $\mathcal{P}_m$ up to phase. 
\end{fact}

\subsection{Clifford Group}
The \emph{$m$-qubit Clifford group}, $\mathcal{C}_m$, is the set of $m$-qubit unitaries (under multiplication) normalizing the $m$-qubit Pauli group,
$$
\mathcal{C}_m := \{ U \in \mathrm{U}(2^{m}) : U \Paulis_m U^{\dag} = \Paulis_m \}.
$$
That is, a unitary $U \in \mathrm{U}(2^m)$ is \emph{Clifford} if for any $P \in \Paulis_m$, conjugation by $U$ gives an element of $\Paulis_m$. By construction, $\Paulis_m$ is a normal subgroup of $\mathcal{C}_m$, so 
$$
\mathcal{Z}_m \normalin \Paulis_m \normalin \Clifford_m
$$
Since conjugation of a Pauli by a Clifford operation is so common, we define the notation $\bullet \colon \Clifford_m \times \Paulis_m \to \Paulis_m$ where $U \bullet P := U P U^{\dag}$ for any $U \in \Clifford_m$ and $P \in \Paulis_m$.  

A \emph{Clifford circuit} is any quantum circuit built from a basis of Clifford unitaries as gates. An important theorem is that $\CNOT$, Hadamard ($H = \frac{1}{\sqrt{2}}(\begin{smallmatrix} 1 & 1 \\ 1 & -1 \end{smallmatrix})$), and Phase ($\RZ = \RZ(\pi/4) = (\begin{smallmatrix} 1 & 0 \\ 0 & i \end{smallmatrix})$) gates suffice to generate all Clifford unitaries, so this is the most common basis. A \emph{Clifford state} or \emph{stabilizer state} is any quantum state of the form $U \ket{0}^{\otimes m}$, where $U \in \mathrm{U}(2^m)$ is Clifford. A standard way to define a Clifford state is by its \emph{stabilizer group}:
$$
\mathrm{Stab}_{\ket{\psi}} := \{ P \in \Paulis_m : P \ket{\psi} = \ket{\psi} \}.
$$ 
Clifford states exhibit long-range entanglement, and become universal for quantum computation when augmented with non-Clifford gate. Nevertheless, Clifford circuits are efficiently classically simulable, in fact, in the complexity class $\parityL$ \cite{ag:2004}. We describe the essentials of this simulation in the next subsection. 

\subsection{Clifford Tableaux}

One way of proving the simulation result is through the \emph{tableau} representation of Clifford operations.  Any unitary can be defined (up to phase) by how it conjugates density matrices, i.e., by the set $\{(\rho, U \rho U^{\dag}) : \rho \succeq 0 \}$. Since the Pauli group forms a basis for all matrices, this can be reduced to the set $\{(P, U P U^{\dag}) : P \in \Paulis_m \}$. In fact, since $\pfont{X}^{e_1}, \ldots, \pfont{X}^{e_m}, \pfont{Z}^{e_1}, \ldots, \pfont{Z}^{e_m}$ generate $\Paulis_m$, and conjugation preserves products (i.e., $U P U^{\dag} U Q U^{\dag} = U PQ U^{\dag}$) it suffices to write down $U \pfont{X}^{e_i} U^{\dag}$ and $U \pfont{Z}^{e_i} U^{\dag}$ for all $i$. When $U$ is a Clifford unitary, we also get that $U \pfont{X}^{e_i} U^{\dag}$ and $U \pfont{Z}^{e_i} U^{\dag}$ are in $\Paulis_m$. For example, $\RZ \pfont{X} \RZ^{\dag} = \pfont{Y}$ and $\RZ \pfont{Z} \RZ^{\dag} = \pfont{Z}$, from which we can derive that 
$$\RZ \pfont{Y} \RZ^{\dag} = i \RZ \pfont{X} \RZ^{\dag} \RZ \pfont{Z} \RZ^{\dag} = i \pfont{Y} \pfont{Z} = -\pfont{X}.$$
Clearly $\RZ \pfont{I} \RZ^{\dag} = \pfont{I}$, and thus we can derive $\RZ \rho \RZ^{\dag}$ for any $\rho$ by linearity. Similarly, $H \pfont{X} H^{\dag} = \pfont{Z}, H \pfont{Z} H^{\dag} = \pfont{X}$ and 
\begin{align*}
\CNOT (\pfont{XI}) \CNOT^{\dag} &= \pfont{XX}, & \CNOT (\pfont{IX}) \CNOT^{\dag} &= \pfont{IX}, \\
\CNOT (\pfont{ZI}) \CNOT^{\dag} &= \pfont{ZI}, & \CNOT (\pfont{IZ}) \CNOT^{\dag} &= \pfont{ZZ}.
\end{align*}

In general, a Pauli operator $P = \alpha \pfont{X}^{a} \pfont{Z}^b \in \Paulis_m$ can be represented with two $m$-bit vectors $a, b \in \mathbb F_2^{m}$ and a pair of bits for $\alpha$, and any Clifford operation is defined by $2m$ Pauli operators $U \pfont{X}^{e_1} U^{\dag}$, \ldots, $U \pfont{X}^{e_m} U^{\dag}$, $U \pfont{Z}^{e_1} U^{\dag}$, \ldots, $U \pfont{Z}^{e_m} U^{\dag}$, the entire operation can be described by a matrix of $2n$ rows with $2m + 2$ bits per row. The phase information can be reduced to one bit per row (instead of two) to give a tableau, but since we will not need that part of the tableau in this paper, we will skip it and focus on the remaining $2m \times 2m$ binary matrix.\footnote{One can check that for all $X^a$ and $Z^b$, the phase of $U X^a U^\dag$ and $U Z^b U^\dag$ is either $+1$ or $-1$ for all $U \in \Clifford_m$.  We formally justify the sign issue in the next section, and in particular \Cref{lem:pauli_mod}.} We divide the matrix into four blocks, 
$$
\begin{bmatrix}
A & B \\
C & D 
\end{bmatrix}
$$
where $A, B, C, D \in \mathbb F_2^{m \times m}$.  That is, the $i$th row of $A$ records the $X$ component of the Pauli string $U \pfont{X}^{e_i} U^{\dag}$, and the $i$th row of $B$ represents the $Z$ component of the same Pauli string. Similarly, the $i$th row of $C$ and $D$ represent the $X$ and $Z$ components of the Pauli string $U \pfont{Z}^{e_i} U^{\dag}$. Then we have the following facts.
\begin{fact}
The tableau $[\begin{smallmatrix} A & B \\ C & D \end{smallmatrix}]$ corresponds to a Clifford operation if and only if it is \emph{symplectic}. The matrix is symplectic if and only if $A D^{T} + B C^{T} = I$ and both $AB^{T}$ and $CD^{T}$ are symmetric. 
\end{fact}
\begin{fact}
\label{fact:tableau_mult}
Let $U \in \Clifford_m$ be a Clifford unitary with tableau $[\begin{smallmatrix} A & B \\ C & D \end{smallmatrix}]$. Then $U (\pfont{X}^{r} \pfont{Z}^{s}) U^{\dag} = \alpha \pfont{X}^{u} \pfont{Z}^{v}$ for some $\alpha \in \{ \pm 1, \pm i \}$ if and only if 
$$
\begin{pmatrix}
r & s \\
\end{pmatrix}
\begin{pmatrix}
A & B \\
C & D
\end{pmatrix}
=
\begin{pmatrix}
u & v
\end{pmatrix}.
$$
It follows that, ignoring phase, the tableau for a unitary $VU$ is the standard matrix product of the tableau for $U$ and the tableau for $V$. 
\end{fact}

Since matrix multiplication over $\F_2$ is in $\parityL$, we can use the above fact to simulate a sequence of Clifford gates.  Measurement statistics can also be calculated in $\parityL$ (see \Cref{thm:multiqubit_measurement}).

\subsection{Clifford modulo Pauli, and Pauli modulo Phase}

We have already discussed in the detail the Clifford operations modulo Paulis and Paulis modulo phase in the background (\Cref{sec:background_clifford}).  The purpose of this section is to prove \Cref{lem:pauli_mod}, which we reword slightly below:

\begin{replemma}{lem:pauli_mod}
Define the homomorphism $\phi_C \colon \Paulis_m \to \Paulis_m$ where $\phi_C(Q) := C \bullet Q$ for $C \in \Clifford_m$ and $Q \in \Paulis_m$.
Two Clifford operations $C_1, C_2 \in \Clifford_m$ are equivalent modulo Paulis if and only if their action on $\Paulis_m / \mathcal{Z}_m$, i.e., $\phi_{C_1}(Q) \equiv \phi_{C_2}(Q) \pmod{\mathcal{Z}_m}$ for all $Q \in \Paulis_m$.
\end{replemma}
\begin{proof}
Let $C \in \Clifford_m$ be a Clifford operation. First, it is clear that the map $\phi_C \colon \Paulis_m \to \Paulis_m$ such that $\phi_C(Q) := C \bullet Q$ is a homomorphism. Since $\mathcal{Z}_m$ commutes with $C$, we have $\phi_C(\mathcal{Z}_m) = \mathcal{Z}_m$. Hence, modding out by $\mathcal{Z}_m$ gives $\tilde{\phi}_C \colon \Paulis_m / \mathcal{Z}_m \to \Paulis_m / \mathcal{Z}_m$. The formal claim is that $C_1$ and $C_2$ are equivalent modulo $\Paulis_m$ if and only if $\tilde{\phi}_{C_1} = \tilde{\phi}_{C_2}$. This reduces to checking that $C = C_1 C_2^{-1}$ is equivalent to $\pfont{I} \cdots \pfont{I}$ (i.e., $C$ is in $\Paulis_m$) if and only if $\tilde{\phi}_{C} = \tilde{\phi}_{\pfont{I} \cdots \pfont{I}}$ (i.e., if $\tilde{\phi}_C$ is the identity permutation).

It is a fact that $PQP^{\dag} = \pm Q$ for any $P, Q \in \Paulis_m$. Hence, if $C \in \Paulis_m$ then $\phi_C(Q) = \pm Q$, so $\tilde{\phi}_C$ is the identity permutation. This means there are at least $4^m$ Clifford operations $C$ (i.e., the Pauli operations) for which $\tilde{\phi}_C$ is the identity. In the tableau representation, these operations all have the same $2m \times 2m$ matrix, and only differ in the phase bits. Since we know $2m$ phase bits suffice, there are at most $2^{2m} = 4^{m}$ such Clifford operations, so they must be exactly the Pauli operations. 
\end{proof}


\section{\texorpdfstring{$\parityL$}{Parity L}-Complete Problems}
\label{sec:parityL_complete_problems}

The purpose of this section is to show that the $\CNOT$ multiplication problem is $\parityL$-hard.  We are particularly interested in the following variant:

\begin{problem}[$\CNOT$ Multiplication with Cycle Promise (\promisecnotmult)]
Given gates $g_1, \ldots, g_n \in \CNOT_m$, determine if $g_1 \cdots g_n$ is equal to the three-cycle on the first three bits or the identity transformation, promised that one is the case.  Each gate is represented by an $m \times m$ binary matrix such that the only non-zero entry $(i,j)$ denotes the $\CNOT$ gate from bit $i$ to bit $j$.  We take $m$ and $n$ polynomially related.
\end{problem}

The reduction will go through the following two problems, both of which were already known to be $\parityL$-complete (e.g., see Damm \cite{damm:1990}). We include this material for both completeness and to verify the efficiency of the reductions.

\begin{problem}[Layered DAG Path Parity]
A \emph{layered DAG} is a directed acyclic graph where the vertices are divided into an ordered sequence of \emph{layers} such that there are only edges from each layer to the next (no loops, no backwards edges, and no skipping layers).
 
An instance of $\ldagpar$ consists of a binary matrix describing a layered DAG with $n+1$ layers having $m$ nodes each where $m$ and $n$ are polynomially related. The goal is to compute the parity of the number of paths from some source node to some target node, where the source is in the first layer and the target is in the last layer. 
\end{problem}

\begin{problem}[$\CNOT$ Multiplication (\cnotmult)]
Given gates $g_1, \ldots, g_n \in \CNOT_m$, compute the top right entry of the binary matrix $g_1 \cdots g_n$ where $m$ and $n$ are polynomially related.
\end{problem}

\subsection{Reductions}

\begin{lemma}
$\parityL \subseteq (\NC^0)^{\ldagpar}$ for $\DLOGTIME$-uniform $\NC^0$ circuits.
\end{lemma}
\begin{proof}
Let $A \in \parityL$.  Consider the non-deterministic log-space machine $M$ for the language $A$. The main idea is to construct a layered DAG where there is a node for each configuration of $M$ (where the \emph{configuration} includes the contents of the work tape, the position of all tape heads, and the internal state) at each time step. There is a connection from a node in one layer to the next if $M$ could transition between the corresponding configurations in one time step. We leave it as an exercise to check that, in general, paths in the DAG correspond to computation paths.

All that remains is to check a few details:
\begin{itemize} [itemsep=0pt]
	\item Connectivity in the DAG is determined by the two configurations and a single bit of the input (so there is an $\NC^0$ circuit reducing an instance of $A$ to an instance of $\ldagpar$).  Furthermore, the compatibility of the two configurations, the location of the input bit, and its effect can be calculated in $\DLOGTIME$.  That is, the $\NC^0$ circuit is $\DLOGTIME$-uniform.
	\item We assume there is a single accepting configuration, so the number of accepting computation paths is exactly the number of DAG paths between a single source and target.
\end{itemize}
\end{proof}

\begin{lemma}
\label{lem:cnot_decomp}
Let $A = \{a_{i,j}\}$ be an $n \times n$ upper triangular binary matrix with ones along the diagonal. 
Then,
$$
A = \prod_{i=n-1}^{1} \prod_{j = i+1}^n \CNOT(j, i)^{a_{i,j} }.
$$
(assume $\prod_{i=1}^n x_i = x_1 x_2 \ldots x_n$ and $\prod_{i=n}^1 x_i = x_n x_{n-1} \ldots x_1$ for non-commutative variables $x_i$.)
\end{lemma}
\begin{proof}
The proof follows from a straightforward calculation using the fact that the matrix for $\CNOT(j,i)$ is the identity matrix except that the $(i,j)$th entry is 1.
\end{proof}

\begin{lemma}
\label{lem:cnotmult_reduction}
$\ldagpar \subseteq (\NC^0)^{\cnotmult}$ for $\DLOGTIME$-uniform $\NC^0$ circuits. 
\end{lemma}
\begin{proof}
The transitions between a pair of layers in a layered DAG can be expressed as an adjacency matrix. Thus, an instance of $\ldagpar$ is equivalent to binary matrices $A_1, \ldots, A_n \in \F_2^{m \times m}$, where the number of paths from a source $s_i$ in the first layer to a target $t_j$ in the last layer is the $(i,j)$th entry of the product $A_1 \cdots A_n$ over $\F_2$. 

Consider the matrix $\mathbf{A} \in \F_2^{m(n+1) \times m(n+1)}$ given below:
$$
\mathbf{A} = \begin{pmatrix}
I & A_1 & 0 & 0 & \cdots & 0 & 0 \\
0 & I & A_2 & 0 & \cdots & 0 & 0 \\
0 & 0 & I & A_3 & \cdots & 0 & 0 \\
\vdots & \vdots & & \ddots & \ddots & \vdots & \vdots \\
0 & 0 & 0 & 0 & \cdots & I & A_{n} \\
0 & 0 & 0 & 0 & \cdots & 0 & I
\end{pmatrix}.
$$
The inverse, $\mathbf{A}^{-1}$, is upper triangular, and we leave it as an exercise to show that for $i \leq j$, the $(i,j)$th block of $\mathbf{A}^{-1}$ is $(-1)^{j-i} A_{i} \cdots A_{j-1}$. In particular, the upper right block is (ignoring sign, since we are working in $\F_2$) the product $A_1 \cdots A_n$. In other words, we can assume without loss of generality that the top right entry of $\mathbf A^{-1}$ reports the parity of the number of paths from the source to the sink in the $\ldagpar$ problem.  Given the decomposition of $\mathbf A$ into $\CNOT$ gates by \Cref{lem:cnot_decomp}, we can reverse the $\CNOT$ gates to construct the circuit for the inverse of $\mathbf A$.

We can construct the sequence of $\CNOT$ gates in advance (see \Cref{lem:cnot_decomp}), and each $\CNOT$ is included or omitted (i.e., replaced with the identity) based on whether there is an edge between a corresponding pair of nodes in the layered DAG. Thus, each $\CNOT$ depends on a single input bit, so there is an $\NC^0$ circuit reducing an instance of $\ldagpar$ to $\cnotmult$. Moreover, the sequence and corresponding input bits can be determined in $\DLOGTIME$.
\end{proof}

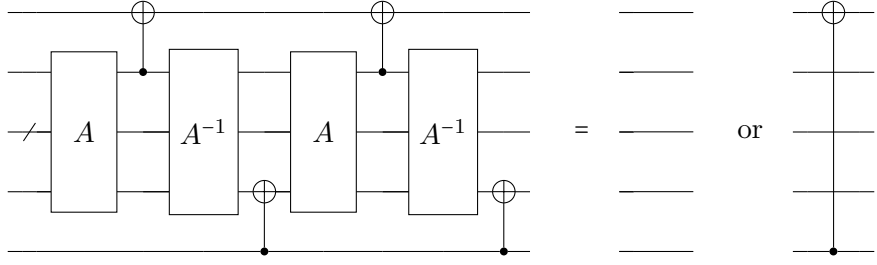
\begin{figure}
\centering
\begin{quantikz}[row sep=0em, column sep=.5em, thin lines]
&\qw &\qw & \qw & \targ{} & \qw & \qw & \qw & \targ{} & \qw & \qw & \qw &&&& \hphantomgate{wide}\qw &&&& \qw & \targ{} & \qw & \qw \\
&\qw &\qw & \gate[wires=3]{\,\; A \,\;} & \ctrl{-1} & \gate[wires=3]{A^{-1}} & \qw & \gate[wires=3]{\,\; A \,\;} & \ctrl{-1} & \gate[wires=3]{A^{-1}} & \qw & \qw &&&& \hphantomgate{wide}\qw &&&& \qw & \qw & \qw & \qw \\
&\qw &\qw\qwbundle{} & \qw & \qw & \qw {}&\qw &\qw &\qw &\qw &\qw &\qw  &&\hspace{.2cm} =\hspace{.2cm} && \hphantomgate{wide}\qw &&\hspace{.2cm}\text{or}\hspace{.2cm} && \qw&\qw&\qw & \qw\\
&\qw&\qw & \qw & \qw & \qw & \targ{} & \qw & \qw & \qw & \targ{} & \qw &&&& \hphantomgate{wide}\qw &&&& \qw & \qw & \qw & \qw \\
&\qw &\qw & \qw & \qw & \qw & \ctrl{-1} & \qw & \qw & \qw & \ctrl{-1} & \qw &&&& \hphantomgate{wide}\qw  &&&& \qw & \ctrl{-4} & \qw & \qw
\end{quantikz}
\caption{Construction of $\CNOT$ if top right entry of $A$ is equal to 1.}  
\label{fig:cnot_to_swap_construction}
\end{figure}

\begin{theorem}
\label{thm:CNOT_mult_hardness}
$\cnotmult \subseteq (\NC^0)^\promisecnotmult$ for $\DLOGTIME$-uniform $\NC^0$ circuits. 
\end{theorem}
\begin{proof}
Let $A = \{a_{i,j}\}$ be the matrix from the reduction in \Cref{lem:cnotmult_reduction} whose top right entry is $\parityL$-hard to compute.  We now give a procedure to create a new matrix from $A$ and additional $\CNOT$ gates, which isolates the effect of the top right entry---namely, if the top right entry is 1, the product of the gates involved is the 3-cycle on the first three bits; otherwise, the product is the identity gate.  We show the entire construction in \Cref{fig:cnot_to_swap_construction}.  

For the analysis, we break the construction into a simpler units. First, we extract information about the first row of $A$ using the following construction:
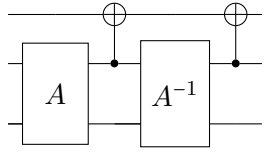
\begin{center}
\begin{quantikz}[row sep=.5em, column sep=.5em, thin lines]
& \qw & \targ{} & \qw & \targ{} & \qw \\
& \gate[wires=2]{\,\; A \,\;} & \ctrl{-1} & \gate[wires=2]{A^{-1}} & \ctrl{-1} & \qw \\
& \qw & \qw & \qw & \qw & \qw
\end{quantikz}
\end{center}
Call this matrix $B$.  We can determine $B$ by considering its effect on each standard basis vector $e_i \in \{0,1\}^{m+1}$ (where $e_i$ is the binary vector whose only non-zero entry is at location $i$). We have
\begin{center}
\scalebox{.8}{$
\begin{matrix}
(1,0, \ldots, 0) & \xrightarrow{A} & (1,0, \ldots, 0) & \xrightarrow{\CNOT} & (1,0, \ldots, 0) & \xrightarrow{A^{-1}} & (1,0, \ldots, 0) & \xrightarrow{\CNOT} & (1,0, \ldots, 0) \\
(0,1,0 \ldots, 0) & \xrightarrow{A} & (0,1,0, \ldots, 0) & \xrightarrow{\CNOT} & (1,1, \ldots, 0) & \xrightarrow{A^{-1}} & (1,1, \ldots, 0) & \xrightarrow{\CNOT} & (0,1, \ldots, 0) \\
e_{i+1} & \xrightarrow{A} & (0, a_{1,i}, \ldots, a_{m,i}) & \xrightarrow{\CNOT} & (a_{1,i}, a_{2,i} \ldots, a_{m,i}) & \xrightarrow{A^{-1}} & (a_{1,i},1,0, \ldots, 0) & \xrightarrow{\CNOT} & (a_{1,i},1,0, \ldots, 0) 
\end{matrix}$}
\end{center}
where the last row is for $i > 2$.  That is, $B$ is the identity matrix except the first row of $A$ is now (more or less) the first row of $B$.  It is now clear that if we repeat the construction except with the matrix $B^T$, we can extract the information from the last column:
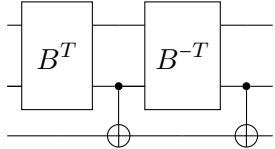
\begin{center}
\begin{quantikz}[row sep=.5em, column sep=.5em, thin lines]
& \gate[wires=2]{\, B^T \,} & \qw & \gate[wires=2]{B^{-T}} & \qw & \qw \\
& \qw & \ctrl{1} & \qw & \ctrl{1} & \qw \\
& \qw & \targ{} & \qw  & \targ{} & \qw 
\end{quantikz}
\end{center}

Notice that this will produce a $\CNOT$ gate if the top right entry of $A$ is 1; otherwise, it is the identity.  From here, we can simply use this $\CNOT$ gate to construct the 3-cycle on the first three qubits.  The entire construction of the $\CNOT$ gate is shown in \Cref{fig:cnot_to_swap_construction}, where some $\CNOT$ gates have been eliminated by noticing that the construction of the circuit $B$ is the same if you apply the $\CNOT$ gate first or last.

Clearly the new $\CNOT$ circuit is composed of copies of the old $\CNOT$ circuit, its inverse (i.e., the reversed sequence of $\CNOT$ gates), and a handful of new $\CNOT$ gates. It follows that there is a $\DLOGTIME$-uniform $\NC^{0}$ circuit for the reduction.
\end{proof}


\section{\texorpdfstring{$\parityL$}{Parity-L} Upper Bound}
\label{sec:upper_bounds}

One of the main results of this paper is that $\parityL$ is required for Wide Cluster Clifford Simulation (\Cref{prob:parityL}).  In this section, we prove a matching upper bound, showing that new problems involving non-Clifford gates will be required to go beyond $\parityL$-hardness.  There are two results that are almost what we want.  A theorem of Jozsa and Van den Nest shows that sampling Clifford circuits and calculating marginal probabilities is in \emph{polynomial time}, rather than $\parityL$ \cite{jv:2014_clifford}.  Second, the $\parityL$ procedure of Aaronson and Gottesman for simulating Clifford circuits only suffices for measurement on a single qubit \cite{ag:2004}.  Since our protocol involves sampling from many qubits, we must extend those arguments. 

\begin{problem}[Multi-qubit Clifford Measurement]
Given $n$ generators for an $n$-qubit stabilizer state $\ket{\psi}$ and an integer $m \le n$, output a measurement of the first $m$ qubits of $\ket{\psi}$ in the $Z$-basis and a corresponding set of generators for the measured state. 
\end{problem}
\begin{problem}[Postselected Multi-qubit Clifford Measurement]
Given $n$ generators for an $n$-qubit stabilizer state $\ket{\psi}$ and $p \in \{0,1\}^m$, output $n$ generators for the projection of $\ket{p}\bra{p} \otimes I^{m-n}$ onto $\ket {\psi}$, or report that projection is empty.  
\end{problem}

Let us define some useful notation for the following theorem. Recall that we can write any Pauli string $g \in \mathcal P_n$ as a $(2n +1)$-length binary vector: the Pauli-$\pfont{X}$ components, the Pauli-$\pfont{Z}$ components, and the sign bit.  For this reason, let us define the characteristic function for the $\pfont{X}$ Paulis of $g$ as $g^{(x)}$ where $g_{j}^{(x)} = 1$ if the $j$th Pauli is $\pfont{X}$ or $\pfont{Y}$, and 0 otherwise.  Similarly, define the characteristic function for the $\pfont{Z}$ Paulis of $g$ as $g^{(z)}$ where $g_{j}^{(z)} = 1$ if the $j$th Pauli is $\pfont{Z}$ or $\pfont{Y}$.  Let $\pfont{Z}_j \in \mathcal P_n$ be the Pauli operator which applies Pauli-$\pfont{Z}$ to the $j$th qubit.

\begin{theorem}
\label{thm:multiqubit_measurement}
The Multi-qubit Clifford Measurement Problem and the Postselected Multi-qubit Clifford Measurement Problem are in relational $\parityL$.
\end{theorem}
\begin{proof}
We will use the fact that a problem is in $\parityL$ if it is computable by an $\cL$ machine with access to a $\parityL$ oracle \cite{hrv:2000_parityL}.  Let us assume that our initial state has generators $G = \{g_1, g_2, \ldots, g_n\} \subseteq \mathcal P_n$.   To keep track of the characteristic Pauli-$\pfont{X}$ and Pauli-$\pfont{Z}$ vectors of the generators, we define the matrix $G^{(x)}_{s,t} = \{g^{(x)}_{i,j}\}$ for $i \in [s]$ and $j \in [t]$, and $G^{(z)}_{s,t}$ similarly.  We divide any multi-qubit measurement into two steps: measurement on a subset of the qubits where the outcomes will be random, and then measurement on the remaining qubits where the outcome is deterministic.  

\vspace{1em}\noindent\textit{Step 1: Measurements which return a random outcome}

\noindent When measuring multiple qubits at a time, we first identify a maximal set of qubits that yield a random outcome when measured.\footnote{It is worth noting here that this choice is not unique.  Consider the state $\frac{\ket{00} + \ket{11}}{\sqrt{2}}$.  Measuring either the first or second qubit results in a random outcome, but whichever qubit is measured second will have the same outcome as the first.} To do this, let us first recall the procedure for measuring a single qubit \cite{ag:2004, gottesman:1998_stabilizers}.  If we measure qubit $j$ and there exists a generator $g_i$ such that $g_{i,j}^{(x)} = 1$, then the measurement outcome is random.  We can construct a new set of generators as follows:
\begin{itemize} [itemsep = 0pt]
\item For all $k \neq i$ such that $g_{k,j}^{(x)} = 1$, replace $g_k$ with $g_k g_i$,
\item Replace $g_i$ with $\sigma_j \pfont{Z}_j$ where $\sigma_j$ is a random sign in $\{-1,+1\}$.
\end{itemize}

Thus, the positions of the nondeterministic measurements are a function of the matrix $G^{(x)}_{n,m}$.  In fact, the number of random measurements is equal to the rank of $G^{(x)}_{n,m}$ since the single-qubit measurement procedure described above row reduces on this submatrix.  We can break the row reduction into two $\parityL$ steps: choosing a basis, and selecting new generators.

First, we select a basis $B = \{b_1, \ldots, b_r\} \subseteq G$ such that $|B| = \rank(G^{(x)}_{n,m})$ and the row space of $B^{(x)}_{r,m}$ is equal to the row space of $G^{(x)}_{n,m}$, where the matrix $B^{(x)}_{s,t}$ is once again defined as $\{b^{(x)}_{i,j}\}$ for $i \in [s]$ and $j \in [t]$.  We select this basis in the following way:  
$$
g_i \in B \iff \rank(G^{(x)}_{i,m}) > \rank(G^{(x)}_{i-1,m})
$$
Notice that we can construct this set with $\cL$ machine with access to an oracle for rank (which is in $\parityL$). 

The next step is to compute a new set of generators such that their characteristic Pauli-$X$ vectors are in row echelon form.  First, define vectors $f_j = (0,\ldots, 0, 1)$ of length $j$.  Notice that there are exactly $r$ different $f_j$ such that $f_j$ is in the row space of $B^{(x)}_{r,j}$.  This comes directly from the row echelon form of $B^{(x)}_{r,m}$.  Once again, we can find these $r$ vectors in $\parityL$ by comparing the rank of $B^{(x)}_{r,j}$ to the rank of $f_j$ appended to $B^{(x)}_{r,j}$.  Let $F = \{f_j : f_j \in \operatorname{rowsp}(B^{(x)}_{r,j}) \}$ be this set.

For each $f_j \in F$, solve the linear equation $x B^{(x)}_{r,j} = f_j$ where $x \in \{0,1\}^{1 \times r}$ is unknown, and add $\prod_{k=1}^r x_k b_k$ to the new list of generators.  For each generator $g_i \in G \backslash B$ not in the basis, solve the linear equation $x B^{(x)}_{r,m} = (g^{(x)}_{i,1}, \ldots, g^{(x)}_{i,m})$, and add $g_i \prod_{k=1}^r x_k b_k$ to the new list of generators.  Let us return to the issue of computing these products in $\parityL$ later.

The final stage of the random measurements procedure depends on whether or not there is postselection. For the Multi-qubit Clifford Measurement Problem, we replace each generator corresponding to a vector $f_j \in F$ by a new generator $\sigma_j \pfont{Z}_j$ where $\sigma_j$ is a random sign in $\{-1,+1\}$.  For the Postselected Multi-qubit Clifford Measurement Problem, we perform the same replacement, but set $\sigma_j = (-1)^{p_j}$ where $p_j$ was the intended value of the postselected bit.

\vspace{1em}\noindent\textit{Step 2: Measurements with deterministic outcome}

\noindent Now that the random measurement outcomes are fixed, we must compute the outcomes of the deterministic measurements. Since these measurements do not affect the state, we already know what the final generators for the state are (except in the postselected case, where the projection may be empty). Let $G$ be the set of generators for the state after measurement in the first step.  Supposing the outcome on qubit $j$ is deterministic, it must be that some product of the generators in $G$ yields $\pfont{Z}_j$ or $-\pfont{Z}_j$.  We only need to report which is true for each $j \in [m]$.
 
Let $e_j$ to be the all-zeros vector of length $n$ which has a $1$ at position $j$.  Solve the equations $x G^{(z)} = e_j$ and $x G^{(x)} = (0, \ldots, 0)$ for $x \in \{0,1\}^{1 \times n}$, and let $\sigma_j = \operatorname{sgn}( \prod_{k=1}^n x_k g_k)$.  For the Multi-qubit Clifford Measurement Problem, output $\sigma_j$. For the Postselected Multi-qubit Clifford Measurement Problem, report that the projection is empty if $\sigma_j \neq (-1)^{p_j}$.

All that remains to complete the theorem is to show that we can compute products of generators in $\parityL$.  The first thing to note is that we cannot keep an entire generator $g \in \mathcal P_n$ in memory.  Each time we need the $j$th Pauli of $g$, we recompute it from scratch.  Notice that calculating the product of polynomially many elements of $\mathcal P_1$ can performed in $\cL$.  Since we are computing each Pauli locally, we must also keep track of a local sign bit in $\{\pm 1, \pm i\}$.   To compute the sign of a generator (which will necessarily be either $+1$ or $-1$), we compute the product of each of the local sign bits, which once again is an operation in $\cL$.  This completes the proof.
\end{proof}


\section{ZX-calculus}
\label{app:zx}

The ZX-calculus is a graphical language for quantum computing. It consists of a small set of generators (\Cref{tab:zx_generators}), which are graphical descriptions of states, operations, isometries, and projections, and a set of local replacement rules (\Cref{tab:zx_rules}). We only describe those generators and rules that we will use throughout this paper, which are not complete for quantum computation.

\begin{table}[ht!]
\centering
\begin{tabular}{| l | c l c | l |} \hline
Type  & \hspace{.5em} & Generator & \hspace{.5em}  & Meaning \\ \hline
state  && \parbox[c][2em]{0pt}{\begin{tikzpicture}
	\begin{pgfonlayer}{nodelayer}
		\node [style=none] (1) at (1, 0) {};
		\node [style=Z phase dot] (2) at (0, 0) {$\alpha$};
	\end{pgfonlayer}
	\begin{pgfonlayer}{edgelayer}
		\draw (2) to (1.center);
	\end{pgfonlayer}
\end{tikzpicture}} && $\ket{0} + e^{i \alpha} \ket{1}$ \\ \hline
unitary && \parbox[c][2em]{0pt}{\begin{tikzpicture}
	\begin{pgfonlayer}{nodelayer}
		\node [style=none] (1) at (0.75, 0) {};
		\node [style=Z phase dot] (2) at (0, 0) {$\alpha$};
		\node [style=none] (3) at (-0.75, 0) {};
	\end{pgfonlayer}
	\begin{pgfonlayer}{edgelayer}
		\draw (2) to (1.center);
		\draw (2) to (3.center);
	\end{pgfonlayer}
\end{tikzpicture}} && $\ket{0}\bra{0} + e^{i \alpha} \ket{1}\bra{1}$ \\ \hline
unitary && \parbox[c][2em]{0pt}{\begin{tikzpicture}
	\begin{pgfonlayer}{nodelayer}
		\node [style=none] (1) at (0.75, 0) {};
		\node [style=none] (3) at (-0.75, 0) {};
		\node [style=hadamard] (4) at (0, 0) {};
	\end{pgfonlayer}
	\begin{pgfonlayer}{edgelayer}
		\draw (3.center) to (4);
		\draw (4) to (1.center);
	\end{pgfonlayer}
\end{tikzpicture}} && $\ket{+}\bra{0} + \ket{-}\bra{1}$ \\ \hline
projection && \parbox[c][2em]{0pt}{\input{figures/z_projection.tikz}} && $\bra{0} + e^{i \alpha} \bra{1}$ \\ \hline
spider && \hspace{-15pt}\parbox[c][5em]{0pt}{\begin{tikzpicture}
	\begin{pgfonlayer}{nodelayer}
		\node [style=none] (1) at (-0.75, 0.75) {};
		\node [style=Z phase dot] (2) at (0, 0) {$\alpha$};
		\node [style=none] (3) at (-0.75, 0.5) {};
		\node [style=none] (4) at (-0.75, -0.75) {};
		\node [style=none] (5) at (0.75, -0.75) {};
		\node [style=none] (6) at (0.75, 0.75) {};
		\node [style=none] (7) at (0.75, 0.5) {};
		\node [style=none] (8) at (-0.5, 0) {$\vdots$};
		\node [style=none] (9) at (0.5, 0) {$\vdots$};
		\node [style=none] (10) at (1, 0.75) {};
		\node [style=none] (11) at (1, -0.75) {};
		\node [style=none] (12) at (-1, 0.75) {};
		\node [style=none] (13) at (-1, -0.75) {};
		\node [style=none] (14) at (1.25, 0) {$m$};
		\node [style=none] (16) at (-1.25, 0) {$n$};
	\end{pgfonlayer}
	\begin{pgfonlayer}{edgelayer}
		\draw [bend right, looseness=1.25] (2) to (1.center);
		\draw [bend right] (2) to (3.center);
		\draw [bend left, looseness=0.75] (2) to (4.center);
		\draw [bend right] (2) to (5.center);
		\draw [bend left] (2) to (7.center);
		\draw [bend left, looseness=1.25] (2) to (6.center);
		\draw [style=brace edge] (10.center) to (11.center);
		\draw [style=brace edge] (13.center) to (12.center);
	\end{pgfonlayer}
\end{tikzpicture}} && $\ket{0}^{\otimes m} \bra{0}^{\otimes n}  + e^{i \alpha} \ket{1}^{\otimes m} \bra{1}^{\otimes m} $ \\ \hline
\end{tabular}
\caption{Generators for the ZX-calculus. For each green generator listed above for the $Z$-basis, there is an analogous red generator for the $X$-basis. By convention, a solid green/red circle implies that $\alpha$ equals $0$.}
\label{tab:zx_generators}
\end{table}

\begin{table}[ht!]
\centering
\begin{tabular}{| l | c l c |} \hline
Rule name  & \hspace{.5em} & Rule \hspace{11em} & \hspace{.5em}  \\ \hline
Spider fusion  && \parbox[c][5em]{0pt}{\begin{tikzpicture}
	\begin{pgfonlayer}{nodelayer}
		\node [style=none] (1) at (0.5, 0.75) {};
		\node [style=Z phase dot] (2) at (1.25, 0) {$\;\alpha + \beta\;$};
		\node [style=none] (3) at (0.5, 0.5) {};
		\node [style=none] (4) at (0.5, -0.75) {};
		\node [style=none] (5) at (2, -0.75) {};
		\node [style=none] (6) at (2, 0.75) {};
		\node [style=none] (7) at (2, 0.5) {};
		\node [style=none] (8) at (0.6, 0) {\scriptsize$\vdots$};
		\node [style=none] (9) at (1.9, 0) {\scriptsize$\vdots$};
		\node [style=none] (17) at (-3, 0.75) {};
		\node [style=Z phase dot] (18) at (-2.25, 0.25) {$\alpha$};
		\node [style=none] (19) at (-3, 0.6) {};
		\node [style=none] (20) at (-3, 0) {};
		\node [style=none] (21) at (-0.75, 0) {};
		\node [style=none] (22) at (-0.75, 0.75) {};
		\node [style=none] (23) at (-0.75, 0.6) {};
		\node [style=none] (24) at (-2.75, 0.4) {\scriptsize$\vdots$};
		\node [style=none] (25) at (-1.25, 0.35) {\scriptsize$\vdots$};
		\node [style=none] (32) at (-0.1, 0) {$=$};
		\node [style=none] (33) at (-3, -0.25) {};
		\node [style=Z phase dot] (34) at (-1.75, -0.5) {$\beta$};
		\node [style=none] (35) at (-3, -0.4) {};
		\node [style=none] (36) at (-3, -0.75) {};
		\node [style=none] (37) at (-0.75, -0.75) {};
		\node [style=none] (38) at (-0.75, -0.25) {};
		\node [style=none] (39) at (-0.75, -0.4) {};
		\node [style=none] (40) at (-2.5, -0.47) {\scriptsize$\vdots$};
		\node [style=none] (41) at (-1.18, -0.48) {\scriptsize$\vdots$};
	\end{pgfonlayer}
	\begin{pgfonlayer}{edgelayer}
		\draw [bend right, looseness=1.25] (2) to (1.center);
		\draw [bend right] (2) to (3.center);
		\draw [bend left, looseness=0.75] (2) to (4.center);
		\draw [bend right] (2) to (5.center);
		\draw [bend left] (2) to (7.center);
		\draw [bend left, looseness=1.25] (2) to (6.center);
		\draw [bend right, looseness=1.25] (18) to (17.center);
		\draw [bend right] (18) to (19.center);
		\draw [bend left, looseness=0.75] (18) to (20.center);
		\draw [bend right=15] (18) to (21.center);
		\draw [bend left=15] (18) to (23.center);
		\draw [bend left=15, looseness=1.25] (18) to (22.center);
		\draw [bend right, looseness=0.75] (34) to (33.center);
		\draw [bend right, looseness=0.75] (34) to (35.center);
		\draw [bend left, looseness=0.75] (34) to (36.center);
		\draw [bend right, looseness=0.75] (34) to (37.center);
		\draw [bend left=15, looseness=1.25] (34) to (39.center);
		\draw [bend left] (34) to (38.center);
		\draw (18) to (34);
	\end{pgfonlayer}
\end{tikzpicture}} & \\ \hline
Identity rule  && \parbox[c][2em]{0pt}{\begin{tikzpicture}
	\begin{pgfonlayer}{nodelayer}
		\node [style=none] (1) at (-0.75, 0) {};
		\node [style=none] (3) at (-2.25, 0) {};
		\node [style=Z dot] (4) at (-1.5, 0) {};
		\node [style=none] (6) at (2.25, 0) {};
		\node [style=none] (7) at (0.75, 0) {};
		\node [style=none] (11) at (0, 0) {$=$};
	\end{pgfonlayer}
	\begin{pgfonlayer}{edgelayer}
		\draw (3.center) to (4);
		\draw (4) to (1.center);
		\draw (7.center) to (6.center);
	\end{pgfonlayer}
\end{tikzpicture}} & \\ \hline
Color change  && \parbox[c][5em]{0pt}{\begin{tikzpicture}
	\begin{pgfonlayer}{nodelayer}
		\node [style=none] (1) at (-2, 0.75) {};
		\node [style=Z phase dot] (2) at (-1.25, 0) {$\alpha$};
		\node [style=none] (4) at (-2, -0.75) {};
		\node [style=none] (5) at (-0.5, -0.75) {};
		\node [style=none] (6) at (-0.5, 0.75) {};
		\node [style=none] (8) at (-1.9, 0.1) {\scriptsize$\vdots$};
		\node [style=none] (9) at (-0.6, 0.1) {\scriptsize$\vdots$};
		\node [style=none] (32) at (.1, 0) {$=$};
		\node [style=hadamard] (33) at (-1.5, 0.5) {};
		\node [style=hadamard] (34) at (-1.5, -0.5) {};
		\node [style=hadamard] (35) at (-1, 0.5) {};
		\node [style=hadamard] (36) at (-1, -0.5) {};
		\node [style=none] (37) at (0.65, 0.75) {};
		\node [style=X phase dot] (38) at (1.4, 0) {$\alpha$};
		\node [style=none] (39) at (0.65, -0.75) {};
		\node [style=none] (40) at (2.15, -0.75) {};
		\node [style=none] (41) at (2.15, 0.75) {};
		\node [style=none] (42) at (0.75, 0.1) {\scriptsize$\vdots$};
		\node [style=none] (43) at (2.05, 0.1) {\scriptsize$\vdots$};
	\end{pgfonlayer}
	\begin{pgfonlayer}{edgelayer}
		\draw [bend left] (2) to (6.center);
		\draw [bend left] (1.center) to (2);
		\draw [bend left] (2) to (4.center);
		\draw [bend right] (2) to (5.center);
		\draw [bend left] (38) to (41.center);
		\draw [bend left] (37.center) to (38);
		\draw [bend left] (38) to (39.center);
		\draw [bend right] (38) to (40.center);
	\end{pgfonlayer}
\end{tikzpicture}} & \\ \hline
$\pi$-copy  && \parbox[c][5em]{0pt}{\begin{tikzpicture}
	\begin{pgfonlayer}{nodelayer}
		\node [style=Z phase dot] (2) at (-1.25, 0) {$\alpha$};
		\node [style=none] (9) at (-0.6, -0.05) {\scriptsize$\vdots$};
		\node [style=none] (32) at (0.1, 0) {$=$};
		\node [style=X phase dot] (44) at (-1.75, 0) {$\pi$};
		\node [style=none] (45) at (-2.25, 0) {};
		\node [style=none] (46) at (-0.5, 0.75) {};
		\node [style=none] (47) at (-0.5, 0.25) {};
		\node [style=none] (48) at (-0.5, -0.5) {};
		\node [style=Z phase dot] (49) at (1.1, 0) {$-\alpha$};
		\node [style=none] (50) at (2.25, -0.05) {\scriptsize$\vdots$};
		\node [style=none] (52) at (0.6, 0) {};
		\node [style=none] (53) at (2.35, 0.75) {};
		\node [style=none] (54) at (2.35, 0.25) {};
		\node [style=none] (55) at (2.35, -0.5) {};
		\node [style=X phase dot] (56) at (1.75, 0.75) {$\pi$};
		\node [style=X phase dot] (57) at (1.75, 0.25) {$\pi$};
		\node [style=X phase dot] (58) at (1.75, -0.5) {$\pi$};
	\end{pgfonlayer}
	\begin{pgfonlayer}{edgelayer}
		\draw (45.center) to (44);
		\draw (44) to (2);
		\draw [bend left] (2) to (46.center);
		\draw [bend left=15] (2) to (47.center);
		\draw [bend right, looseness=0.75] (2) to (48.center);
		\draw (52.center) to (49);
		\draw [bend left=45, looseness=0.75] (49) to (56);
		\draw [bend left] (49) to (57);
		\draw [bend right] (49) to (58);
		\draw (56) to (53.center);
		\draw (57) to (54.center);
		\draw (58) to (55.center);
	\end{pgfonlayer}
\end{tikzpicture}} & \\ \hline
\end{tabular}
\caption{Rules for the ZX-calculus. In every rule, red and green nodes can be exchanged.}
\label{tab:zx_rules}
\end{table}

We will primarily use the ZX-calculus to determine the affect of the $X$ and $Y$ measurements on graph states.\footnote{Using the ZX-calculus to analyze measurement-based computation is not new. For instance, see references \cite{cd:2011_zx} and \cite{duncan:2010}.} First, one can verify that the following diagram represents a $\CSIGN$ gate:
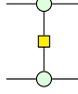
\begin{center}
\begin{tikzpicture}
	\begin{pgfonlayer}{nodelayer}
		\node [style=none] (1) at (-0.5, 0) {};
		\node [style=none] (2) at (0.5, 0) {};
		\node [style=none] (4) at (-0.5, -1) {};
		\node [style=none] (6) at (0.5, -1) {};
		\node [style=Z dot] (7) at (0, 0) {};
		\node [style=Z dot] (8) at (0, -1) {};
		\node [style=hadamard] (9) at (0, -0.5) {};
	\end{pgfonlayer}
	\begin{pgfonlayer}{edgelayer}
		\draw (1.center) to (7);
		\draw (7) to (2.center);
		\draw (7) to (9);
		\draw (9) to (8);
		\draw (8) to (4.center);
		\draw (8) to (6.center);
	\end{pgfonlayer}
\end{tikzpicture}
\end{center}
Using this fact and the spider fusion rule from \Cref{tab:zx_rules}, it is easy to prove the following lemma:
\begin{lemma}
\label{lem:zx_graph_state}
Every graph state $\ket{G}$ for $G = (V,E)$ is represented by a $|V|$-qubit ZX-calculus diagram where two qubits in the diagram are connected by edge with a Hadamard gate if and only if they share an edge in the graph $G$.
\end{lemma}
For instance, the triangle graph on three vertices is represented by the diagram
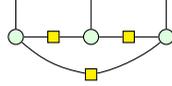
\begin{center}
\begin{tikzpicture}
	\begin{pgfonlayer}{nodelayer}
		\node [style=Z dot] (1) at (0, 0) {};
		\node [style=Z dot] (2) at (1, 0) {};
		\node [style=Z dot] (3) at (2, 0) {};
		\node [style=hadamard] (4) at (0.5, 0) {};
		\node [style=hadamard] (5) at (1.5, 0) {};
		\node [style=hadamard] (6) at (1, -0.5) {};
		\node [style=none] (7) at (0, 0.5) {};
		\node [style=none] (8) at (1, 0.5) {};
		\node [style=none] (9) at (2, 0.5) {};
	\end{pgfonlayer}
	\begin{pgfonlayer}{edgelayer}
		\draw (7.center) to (1);
		\draw (8.center) to (2);
		\draw (9.center) to (3);
		\draw (1) to (4);
		\draw (4) to (2);
		\draw (2) to (5);
		\draw (5) to (3);
		\draw [bend right=15] (1) to (6);
		\draw [bend right=15, looseness=0.75] (6) to (3);
	\end{pgfonlayer}
\end{tikzpicture}
\end{center}

We now want to show the effect of $X$ and $Y$ measurements on the graph.  In the subsequent discussion, we will actually refer to \emph{projections}, rather than measurements.  Let us briefly justify this choice.  First, recall that Pauli measurement on a single qubit for Pauli $P \in \Paulis_1$ projects the state onto either the $+1$ or $-1$ eigenspace of $P$.  The key observation is that the net result of the measurement can always be described by applying a Pauli operator, projecting the qubit onto the $+1$ eigenspace, and applying another Pauli operator.  There are two cases:  if the Pauli measurement would have resulted in a projection into the $+1$ eigenspace, then do nothing; if it would have projected into the $-1$ eigenspace, then apply Pauli $Q$ such that $Q \neq P$, project into the $+1$ eigenspace, and apply Pauli $Q$ again.  This is justified by the following equation:
$$
\frac{\pfont{I} - P}{2} = Q \( \frac{\pfont{I} + P}{2} \) Q.
$$
To be clear, in the following $ZX$-calculus diagrams we will make no explicit mention of the Pauli $Q$ operation mentioned above, which is required to accurately describe the state if the measurement projects the state onto the $-1$ eigenspace.  However, using the $\pi$-copy rule from \Cref{tab:zx_rules}, these phases can easily\footnote{On the other hand, such calculations will not be ``easy'' for sufficiently weak machines, which in some sense is the basis for the hardness results of sections \ref{sec:nc1} and \ref{sec:parityL}.} be pushed to the end of the circuit.  These phases determine the Pauli correction operation which appears in \Cref{thm:rbb_mbqc} (and, in fact, all measurement-based computation schemes).

With this out of the way, let us discuss measurements in the context of the ZX-calculus.  To measure a qubit in the $X$-basis, we will apply the green projection generator in \Cref{tab:zx_generators} to the corresponding qubit.  To measure in the $Y$-basis, we first apply a $\pi/2$-phase gate to the corresponding qubit, and then project onto the $X$-basis.  Since
\begin{center}
\begin{tikzpicture}
	\begin{pgfonlayer}{nodelayer}
		\node [style=Z dot] (0) at (1, 0) {};
		\node [style=Z phase dot] (1) at (0.25, 0) {$\frac{\pi}{2}$};
		\node [style=none] (2) at (-0.5, 0) {};
		\node [style=none] (4) at (1.75, 0) {$=$};
		\node [style=none] (6) at (2.5, 0) {};
		\node [style=Z phase dot] (7) at (3.5, 0) {$\frac{\pi}{2}$};
	\end{pgfonlayer}
	\begin{pgfonlayer}{edgelayer}
		\draw (2.center) to (1);
		\draw (1) to (0);
		\draw (6.center) to (7);
	\end{pgfonlayer}
\end{tikzpicture}
\end{center}
by the spider fusion rule, we get that each measurement corresponds to a $Z$-projection with $\alpha = 0$ for $X$ measurements, and a $Z$-projection with $\alpha = \pi/2$ for $Y$ measurements. 

\section{Measurement-Based Quantum Computation Gadgets}
\label{app:mbqc_gadgets}

The goal of this appendix is to construct explicit graph states and show how measuring them can \emph{efficiently} implement desired operations. That is, we describe \emph{measurement-based computation gadgets} (or simply \emph{measurement gadgets}) which implement various Clifford operations under appropriate measurements. 

\begin{definition}
	A \emph{measurement-based computation gadget} (or just \emph{measurement gadget} in this paper) is a graph state $\ket{G}$ with $m$ distinguished \emph{input} qubits and $m$ distinguished \emph{output} qubits. To use the gadget, we connect the graph state $\ket{G}$ to an \emph{input register} (not to be confused with the input qubits) $\ket{\psi}$ and an \emph{output register} (also not to be confused with the output qubits), initially in state $\ket{+}^{\otimes m}$, with $\CSIGN$ gates. Specifically, starting with state $\ket{\psi} \otimes \ket{G} \otimes \ket{+}^{\otimes m}$, we apply $\CSIGN$ gates from the $i$th qubit of the input state to the $i$th input qubit, and from the $i$th output qubit to the $i$th qubit of the output register. We measure the qubits of $\ket{\psi}$ in the $X$-basis and each qubit of $\ket{G}$ in the $X$ or $Y$ basis. The output register contains some state $P U \ket{\psi}$, where $U$ is a unitary depending on the choice of measurement bases, and $P$ is a Pauli operation depending on the measurement outcomes (and bases). 
\end{definition}

Functionally, a measurement gadget implements one of a family of unitaries, depending on the choice of measurement bases. Naturally, we expect to be able to compose measurement gadgets as one would compose gates.  For instance, if we put two measurement gadgets beside each other (the disjoint union of the graphs), this gives a gadget which implements $U_1 \otimes U_2$ for any unitaries $U_1, U_2$ implemented by the original gadgets. Similarly, we can compose measurement gadgets sequentially (assuming they have the same number of input/output qubits) as follows, by a procedure of Raussendorf, Browne, and Briegel.

\begin{lemma}[Raussendorf, Browne, and Briegel \cite{rbb:2003_mbc}]
	\label{lem:mbqc_composition}
	Let $\mathcal G_1$ and $\mathcal G_2$ be two measurement gadgets. Let $\mathcal G$ be the gadget with $\mathcal G_1$, followed by a layer of $m$ \emph{buffer qubits}, followed by $\mathcal G_2$, and connect the $i$th buffer qubit to the $i$th output qubit of $\mathcal G_1$ and the $i$th input qubit of $\mathcal G_2$. The input qubits of $\mathcal{G}_1$ and the output qubits of $\mathcal{G}_2$ become the input and output qubits of $\mathcal{G}$. 
	
	For any unitary $U_1$ implemented by $\mathcal{G}_1$ with measurements $M_1$, and any unitary $U_2$ implemented by $\mathcal{G}_2$ with measurements $M_2$, we can implement $U_2 U_1$ with $\mathcal{G}$ by measuring $M_1$ on the vertices from $\mathcal{G}_1$, $M_2$ on the vertices from $\mathcal{G}_2$, and $X$ on the new vertices. 
\end{lemma}
\begin{proof}[Sketch]
	This appears to be trivial: the output bits of $\mathcal G_1$ are literally the same as the input bits for $\mathcal G_2$. However, when we execute $\mathcal G$, we apply the $\CSIGN$ gates from the buffer qubits to the input of $\mathcal G_2$ at the beginning rather than after $\mathcal G_1$ has executed. As discussed previously, the $\CSIGN$ gates commute with measurement in $\mathcal G_1$, so we can imagine they occur after gadget $\mathcal G_1$ has executed and placed an intermediate quantum state in the buffer qubits.

From the ZX-calculus, it's easy to see that the $X$ measurements on the buffer qubits essentially create a wire between the two gadgets.  Consider a single output qubit, followed by a buffer qubit, followed by an input qubit:
	
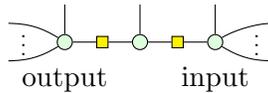
\begin{center} \begin{tikzpicture}
	\begin{pgfonlayer}{nodelayer}
		\node [style=Z dot] (0) at (0, 0) {};
		\node [style=none] (3) at (0, 0.5) {};
		\node [style=none] (4) at (-1, -0.5) {output};
		\node [style=none] (5) at (1, -0.5) {input};
		\node [style=none] (6) at (-1, 0.5) {};
		\node [style=none] (7) at (1, 0.5) {};
		\node [style=none] (8) at (-1.75, 0.25) {};
		\node [style=none] (9) at (1.75, 0.25) {};
		\node [style=Z dot] (10) at (-1, 0) {};
		\node [style=Z dot] (11) at (1, 0) {};
		\node [style=hadamard] (12) at (0.5, 0) {};
		\node [style=hadamard] (13) at (-0.5, 0) {};
		\node [style=none] (14) at (-1.75, -0.25) {};
		\node [style=none] (15) at (1.75, -0.25) {};
		\node [style=none] (16) at (-1.55, 0.05) {\scriptsize$\vdots$};
		\node [style=none] (17) at (1.55, 0.05) {\scriptsize$\vdots$};
	\end{pgfonlayer}
	\begin{pgfonlayer}{edgelayer}
		\draw (3.center) to (0);
		\draw (6.center) to (10);
		\draw (10) to (13);
		\draw (13) to (0);
		\draw (0) to (12);
		\draw (12) to (11);
		\draw (11) to (7.center);
		\draw [bend left=15, looseness=0.75] (11) to (9.center);
		\draw [bend left=15] (8.center) to (10);
		\draw [bend right=15] (14.center) to (10);
		\draw [bend right=345] (15.center) to (11);
	\end{pgfonlayer}
\end{tikzpicture} \end{center}
	Measuring the new vertex in the $X$ basis teleports the output qubit of the first gadget to the input qubit of the next.
	
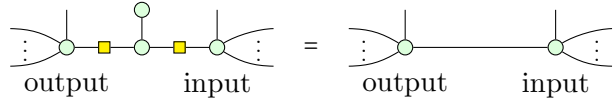
\begin{center} \begin{tikzpicture}
	\begin{pgfonlayer}{nodelayer}
		\node [style=Z dot] (0) at (0, 0) {};
		\node [style=none] (4) at (-1, -0.5) {output};
		\node [style=none] (5) at (1, -0.5) {input};
		\node [style=none] (6) at (-1, 0.5) {};
		\node [style=none] (7) at (1, 0.5) {};
		\node [style=none] (8) at (-1.75, 0.25) {};
		\node [style=none] (9) at (1.75, 0.25) {};
		\node [style=Z dot] (10) at (-1, 0) {};
		\node [style=Z dot] (11) at (1, 0) {};
		\node [style=hadamard] (12) at (0.5, 0) {};
		\node [style=hadamard] (13) at (-0.5, 0) {};
		\node [style=none] (14) at (-1.75, -0.25) {};
		\node [style=none] (15) at (1.75, -0.25) {};
		\node [style=none] (16) at (-1.55, 0.05) {\scriptsize$\vdots$};
		\node [style=none] (17) at (1.55, 0.05) {\scriptsize$\vdots$};
		\node [style=Z dot] (18) at (0, 0.5) {};
		\node [style=none] (19) at (2.25, 0) {$=$};
		\node [style=none] (21) at (3.5, -0.5) {output};
		\node [style=none] (22) at (5.5, -0.5) {input};
		\node [style=none] (23) at (3.5, 0.5) {};
		\node [style=none] (24) at (5.5, 0.5) {};
		\node [style=none] (25) at (2.75, 0.25) {};
		\node [style=none] (26) at (6.25, 0.25) {};
		\node [style=Z dot] (27) at (3.5, 0) {};
		\node [style=Z dot] (28) at (5.5, 0) {};
		\node [style=none] (31) at (2.75, -0.25) {};
		\node [style=none] (32) at (6.25, -0.25) {};
		\node [style=none] (33) at (2.95, 0.05) {\scriptsize$\vdots$};
		\node [style=none] (34) at (6.05, 0.05) {\scriptsize$\vdots$};
	\end{pgfonlayer}
	\begin{pgfonlayer}{edgelayer}
		\draw (6.center) to (10);
		\draw (10) to (13);
		\draw (13) to (0);
		\draw (0) to (12);
		\draw (12) to (11);
		\draw (11) to (7.center);
		\draw [bend left=15, looseness=0.75] (11) to (9.center);
		\draw [bend left=15] (8.center) to (10);
		\draw [bend right=15] (14.center) to (10);
		\draw [bend right=345] (15.center) to (11);
		\draw (18) to (0);
		\draw (23.center) to (27);
		\draw (28) to (24.center);
		\draw [bend left=15, looseness=0.75] (28) to (26.center);
		\draw [bend left=15] (25.center) to (27);
		\draw [bend right=15] (31.center) to (27);
		\draw [bend right=345] (32.center) to (28);
		\draw (27) to (28);
	\end{pgfonlayer}
\end{tikzpicture} \end{center}
\end{proof}

\subsection{Single-qubit Clifford}
This section will be devoted to proving the following theorem:
\begin{theorem}
	\label{thm:single_qubit_mbqc}
	Let $G$ be the line graph on $3n+2$ vertices.  Given single-qubit Clifford gates $g_1, \ldots, g_n \in \mathcal C^1$, there exist projections $P_1, \ldots, P_{3n+1}$ such that 
	\begin{itemize}[itemsep = 0pt]
		\item $P_i$ is either the projection $\frac{\pfont{I} + \pfont{X}}{2}$ or $\frac{\pfont{I} + \pfont{Y}}{2}$ onto the $i$th qubit,
		\item $P_1 \otimes \ldots \otimes P_{3n+1} \otimes \pfont{I} \ket{G} \propto \ket{\psi} \otimes \prod_{i=1}^n g_i \ket{+}$, 
		\item $P_1$ is an $X$ projection, and 
		\item For $i \in [n]$, $P_{3i-1}$, $P_{3i}$, and $P_{3i+1}$ depend only on $g_i \in \mathcal C_1$ and whether or not $i$ is equal to $n$.
	\end{itemize}
\end{theorem}

It is a standard result that every single-qubit Clifford operation can be decomposed into $\RX(\theta_3) \RZ(\theta_2) \RX(\theta_1)$ for some $\theta_i \in \{0,\pi/2, \pi, 3\pi/2\}$.  Three single-qubit $X$ and $Y$ measurements on the line suffice to mimic such a decomposition using measurement-based computation (see e.g., \cite{backens:2014_zx}).  We show this relationship in \Cref{fig:single_qubit_mbqc}.

\begin{figure}[ht!]
	\begin{center}
		\begin{quantikz}[row sep=1em, column sep=1em, thin lines]
			\lstick{$\ket{\psi^{\operatorname{in}}}$} & \ctrl{1} & \gate{H} & \qw & \meter{} \\
			\lstick[wires =4] {$\ket{G}$} & \control{} & \gate{H} & \gate{\RZ(\theta_1)} &  \meter{}  \\
			& \qw & \gate{H} &  \gate{\RZ(\theta_2)} &  \meter{}  \\
			& \qw & \gate{H} &  \gate{\RZ(\theta_3)} &  \meter{}  \\
			& \qw & \qw & \qw & \qw \rstick{$\ket{\psi^{\operatorname{out}}}$}
		\end{quantikz}
		$\stackrel{\Paulis_1}{=}$ \begin{quantikz}[row sep=1em, column sep=1em, thin lines]
			\lstick{$\ket{\psi^{\operatorname{in}}}$} & \gate{\RX(\theta_1)} & \gate{\RZ(\theta_2)} & \gate{\RX(\theta_3)} & \qw \rstick{$\ket{\psi^{\operatorname{out}}}$}
		\end{quantikz}
	\end{center}
	\caption{Generating any single-qubit Clifford by measuring graph state for line graph $G$ on 4 vertices.  The choice of basis in the middle three qubits, determines the gate applied to input $\ket{\psi^{\operatorname{in}}}$ up to a Pauli correction.}
	\label{fig:single_qubit_mbqc}
\end{figure}
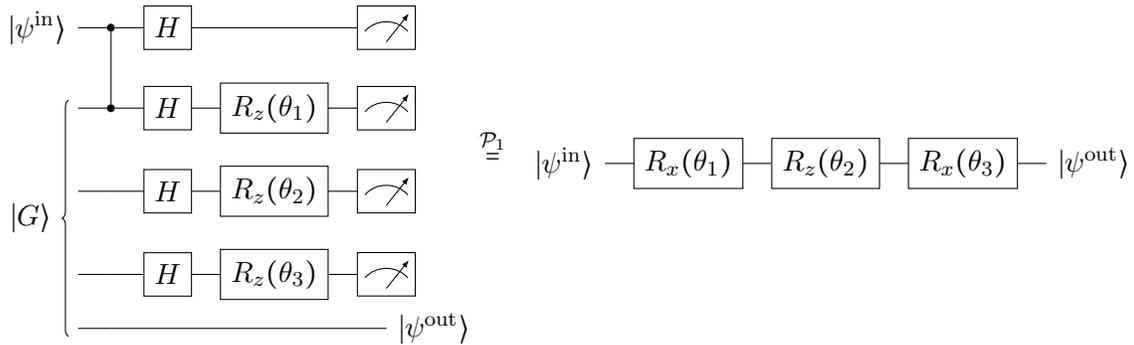

Let us prove this statement using the ZX-calculus.  By \Cref{lem:zx_graph_state}, we can represent the 5-qubit line graph state as the following diagram in the ZX-calculus:
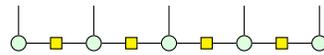
\begin{center}
	\begin{tikzpicture}
	\begin{pgfonlayer}{nodelayer}
		\node [style=Z dot] (0) at (-2, 0) {};
		\node [style=hadamard] (1) at (-1.5, 0) {};
		\node [style=Z dot] (2) at (-1, 0) {};
		\node [style=Z dot] (3) at (0, 0) {};
		\node [style=Z dot] (4) at (1, 0) {};
		\node [style=Z dot] (5) at (2, 0) {};
		\node [style=hadamard] (6) at (-0.5, 0) {};
		\node [style=hadamard] (7) at (0.5, 0) {};
		\node [style=hadamard] (8) at (1.5, 0) {};
		\node [style=none] (9) at (-2, 0.5) {};
		\node [style=none] (10) at (-1, 0.5) {};
		\node [style=none] (11) at (0, 0.5) {};
		\node [style=none] (12) at (1, 0.5) {};
		\node [style=none] (13) at (2, 0.5) {};
	\end{pgfonlayer}
	\begin{pgfonlayer}{edgelayer}
		\draw (0) to (1);
		\draw (1) to (2);
		\draw (2) to (6);
		\draw (6) to (3);
		\draw (3) to (7);
		\draw (7) to (4);
		\draw (4) to (8);
		\draw (8) to (5);
		\draw (0) to (9.center);
		\draw (2) to (10.center);
		\draw (3) to (11.center);
		\draw (4) to (12.center);
		\draw (5) to (13.center);
	\end{pgfonlayer}
\end{tikzpicture}
\end{center}
Recall that the generator for measurement is \begin{tikzpicture}
	\begin{pgfonlayer}{nodelayer}
		\node [style=none] (1) at (-1, 0) {};
		\node [style=Z phase dot] (2) at (0, 0) {$\alpha$};
	\end{pgfonlayer}
	\begin{pgfonlayer}{edgelayer}
		\draw (2) to (1.center);
	\end{pgfonlayer}
\end{tikzpicture}
, where $\alpha = 0$ for $X$ measurement and $\alpha = \pi/2$ for $Y$ measurement. Using this fact and the color changing rule from \Cref{tab:zx_rules}, we get
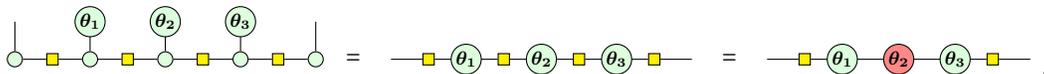
\begin{center}
	\begin{tikzpicture}
	\begin{pgfonlayer}{nodelayer}
		\node [style=Z dot] (0) at (-2, 0) {};
		\node [style=hadamard] (1) at (-1.5, 0) {};
		\node [style=Z dot] (2) at (-1, 0) {};
		\node [style=Z dot] (3) at (0, 0) {};
		\node [style=Z dot] (4) at (1, 0) {};
		\node [style=Z dot] (5) at (2, 0) {};
		\node [style=hadamard] (6) at (-0.5, 0) {};
		\node [style=hadamard] (7) at (0.5, 0) {};
		\node [style=hadamard] (8) at (1.5, 0) {};
		\node [style=none] (9) at (-2, 0.5) {};
		\node [style=none] (13) at (2, 0.5) {};
		\node [style=none] (15) at (2.5, 0) {$=$};
		\node [style=Z phase dot] (16) at (-1, 0.5) {$\theta_1$};
		\node [style=Z phase dot] (17) at (0, 0.5) {$\theta_2$};
		\node [style=Z phase dot] (18) at (1, 0.5) {$\theta_3$};
		\node [style=hadamard] (20) at (3.5, 0) {};
		\node [style=hadamard] (25) at (4.5, 0) {};
		\node [style=hadamard] (26) at (5.5, 0) {};
		\node [style=hadamard] (27) at (6.5, 0) {};
		\node [style=Z phase dot] (30) at (4, 0) {$\theta_1$};
		\node [style=Z phase dot] (31) at (5, 0) {$\theta_2$};
		\node [style=Z phase dot] (32) at (6, 0) {$\theta_3$};
		\node [style=none] (33) at (7.5, 0) {$=$};
		\node [style=none] (41) at (3, 0) {};
		\node [style=none] (42) at (7, 0) {};
		\node [style=hadamard] (43) at (8.5, 0) {};
		\node [style=hadamard] (46) at (11, 0) {};
		\node [style=Z phase dot] (47) at (9, 0) {$\theta_1$};
		\node [style=X phase dot] (48) at (9.75, 0) {$\theta_2$};
		\node [style=Z phase dot] (49) at (10.5, 0) {$\theta_3$};
		\node [style=none] (50) at (8, 0) {};
		\node [style=none] (51) at (11.5, 0) {};
	\end{pgfonlayer}
	\begin{pgfonlayer}{edgelayer}
		\draw (0) to (1);
		\draw (1) to (2);
		\draw (2) to (6);
		\draw (6) to (3);
		\draw (3) to (7);
		\draw (7) to (4);
		\draw (4) to (8);
		\draw (8) to (5);
		\draw (0) to (9.center);
		\draw (5) to (13.center);
		\draw (16) to (2);
		\draw (17) to (3);
		\draw (18) to (4);
		\draw (20) to (30);
		\draw (30) to (25);
		\draw (25) to (31);
		\draw (31) to (26);
		\draw (26) to (32);
		\draw (32) to (27);
		\draw (41.center) to (20);
		\draw (27) to (42.center);
		\draw (43) to (47);
		\draw (49) to (46);
		\draw (50.center) to (43);
		\draw (46) to (51.center);
		\draw (47) to (48);
		\draw (48) to (49);
	\end{pgfonlayer}
\end{tikzpicture}.
\end{center}

That is, the line decomposes into the gate $H \RZ(\theta_3) \RX(\theta_2) \RZ(\theta_1) H$.  To apply more than one single-qubit gate, we simply create groups of three qubits. For example, the following diagram shows the application of gates $g_1 = \RZ(\alpha_3) \RX(\alpha_2) \RZ(\alpha_1) H$, $g_2 = \RZ(\beta_3) \RX(\beta_2) \RZ(\beta_1) H$, and $g_3 = H \RZ(\gamma_3) \RX(\gamma_2) \RZ(\gamma_1) H$ using a chain of length 11. 
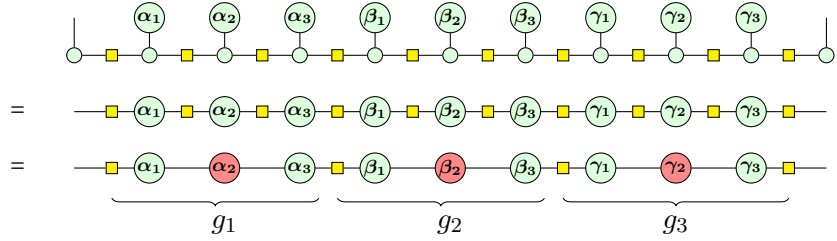
\begin{center}
	\begin{tikzpicture}
	\begin{pgfonlayer}{nodelayer}
		\node [style=Z dot] (0) at (-2, -0.25) {};
		\node [style=hadamard] (1) at (-1.5, -0.25) {};
		\node [style=Z dot] (2) at (-1, -0.25) {};
		\node [style=Z dot] (3) at (0, -0.25) {};
		\node [style=Z dot] (4) at (1, -0.25) {};
		\node [style=hadamard] (6) at (-0.5, -0.25) {};
		\node [style=hadamard] (7) at (0.5, -0.25) {};
		\node [style=hadamard] (8) at (1.5, -0.25) {};
		\node [style=none] (9) at (-2, 0.25) {};
		\node [style=none] (15) at (-2.75, -1) {$=$};
		\node [style=Z phase dot] (16) at (-1, 0.25) {$\alpha_1$};
		\node [style=Z phase dot] (17) at (0, 0.25) {$\alpha_2$};
		\node [style=Z phase dot] (18) at (1, 0.25) {$\alpha_3$};
		\node [style=Z dot] (21) at (2, -0.25) {};
		\node [style=Z dot] (22) at (3, -0.25) {};
		\node [style=Z dot] (23) at (4, -0.25) {};
		\node [style=hadamard] (25) at (2.5, -0.25) {};
		\node [style=hadamard] (26) at (3.5, -0.25) {};
		\node [style=hadamard] (27) at (4.5, -0.25) {};
		\node [style=Z phase dot] (30) at (2, 0.25) {$\beta_1$};
		\node [style=Z phase dot] (31) at (3, 0.25) {$\beta_2$};
		\node [style=Z phase dot] (32) at (4, 0.25) {$\beta_3$};
		\node [style=Z dot] (35) at (5, -0.25) {};
		\node [style=Z dot] (36) at (6, -0.25) {};
		\node [style=Z dot] (37) at (7, -0.25) {};
		\node [style=Z dot] (38) at (8, -0.25) {};
		\node [style=hadamard] (39) at (5.5, -0.25) {};
		\node [style=hadamard] (40) at (6.5, -0.25) {};
		\node [style=hadamard] (41) at (7.5, -0.25) {};
		\node [style=none] (43) at (8, 0.25) {};
		\node [style=Z phase dot] (44) at (5, 0.25) {$\gamma_1$};
		\node [style=Z phase dot] (45) at (6, 0.25) {$\gamma_2$};
		\node [style=Z phase dot] (46) at (7, 0.25) {$\gamma_3$};
		\node [style=hadamard] (48) at (-1.5, -1) {};
		\node [style=hadamard] (52) at (-0.5, -1) {};
		\node [style=hadamard] (53) at (0.5, -1) {};
		\node [style=hadamard] (54) at (1.5, -1) {};
		\node [style=none] (55) at (-2, -1) {};
		\node [style=Z phase dot] (56) at (-1, -1) {$\alpha_1$};
		\node [style=Z phase dot] (57) at (0, -1) {$\alpha_2$};
		\node [style=Z phase dot] (58) at (1, -1) {$\alpha_3$};
		\node [style=hadamard] (62) at (2.5, -1) {};
		\node [style=hadamard] (63) at (3.5, -1) {};
		\node [style=hadamard] (64) at (4.5, -1) {};
		\node [style=Z phase dot] (65) at (2, -1) {$\beta_1$};
		\node [style=Z phase dot] (66) at (3, -1) {$\beta_2$};
		\node [style=Z phase dot] (67) at (4, -1) {$\beta_3$};
		\node [style=hadamard] (72) at (5.5, -1) {};
		\node [style=hadamard] (73) at (6.5, -1) {};
		\node [style=hadamard] (74) at (7.5, -1) {};
		\node [style=none] (75) at (8, -1) {};
		\node [style=Z phase dot] (76) at (5, -1) {$\gamma_1$};
		\node [style=Z phase dot] (77) at (6, -1) {$\gamma_2$};
		\node [style=Z phase dot] (78) at (7, -1) {$\gamma_3$};
		\node [style=none] (79) at (-2.75, -1.75) {$=$};
		\node [style=none] (105) at (-1.5, -2.25) {};
		\node [style=none] (106) at (1.25, -2.25) {};
		\node [style=none] (107) at (1.5, -2.25) {};
		\node [style=none] (108) at (4.25, -2.25) {};
		\node [style=none] (109) at (4.5, -2.25) {};
		\node [style=none] (110) at (7.5, -2.25) {};
		\node [style=none] (111) at (0, -2.5) {$g_1$};
		\node [style=none] (112) at (3, -2.5) {$g_2$};
		\node [style=none] (113) at (6, -2.5) {$g_3$};
		\node [style=hadamard] (114) at (-1.5, -1.75) {};
		\node [style=hadamard] (117) at (1.5, -1.75) {};
		\node [style=none] (118) at (-2, -1.75) {};
		\node [style=Z phase dot] (119) at (-1, -1.75) {$\alpha_1$};
		\node [style=X phase dot] (120) at (0, -1.75) {$\alpha_2$};
		\node [style=Z phase dot] (121) at (1, -1.75) {$\alpha_3$};
		\node [style=hadamard] (124) at (4.5, -1.75) {};
		\node [style=Z phase dot] (125) at (2, -1.75) {$\beta_1$};
		\node [style=X phase dot] (126) at (3, -1.75) {$\beta_2$};
		\node [style=Z phase dot] (127) at (4, -1.75) {$\beta_3$};
		\node [style=hadamard] (130) at (7.5, -1.75) {};
		\node [style=none] (131) at (8, -1.75) {};
		\node [style=Z phase dot] (132) at (5, -1.75) {$\gamma_1$};
		\node [style=X phase dot] (133) at (6, -1.75) {$\gamma_2$};
		\node [style=Z phase dot] (134) at (7, -1.75) {$\gamma_3$};
	\end{pgfonlayer}
	\begin{pgfonlayer}{edgelayer}
		\draw (0) to (1);
		\draw (1) to (2);
		\draw (2) to (6);
		\draw (6) to (3);
		\draw (3) to (7);
		\draw (7) to (4);
		\draw (4) to (8);
		\draw (0) to (9.center);
		\draw (16) to (2);
		\draw (17) to (3);
		\draw (18) to (4);
		\draw (21) to (25);
		\draw (25) to (22);
		\draw (22) to (26);
		\draw (26) to (23);
		\draw (23) to (27);
		\draw (30) to (21);
		\draw (31) to (22);
		\draw (32) to (23);
		\draw (35) to (39);
		\draw (39) to (36);
		\draw (36) to (40);
		\draw (40) to (37);
		\draw (37) to (41);
		\draw (41) to (38);
		\draw (38) to (43.center);
		\draw (44) to (35);
		\draw (45) to (36);
		\draw (46) to (37);
		\draw (8) to (21);
		\draw (35) to (27);
		\draw (55.center) to (48);
		\draw (48) to (56);
		\draw (56) to (52);
		\draw (52) to (57);
		\draw (57) to (53);
		\draw (53) to (58);
		\draw (58) to (54);
		\draw (54) to (65);
		\draw (65) to (62);
		\draw (62) to (66);
		\draw (66) to (63);
		\draw (63) to (67);
		\draw (67) to (64);
		\draw (64) to (76);
		\draw (76) to (72);
		\draw (72) to (77);
		\draw (77) to (73);
		\draw (73) to (78);
		\draw (78) to (74);
		\draw (74) to (75.center);
		\draw [style=brace edge] (110.center) to (109.center);
		\draw [style=brace edge] (108.center) to (107.center);
		\draw [style=brace edge] (106.center) to (105.center);
		\draw (118.center) to (114);
		\draw (114) to (119);
		\draw (121) to (117);
		\draw (117) to (125);
		\draw (127) to (124);
		\draw (124) to (132);
		\draw (134) to (130);
		\draw (130) to (131.center);
		\draw (119) to (120);
		\draw (120) to (121);
		\draw (125) to (126);
		\draw (126) to (127);
		\draw (132) to (133);
		\draw (133) to (134);
	\end{pgfonlayer}
\end{tikzpicture}.
\end{center}
The generalization is clear, so this concludes the proof of \Cref{thm:single_qubit_mbqc}.

\subsection{Two-qubit Clifford}
\label{sec:two_qubit_mbqc}

The analysis for the two-qubit case will be very similar to the single-qubit case, except that we will obviously require a more sophisticated initial graph state.  The two-qubit analogue to the line graph will be concatenations of the \scalebox{.3}{
\begin{tikzpicture}[dot/.style={circle, fill=black, draw, inner sep=1pt}]
    \node[dot] (3) [right of=1] {};
    \node[dot] (4) [below of=3] {};
    \node[dot] (5) [right of=3] {};
    \node[dot] (6) [below of=5] {};

    \path
    (3) edge (4)
    	edge (5)
    (4) edge (6);
\end{tikzpicture}
} subgraph.  More precisely, we define the graph $\mathcal H_n$ to be the graph on $16n+4$ vertices, which has $4n$ instances of that subgraph along with 2 input qubits and 2 output qubits:
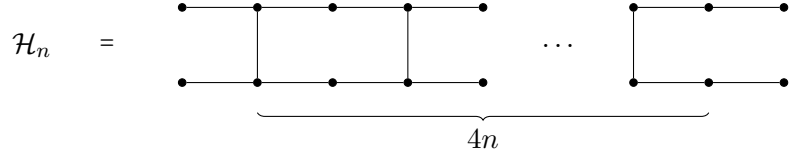
\begin{center}
	\begin{tikzpicture}
	\begin{pgfonlayer}{nodelayer}
		\node [style=vertex] (0) at (0, 0) {};
		\node [style=vertex] (1) at (0, -1) {};
		\node [style=vertex] (2) at (1, 0) {};
		\node [style=vertex] (3) at (1, -1) {};
		\node [style=vertex] (4) at (2, 0) {};
		\node [style=vertex] (5) at (2, -1) {};
		\node [style=none] (6) at (-2, -0.5) {$\mathcal H_n$};
		\node [style=none] (7) at (-1, -0.5) {$=$};
		\node [style=vertex] (8) at (3, 0) {};
		\node [style=vertex] (9) at (3, -1) {};
		\node [style=vertex] (10) at (4, 0) {};
		\node [style=vertex] (11) at (4, -1) {};
		\node [style=none] (12) at (5, -0.5) {$\ldots$};
		\node [style=vertex] (13) at (6, 0) {};
		\node [style=vertex] (14) at (6, -1) {};
		\node [style=vertex] (15) at (7, 0) {};
		\node [style=vertex] (16) at (7, -1) {};
		\node [style=none] (17) at (1, -1.5) {};
		\node [style=none] (18) at (7, -1.5) {};
		\node [style=none] (19) at (4, -1.75) {$4n$};
		\node [style=vertex] (20) at (8, 0) {};
		\node [style=vertex] (21) at (8, -1) {};
	\end{pgfonlayer}
	\begin{pgfonlayer}{edgelayer}
		\draw (0) to (2);
		\draw (2) to (4);
		\draw (2) to (3);
		\draw (3) to (1);
		\draw (3) to (5);
		\draw (4) to (8);
		\draw (5) to (9);
		\draw (8) to (9);
		\draw (8) to (10);
		\draw (9) to (11);
		\draw (15) to (13);
		\draw (13) to (14);
		\draw (14) to (16);
		\draw [style=brace edge] (18.center) to (17.center);
		\draw (15) to (20);
		\draw (16) to (21);
	\end{pgfonlayer}
\end{tikzpicture}
\end{center}
Using measurement-based computation, we will be able to apply one 2-qubit gate using the graph state $\ket{\mathcal H_1}$, two 2-qubit gates using the graph state $\ket{\mathcal H_2}$, and so on.  We formalize this in the theorem below.  For convenience, we number the qubits of $\mathcal H_n$ from top to bottom, and then from left to right.  
\begin{theorem}
	\label{thm:two_qubit_mbqc}
	Given two-qubit Clifford gates $g_1, \ldots, g_n \in \mathcal C_2$, there exist projections $P_1, \ldots, P_{16n+2}$ such that 
	\begin{itemize}[itemsep = 0pt]
		\item $P_i$ is either the projection $\frac{\pfont{I} + \pfont{X}}{2}$ or $\frac{\pfont{I} + \pfont{Y}}{2}$ onto the $i$th qubit,
		\item $P_1 \otimes \ldots \otimes P_{16n+2} \otimes \pfont{II} \ket{\mathcal H_n} \propto \ket{\psi} \otimes \prod_{i=1}^n g_i \ket{++}$,
		\item $P_1$ and $P_2$ are $X$ projections, and 
		\item For $i \in [n]$, $P_{16i -13}, \ldots, P_{16i+2}$ depend only on $g_i \in \mathcal C_2$ and whether or not $i$ is equal to $n$.
	\end{itemize}
\end{theorem}

Let's start by using the ZX-calculus to represent the graph state with a single \scalebox{.3}
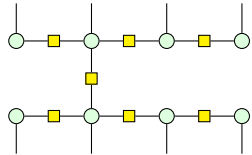
{
\begin{tikzpicture}[dot/.style={circle, fill=black, draw, inner sep=1pt}]
    \node[dot] (3) [right of=1] {};
    \node[dot] (4) [below of=3] {};
    \node[dot] (5) [right of=3] {};
    \node[dot] (6) [below of=5] {};

    \path
    (3) edge (4)
    	edge (5)
    (4) edge (6);
\end{tikzpicture}
} subgraph: 
\begin{center}
	\begin{tikzpicture}
	\begin{pgfonlayer}{nodelayer}
		\node [style=Z dot] (0) at (0, 0) {};
		\node [style=Z dot] (1) at (0, -1) {};
		\node [style=Z dot] (2) at (1, 0) {};
		\node [style=Z dot] (3) at (1, -1) {};
		\node [style=Z dot] (4) at (2, 0) {};
		\node [style=Z dot] (5) at (2, -1) {};
		\node [style=none] (8) at (0, 0.5) {};
		\node [style=none] (9) at (1, 0.5) {};
		\node [style=none] (10) at (2, 0.5) {};
		\node [style=none] (16) at (0, -1.5) {};
		\node [style=none] (17) at (1, -1.5) {};
		\node [style=none] (20) at (2, -1.5) {};
		\node [style=hadamard] (24) at (0.5, 0) {};
		\node [style=hadamard] (25) at (0.5, -1) {};
		\node [style=hadamard] (26) at (1.5, -1) {};
		\node [style=hadamard] (27) at (1.5, 0) {};
		\node [style=hadamard] (29) at (1, -0.5) {};
		\node [style=Z dot] (30) at (3, 0) {};
		\node [style=Z dot] (31) at (3, -1) {};
		\node [style=hadamard] (32) at (2.5, 0) {};
		\node [style=hadamard] (33) at (2.5, -1) {};
		\node [style=none] (34) at (3, 0.5) {};
		\node [style=none] (35) at (3, -1.5) {};
	\end{pgfonlayer}
	\begin{pgfonlayer}{edgelayer}
		\draw (8.center) to (0);
		\draw (0) to (24);
		\draw (24) to (2);
		\draw (2) to (9.center);
		\draw (2) to (27);
		\draw (27) to (4);
		\draw (4) to (10.center);
		\draw (2) to (29);
		\draw (29) to (3);
		\draw (3) to (25);
		\draw (25) to (1);
		\draw (1) to (16.center);
		\draw (3) to (17.center);
		\draw (3) to (26);
		\draw (26) to (5);
		\draw (5) to (20.center);
		\draw (4) to (32);
		\draw (32) to (30);
		\draw (5) to (33);
		\draw (33) to (31);
		\draw (34.center) to (30);
		\draw (31) to (35.center);
	\end{pgfonlayer}
\end{tikzpicture}
\end{center}
Using the same ideas as in the previous section, we get the following after measurements:
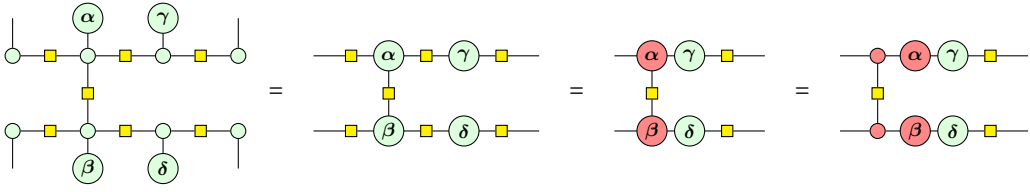
\begin{center}
	\begin{tikzpicture}
	\begin{pgfonlayer}{nodelayer}
		\node [style=Z dot] (0) at (0.25, 0) {};
		\node [style=Z dot] (1) at (0.25, -1) {};
		\node [style=Z dot] (2) at (1.25, 0) {};
		\node [style=Z dot] (3) at (1.25, -1) {};
		\node [style=Z dot] (4) at (2.25, 0) {};
		\node [style=Z dot] (5) at (2.25, -1) {};
		\node [style=none] (8) at (0.25, 0.5) {};
		\node [style=none] (16) at (0.25, -1.5) {};
		\node [style=hadamard] (24) at (0.75, 0) {};
		\node [style=hadamard] (25) at (0.75, -1) {};
		\node [style=hadamard] (26) at (1.75, -1) {};
		\node [style=hadamard] (27) at (1.75, 0) {};
		\node [style=hadamard] (29) at (1.25, -0.5) {};
		\node [style=Z phase dot] (50) at (1.25, 0.5) {$\alpha$};
		\node [style=Z phase dot] (51) at (1.25, -1.5) {$\beta$};
		\node [style=hadamard] (79) at (2.75, 0) {};
		\node [style=hadamard] (80) at (2.75, -1) {};
		\node [style=none] (81) at (3.25, 0.5) {};
		\node [style=none] (82) at (3.25, -1.5) {};
		\node [style=Z dot] (83) at (3.25, 0) {};
		\node [style=Z dot] (84) at (3.25, -1) {};
		\node [style=Z phase dot] (85) at (2.25, 0.5) {$\gamma$};
		\node [style=Z phase dot] (86) at (2.25, -1.5) {$\delta$};
		\node [style=none] (93) at (4.25, 0) {};
		\node [style=none] (94) at (4.25, -1) {};
		\node [style=hadamard] (95) at (4.75, 0) {};
		\node [style=hadamard] (96) at (4.75, -1) {};
		\node [style=hadamard] (97) at (5.75, -1) {};
		\node [style=hadamard] (98) at (5.75, 0) {};
		\node [style=hadamard] (99) at (5.25, -0.5) {};
		\node [style=Z phase dot] (102) at (5.25, 0) {$\alpha$};
		\node [style=Z phase dot] (103) at (5.25, -1) {$\beta$};
		\node [style=hadamard] (104) at (6.75, 0) {};
		\node [style=hadamard] (105) at (6.75, -1) {};
		\node [style=none] (106) at (7.25, 0) {};
		\node [style=none] (107) at (7.25, -1) {};
		\node [style=Z phase dot] (110) at (6.25, 0) {$\gamma$};
		\node [style=Z phase dot] (111) at (6.25, -1) {$\delta$};
		\node [style=none] (112) at (3.75, -0.5) {$=$};
		\node [style=none] (113) at (8.25, 0) {};
		\node [style=none] (114) at (8.25, -1) {};
		\node [style=X phase dot] (120) at (8.75, 0) {$\alpha$};
		\node [style=X phase dot] (121) at (8.75, -1) {$\beta$};
		\node [style=hadamard] (122) at (9.75, 0) {};
		\node [style=hadamard] (123) at (9.75, -1) {};
		\node [style=none] (124) at (10.25, 0) {};
		\node [style=none] (125) at (10.25, -1) {};
		\node [style=Z phase dot] (126) at (9.25, 0) {$\gamma$};
		\node [style=Z phase dot] (127) at (9.25, -1) {$\delta$};
		\node [style=none] (128) at (7.75, -0.5) {$=$};
		\node [style=hadamard] (129) at (8.75, -0.5) {};
		\node [style=none] (130) at (11.25, 0) {};
		\node [style=none] (131) at (11.25, -1) {};
		\node [style=X phase dot] (132) at (12.25, 0) {$\alpha$};
		\node [style=X phase dot] (133) at (12.25, -1) {$\beta$};
		\node [style=hadamard] (134) at (13.25, 0) {};
		\node [style=hadamard] (135) at (13.25, -1) {};
		\node [style=none] (136) at (13.75, 0) {};
		\node [style=none] (137) at (13.75, -1) {};
		\node [style=Z phase dot] (138) at (12.75, 0) {$\gamma$};
		\node [style=Z phase dot] (139) at (12.75, -1) {$\delta$};
		\node [style=none] (140) at (10.75, -0.5) {$=$};
		\node [style=hadamard] (141) at (11.75, -0.5) {};
		\node [style=X dot] (142) at (11.75, 0) {};
		\node [style=X dot] (143) at (11.75, -1) {};
	\end{pgfonlayer}
	\begin{pgfonlayer}{edgelayer}
		\draw (8.center) to (0);
		\draw (0) to (24);
		\draw (24) to (2);
		\draw (2) to (27);
		\draw (27) to (4);
		\draw (2) to (29);
		\draw (29) to (3);
		\draw (3) to (25);
		\draw (25) to (1);
		\draw (1) to (16.center);
		\draw (3) to (26);
		\draw (26) to (5);
		\draw (50) to (2);
		\draw (3) to (51);
		\draw (4) to (79);
		\draw (5) to (80);
		\draw (80) to (84);
		\draw (79) to (83);
		\draw (85) to (4);
		\draw (5) to (86);
		\draw (83) to (81.center);
		\draw (84) to (82.center);
		\draw (93.center) to (95);
		\draw (95) to (102);
		\draw (102) to (98);
		\draw (98) to (110);
		\draw (110) to (104);
		\draw (104) to (106.center);
		\draw (102) to (99);
		\draw (99) to (103);
		\draw (103) to (96);
		\draw (96) to (94.center);
		\draw (103) to (97);
		\draw (97) to (111);
		\draw (111) to (105);
		\draw (105) to (107.center);
		\draw (126) to (122);
		\draw (122) to (124.center);
		\draw (127) to (123);
		\draw (123) to (125.center);
		\draw (113.center) to (120);
		\draw (120) to (129);
		\draw (129) to (121);
		\draw (121) to (114.center);
		\draw (120) to (126);
		\draw [in=180, out=0] (121) to (127);
		\draw (138) to (134);
		\draw (134) to (136.center);
		\draw (139) to (135);
		\draw (135) to (137.center);
		\draw (132) to (138);
		\draw (133) to (139);
		\draw (130.center) to (142);
		\draw (142) to (132);
		\draw (142) to (141);
		\draw (141) to (143);
		\draw (143) to (131.center);
		\draw (143) to (133);
	\end{pgfonlayer}
\end{tikzpicture}
\end{center}
In other words, the measurements induce a Clifford circuit with $\CXX := (H \otimes H) CZ (H \otimes H)$ followed by $\RZ(\gamma) \RX(\alpha) \otimes \RZ(\delta) \RX(\beta)$ followed by $H \otimes H$.  Notice that the extra $H \otimes H$ operation can be pushed into the next copy of the \scalebox{.3}{
\begin{tikzpicture}[dot/.style={circle, fill=black, draw, inner sep=1pt}]
    \node[dot] (3) [right of=1] {};
    \node[dot] (4) [below of=3] {};
    \node[dot] (5) [right of=3] {};
    \node[dot] (6) [below of=5] {};

    \path
    (3) edge (4)
    	edge (5)
    (4) edge (6);
\end{tikzpicture}
} gadget for larger states.  So, for example, the state $\ket{\mathcal H_1}$ after measurement is the following: 
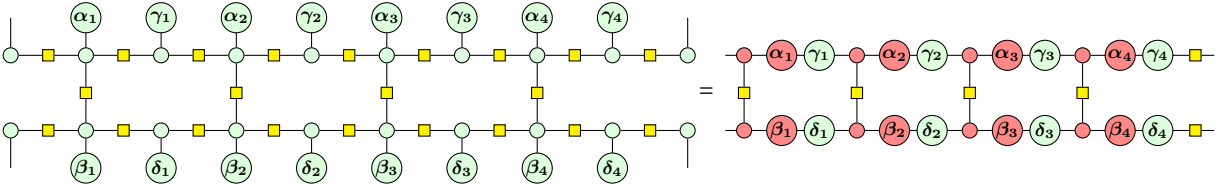
\begin{center}
	\begin{tikzpicture}
	\begin{pgfonlayer}{nodelayer}
		\node [style=Z dot] (0) at (0.25, 0) {};
		\node [style=Z dot] (1) at (0.25, -1) {};
		\node [style=Z dot] (2) at (1.25, 0) {};
		\node [style=Z dot] (3) at (1.25, -1) {};
		\node [style=Z dot] (4) at (2.25, 0) {};
		\node [style=Z dot] (5) at (2.25, -1) {};
		\node [style=none] (8) at (0.25, 0.5) {};
		\node [style=none] (16) at (0.25, -1.5) {};
		\node [style=hadamard] (24) at (0.75, 0) {};
		\node [style=hadamard] (25) at (0.75, -1) {};
		\node [style=hadamard] (26) at (1.75, -1) {};
		\node [style=hadamard] (27) at (1.75, 0) {};
		\node [style=hadamard] (29) at (1.25, -0.5) {};
		\node [style=Z phase dot] (50) at (1.25, 0.5) {$\alpha_1$};
		\node [style=Z phase dot] (51) at (1.25, -1.5) {$\beta_1$};
		\node [style=hadamard] (79) at (2.75, 0) {};
		\node [style=hadamard] (80) at (2.75, -1) {};
		\node [style=Z phase dot] (85) at (2.25, 0.5) {$\gamma_1$};
		\node [style=Z phase dot] (86) at (2.25, -1.5) {$\delta_1$};
		\node [style=Z dot] (140) at (3.25, 0) {};
		\node [style=Z dot] (141) at (3.25, -1) {};
		\node [style=Z dot] (142) at (4.25, 0) {};
		\node [style=Z dot] (143) at (4.25, -1) {};
		\node [style=hadamard] (148) at (3.75, -1) {};
		\node [style=hadamard] (149) at (3.75, 0) {};
		\node [style=hadamard] (150) at (3.25, -0.5) {};
		\node [style=Z phase dot] (151) at (3.25, 0.5) {$\alpha_2$};
		\node [style=Z phase dot] (152) at (3.25, -1.5) {$\beta_2$};
		\node [style=Z phase dot] (159) at (4.25, 0.5) {$\gamma_2$};
		\node [style=Z phase dot] (160) at (4.25, -1.5) {$\delta_2$};
		\node [style=Z dot] (163) at (5.25, 0) {};
		\node [style=Z dot] (164) at (5.25, -1) {};
		\node [style=Z dot] (165) at (6.25, 0) {};
		\node [style=Z dot] (166) at (6.25, -1) {};
		\node [style=hadamard] (169) at (4.75, 0) {};
		\node [style=hadamard] (170) at (4.75, -1) {};
		\node [style=hadamard] (171) at (5.75, -1) {};
		\node [style=hadamard] (172) at (5.75, 0) {};
		\node [style=hadamard] (173) at (5.25, -0.5) {};
		\node [style=Z phase dot] (174) at (5.25, 0.5) {$\alpha_3$};
		\node [style=Z phase dot] (175) at (5.25, -1.5) {$\beta_3$};
		\node [style=Z phase dot] (182) at (6.25, 0.5) {$\gamma_3$};
		\node [style=Z phase dot] (183) at (6.25, -1.5) {$\delta_3$};
		\node [style=Z dot] (186) at (7.25, 0) {};
		\node [style=Z dot] (187) at (7.25, -1) {};
		\node [style=Z dot] (188) at (8.25, 0) {};
		\node [style=Z dot] (189) at (8.25, -1) {};
		\node [style=hadamard] (192) at (6.75, 0) {};
		\node [style=hadamard] (193) at (6.75, -1) {};
		\node [style=hadamard] (194) at (7.75, -1) {};
		\node [style=hadamard] (195) at (7.75, 0) {};
		\node [style=hadamard] (196) at (7.25, -0.5) {};
		\node [style=Z phase dot] (197) at (7.25, 0.5) {$\alpha_4$};
		\node [style=Z phase dot] (198) at (7.25, -1.5) {$\beta_4$};
		\node [style=hadamard] (199) at (8.75, 0) {};
		\node [style=hadamard] (200) at (8.75, -1) {};
		\node [style=none] (201) at (9.25, 0.5) {};
		\node [style=none] (202) at (9.25, -1.5) {};
		\node [style=Z dot] (203) at (9.25, 0) {};
		\node [style=Z dot] (204) at (9.25, -1) {};
		\node [style=Z phase dot] (205) at (8.25, 0.5) {$\gamma_4$};
		\node [style=Z phase dot] (206) at (8.25, -1.5) {$\delta_4$};
		\node [style=none] (207) at (9.5, -0.5) {$=$};
		\node [style=hadamard] (220) at (10, -0.5) {};
		\node [style=X phase dot] (221) at (10.5, 0) {$\alpha_1$};
		\node [style=X phase dot] (222) at (10.5, -1) {$\beta_1$};
		\node [style=Z phase dot] (225) at (11, 0) {$\gamma_1$};
		\node [style=Z phase dot] (226) at (11, -1) {$\delta_1$};
		\node [style=hadamard] (233) at (11.5, -0.5) {};
		\node [style=X phase dot] (234) at (12, 0) {$\alpha_2$};
		\node [style=X phase dot] (235) at (12, -1) {$\beta_2$};
		\node [style=Z phase dot] (236) at (12.5, 0) {$\gamma_2$};
		\node [style=Z phase dot] (237) at (12.5, -1) {$\delta_2$};
		\node [style=hadamard] (246) at (13, -0.5) {};
		\node [style=X phase dot] (247) at (13.5, 0) {$\alpha_3$};
		\node [style=X phase dot] (248) at (13.5, -1) {$\beta_3$};
		\node [style=Z phase dot] (249) at (14, 0) {$\gamma_3$};
		\node [style=Z phase dot] (250) at (14, -1) {$\delta_3$};
		\node [style=hadamard] (259) at (14.5, -0.5) {};
		\node [style=X phase dot] (260) at (15, 0) {$\alpha_4$};
		\node [style=X phase dot] (261) at (15, -1) {$\beta_4$};
		\node [style=hadamard] (262) at (16, 0) {};
		\node [style=hadamard] (263) at (16, -1) {};
		\node [style=none] (264) at (16.25, 0) {};
		\node [style=none] (265) at (16.25, -1) {};
		\node [style=Z phase dot] (268) at (15.5, 0) {$\gamma_4$};
		\node [style=Z phase dot] (269) at (15.5, -1) {$\delta_4$};
		\node [style=none] (270) at (9.75, 0) {};
		\node [style=none] (271) at (9.75, -1) {};
		\node [style=X dot] (272) at (10, 0) {};
		\node [style=X dot] (273) at (10, -1) {};
		\node [style=X dot] (274) at (11.5, 0) {};
		\node [style=X dot] (275) at (11.5, -1) {};
		\node [style=X dot] (276) at (13, 0) {};
		\node [style=X dot] (277) at (13, -1) {};
		\node [style=X dot] (278) at (14.5, 0) {};
		\node [style=X dot] (279) at (14.5, -1) {};
	\end{pgfonlayer}
	\begin{pgfonlayer}{edgelayer}
		\draw (8.center) to (0);
		\draw (0) to (24);
		\draw (24) to (2);
		\draw (2) to (27);
		\draw (27) to (4);
		\draw (2) to (29);
		\draw (29) to (3);
		\draw (3) to (25);
		\draw (25) to (1);
		\draw (1) to (16.center);
		\draw (3) to (26);
		\draw (26) to (5);
		\draw (50) to (2);
		\draw (3) to (51);
		\draw (4) to (79);
		\draw (5) to (80);
		\draw (85) to (4);
		\draw (5) to (86);
		\draw (140) to (149);
		\draw (149) to (142);
		\draw (140) to (150);
		\draw (150) to (141);
		\draw (141) to (148);
		\draw (148) to (143);
		\draw (151) to (140);
		\draw (141) to (152);
		\draw (159) to (142);
		\draw (143) to (160);
		\draw (169) to (163);
		\draw (163) to (172);
		\draw (172) to (165);
		\draw (163) to (173);
		\draw (173) to (164);
		\draw (164) to (170);
		\draw (164) to (171);
		\draw (171) to (166);
		\draw (174) to (163);
		\draw (164) to (175);
		\draw (182) to (165);
		\draw (166) to (183);
		\draw (192) to (186);
		\draw (186) to (195);
		\draw (195) to (188);
		\draw (186) to (196);
		\draw (196) to (187);
		\draw (187) to (193);
		\draw (187) to (194);
		\draw (194) to (189);
		\draw (197) to (186);
		\draw (187) to (198);
		\draw (188) to (199);
		\draw (189) to (200);
		\draw (200) to (204);
		\draw (199) to (203);
		\draw (205) to (188);
		\draw (189) to (206);
		\draw (203) to (201.center);
		\draw (204) to (202.center);
		\draw (140) to (79);
		\draw (141) to (80);
		\draw (169) to (142);
		\draw (170) to (143);
		\draw (166) to (193);
		\draw (165) to (192);
		\draw (270.center) to (272);
		\draw (272) to (220);
		\draw (220) to (273);
		\draw (273) to (271.center);
		\draw (272) to (221);
		\draw (221) to (225);
		\draw (225) to (274);
		\draw (274) to (233);
		\draw (274) to (234);
		\draw (234) to (236);
		\draw (236) to (276);
		\draw (276) to (247);
		\draw (247) to (249);
		\draw (249) to (278);
		\draw (278) to (260);
		\draw (260) to (268);
		\draw (268) to (262);
		\draw (262) to (264.center);
		\draw (278) to (259);
		\draw (276) to (246);
		\draw (273) to (222);
		\draw (222) to (226);
		\draw (226) to (275);
		\draw (275) to (233);
		\draw (275) to (235);
		\draw (235) to (237);
		\draw (237) to (277);
		\draw (277) to (246);
		\draw (277) to (248);
		\draw (248) to (250);
		\draw (250) to (279);
		\draw (279) to (259);
		\draw (279) to (261);
		\draw (261) to (269);
		\draw (269) to (263);
		\draw (263) to (265.center);
	\end{pgfonlayer}
\end{tikzpicture}
\end{center}
This is example is particularly relevant given the following fact, which can be proved through computational brute force.

\begin{fact}
	Every two-qubit Clifford operation can be decomposed as 
	$$
	\prod_{i=1}^4 \( [\RZ(\gamma_i) \RX(\alpha_i) \otimes \RZ(\delta_i) \RX(\beta_i)] \CXX \).
	$$
	for $\alpha_i, \beta_i, \gamma_i, \delta_i \in \{0, \pi/2, \pi, 3\pi/2\}$.
\end{fact}

Therefore, if we have some large $\ket{\mathcal H_n}$ state, we can measure 16 consecutive qubits to implement any two-qubit gate.  The exception is the final 16 qubits which will have an extra $H \otimes H$ applied to the end.  However, instead of setting the $\alpha_i, \beta_i, \gamma_i, \delta_i$ parameters to generate some gate $g$, we can set them to generate $Hg$.  Combining the above observations completes the proof of \Cref{thm:two_qubit_mbqc}.


\section{Sampling vs. Relation Problems}
\label{app:samp_vs_search}

The purpose of this appendix is to discuss the difference between sampling and relation problems as they pertain to this work and related literature.\footnote{For more precise definitions of what we mean by the various models of computation, we refer to the reader to the \Cref{sec:models_of_computation}.}  An obvious first observation is that any model of computation which can solve a sampling problem can also solve the analogous relation problem.  

The reverse direction is not nearly as obvious.  On the one hand, Aaronson shows that search problems can often be constructed such that they are just as hard as as sampling problems \cite{aaronson:2014_equivalence}.  In particular, $\Rel\BPP = \Rel\BQP$ if and only if $\Samp\BPP = \Samp\BQP$.  On the other hand, this does \emph{not} mean that the ability to solve a relation problem immediately grants the power to solve the analogous search problem.  For example, consider the following quantum circuit:  with probability $1/2$, the circuit outputs the all-zeroes string, and with probability $1/2$, the circuit samples according to one of the many quantum supremacy proposals \cite{aaronsonarkhipov:2013, boixo:2018_rcs, bjs:2010_iqp}.  Clearly, the relation problem associated with the output distribution is in $\Rel\P$ (or even $\Rel\NC^0$), but the associated sampling problem is hard for an efficient classical machine unless the polynomial hierarchy collapses.

In the first section below, we show that such a straightforward equivalence \emph{does} exist for the controlled-Clifford sampling problems in this paper and related literature (e.g., \cite{bgk:2018, bgkt:2019}) in the non-interactive setting.  In the second section, we will discuss the difference when interaction is allowed.  In particular, we will show that forcing the classical simulator to \emph{sample} in the second round of our protocol (i.e., $\CliffSim$) makes the proofs of the main results dramatically simpler.

\subsection{Non-interactive setting}
Consider the standard problem of outputting $X$ and $Y$ measurements for some graph state.  We will usually think of such measurements as first applying an $H$ or $H \RZ$ gate (corresponding to $X$ and $Y$ measurements, respectively), and then measuring in the computational basis.  In order to talk about the valid measurement outcomes, we define the support of a quantum state $\ket{\psi} = \sum_{x \in \{0,1\}^n} \alpha_x \ket{x}$ as
$$
\supp(\ket{\psi}) := \{ x : \alpha_x \neq 0 \}.
$$
We would like to show that for Clifford $\ket{\psi}$, sampling from $\supp(\ket{\psi})$ is equivalent to outputting any element of $\supp(\ket{\psi})$ under suitably efficient reductions.  We give the following simple reduction:

\begin{lemma}
\label{lem:search_to_samp}
Let $U \in \Clifford_n$ be a Clifford unitary with tableau $[\begin{smallmatrix} A & B \\ C & D \end{smallmatrix}]$.  Furthermore, let $z \in \supp(U\ket{0}^{\otimes n})$ be any element in the support, and let $r \in \{0,1\}^n$ be a uniformly random $n$-bit string.  Then, $r C \oplus z \in \supp(U\ket{0}^{\otimes n})$ is a uniformly random element of the support.
\end{lemma}
\begin{proof}
Let us first start with the trivial observation that 
$$
z \in \supp(U\ket{0}^{\otimes n}) \iff z \oplus a \in \supp(\alpha \pfont{X}^a \pfont{Z}^b U\ket{0}^{\otimes n})
$$
for any strings $a, b \in \{0,1\}^n$ and phase $\alpha \in \{\pm 1, \pm i\}$.  In fact, this will give us a simple way to move between elements of the support.  Since $U$ is Clifford, for any $c \in \{0,1\}^n$, we have $U \pfont Z^c = (U \pfont Z^c U^\dag) U = \alpha \pfont X^a \pfont Z^b U$ for strings $a, b \in \{0,1\}^n$.  Therefore,
\begin{align*}
z \in \supp(U\ket{0}^{\otimes n}) &\iff z \in \supp(U \pfont Z^c \ket{0}^{\otimes n}) \\
&\iff z \in \supp(\alpha \pfont{X}^a \pfont{Z}^b U \ket{0}^{\otimes n}) \\
&\iff z \oplus a \in \supp(U\ket{0}^{\otimes n})
\end{align*}
where the first equivalence comes from the fact that $\pfont Z^c$ is a stabilizer of $\ket{0}^{\otimes n}$ for any $c \in \{0,1\}^n$.  Using the above observation and \Cref{fact:tableau_mult}, we have that for any $r \in \{0,1\}^n$, $rC \oplus z \in \supp(U\ket{0}^{\otimes n})$.  

Let us now turn to uniformity.  Recall that a single-qubit measurement on a Clifford state is random if there is some generator of the state has an $\pfont X$ or $\pfont Y$ Pauli at that location.  We note that measurements with random outcomes effectively perform Gaussian elimination on the matrix $C$ (see \Cref{sec:upper_bounds} for more details on measurement).  Therefore, we get that $\rank(C) = \log_2 |\supp(U\ket{0}^{\otimes n})|$, and the image of $rC \oplus z$ over all $r$ equals $\supp(U\ket{0}^{\otimes n})$.  Since output of any affine function is uniform over its image given a uniform input, we get that $rC \oplus z$ is uniform over $\supp(U\ket{0}^{\otimes n})$ for uniform $r$.
\end{proof}

So what is the Clifford tableau corresponding to $X$ and $Y$ measurements on a graph state?  Let $A \in \{0,1\}^{n \times n}$ be the adjacency matrix for some simple undirected graph on $n$ vertices, and let $b \in \{0,1\}^n$ specify the measurement basis for each qubit (0 for $X$ measurement, and 1 for $Y$ measurement).  Note that $A$ and $b$ are exactly the inputs to the Hidden Linear Function ($\HLF$) problem of Bravyi, Gosset, and K\"onig \cite{bgk:2018}, and an answer to the measurement problem is equivalent to an answer to the $\HLF$ problem.

Starting from the all-zeros state, the exact sequence of gates is the following:  Hadamard on all qubits, $\CSIGN$ between adjacent qubits, $\RZ$ on $Y$-basis qubits, and finally another layer of Hadamard gates.  Using the tableau manipulations of Aaronson and Gottesman \cite{ag:2004}, we can track the tableau (ignoring signs):
$$
\begin{bmatrix} I_n & 0 \\ 0 & I_n \end{bmatrix} \xrightarrow{H} 
\begin{bmatrix} 0 & I_n \\ I_n & 0 \end{bmatrix} \xrightarrow{\CSIGN}
\begin{bmatrix} 0 & I_n \\ I_n & A \end{bmatrix} \xrightarrow{\RZ}
\begin{bmatrix} 0 & I_n \\ I_n & A + B \end{bmatrix} \xrightarrow{H}
\begin{bmatrix} I_n & 0 \\ A+B & I_n \end{bmatrix}
$$
where $B = \diag(b)$.

We're now ready to show that the relation version of $\HLF$ is equivalent to the sampling version.  For clarity, we will denote these problems as $\Samp\HLF$ and $\Rel\HLF$ for the sampling and relation problems, respectively.  Previously, we have used $\HLF$ to refer to $\Rel\HLF$.

\begin{theorem}
\label{thm:sampHLF_equals_relHLF}
$\Samp\HLF \in (\Samp\NC^0)^{\Rel\HLF}$.
\end{theorem}
\begin{proof}
Let some instance of the $\Samp\HLF$ problem be specified by matrix $A \in \{0,1\}^{n \times n}$ and vector $b \in \{0,1\}^n$.  Let $z \in \{0,1\}^n$ be any answer to the $\Rel\HLF$ problem.  By \Cref{lem:search_to_samp}, we have that $r (A \oplus \diag(b)) \oplus z$ solves the $\Samp\HLF$ problem for uniform $r \in \{0,1\}$.  All that remains to show is that an $\NC^0$ circuit can compute this affine function.  Generally, such a computation requires a parity gate (i.e., an $\NC^0[2]$ circuit rather than a $\NC^0$ circuit).  However, we can use that for the $\HLF$ problem, $A$ is sparse since it encodes a graph which is embedded in a 2D grid.  That is, multiplying $r$ by a column of $(A \oplus \diag(b))$ is only a function of constantly many (predetermined) entries of $r$ and the matrix.  Since we can get $z$ from a single call to the $\Rel\HLF$ oracle, the entire circuit can be constructed in constant depth with with bounded fan-in gates.
\end{proof}

\subsection{2-round interactive setting}

We would like to prove a similar reduction from \emph{interactive} sampling tasks to \emph{interactive} relation tasks, since there is a tantalizing possibility of simplifying our proof. Specifically, there are easy measurements to distinguish $\ket{00}$ from $\ket{++}$ if we can assume the outcome is a \emph{sample} rather than a deterministic, adversarial choice. That is, if we assume the rewind oracle produces independent samples in the second round, every time we rewind, then we use it as follows. 
\begin{observation}
\label{obs:simple_nc1_proof}
Let $\coracle$ be the sampling rewind oracle for the $2$-round Clifford simulation problem (\Cref{prob:nc1}). Then $$\NC^{1} \subseteq (\BPAC^{0})^{\coracle}.$$  
\end{observation}
\begin{proof}
Let $U_1, \ldots, U_n \in \Clifford_2 / \Paulis_2$ be given as input.  By \Cref{cor:nc1_hardness}, it is $\NC_1$-hard to determine if $U_1 \cdots U_n$ is either the identity or $H \otimes H$ (modulo Paulis), promised that one is the case.  In round 1, let the measurements correspond to measurement-based computation of these unitaries.  That is, the state on the remaining qubits is either $P \ket{00}$ or $Q \ket{++}$ for some Paulis $P, Q \in \Paulis_2$, and distinguishing between them solves an $\NC^1$-hard problem.  However, if the classical machine can sample from the conditional distribution of measurement outcomes for these states, then it can easily distinguish between them since $|\supp(P \ket{00})| =1$ and $|\supp(Q \ket{++})| > 1$.
\end{proof}

In light of this, we adapt \Cref{thm:sampHLF_equals_relHLF} to the interactive case directly in the following lemma. 
\begin{lemma}
\label{lem:multiround_rel_to_samp}
Consider a $k$-round interactive Clifford simulation task $T$ that consists of measuring qubits of a graph state in the $X$ or $Y$ basis. Let $\oracle_{rel}$ be an oracle implementing the relation version of task $T$ (i.e., it suffices to output any possible outcome in each round). We construct an oracle $\oracle_{samp}$ which solves a sampling version of task $T$. Namely, the distribution over transcripts of the entire interaction is identical to the distribution over interactions for a quantum device.
\end{lemma}
\begin{proof}
The strategy is simple---we use the algorithm from \Cref{thm:sampHLF_equals_relHLF} but wait until the appropriate round to process any particular qubit.  Suppose the initial graph state is specified by matrix $A \in \{0,1\}^{n \times n}$.  Before the protocol starts, generate a random $r \in \{0,1\}^n$.  Then, for every query to the relation oracle $\oracle_{rel}$ on qubits in the subset $I \subseteq [n]$, return the outcomes $rA_i \oplus r_ib_i \oplus z_i$ for all $i \in I$, where $z_i$ is the output of the relation oracle and $b_i$ is 0 for measurements in the $X$ basis and 1 for measurements in the $Y$ basis.  

Notice that after all rounds of interaction, this procedure samples from the uniform distribution over outcomes for some fixed $A$ and $b$ by \Cref{lem:search_to_samp}, and so the transcript for that interaction is uniform over transcripts with those measurement bases.
\end{proof}

It is worth emphasizing that the sampling oracle $\oracle_{samp}$ is insufficient to simplify our proof as is done in \Cref{obs:simple_nc1_proof}, since that observation requires the \emph{rewind oracle} to sample from the correct distribution. Of course, any classical implementation of $\oracle_{samp}$ still gives a classical rewind oracle $\coracle_{samp}$, but rewinding and running the round again (with the same input) may \emph{not} produce an independent output. Consider the construction of the sampling oracle $\oracle_{samp}$ in \Cref{lem:multiround_rel_to_samp}.  Since it does not use any additional randomness in the second round, its output is a deterministic function of the $\oracle_{rel}$ output (which may be deterministic) and its internal state. When we construct the rewind oracle, this manifests as the second round outputs being non-independent.  

In some sense, we \emph{can} satisfactorily randomize an oracle (i.e., go from relational simulation to sampling) for the $2$-round Clifford simulation task, but the construction goes through our results (rather than simplifying them), and it is not black box (since it requires using the rewind oracle).  To summarize, we use \Cref{lem:multiround_rel_to_samp} to randomize the first round outputs, and then use the power of the rewind oracle to completely learn the state of the remaining qubits. Since we have a complete classical description of the state and it is only constant size,\footnote{Notice that this step is not completely general.  For the $\parityL$ result, it is critical that the state $U_1 \cdots U_n \ket{+}^{\otimes m}$ is promised to be relatively simple.} we can simulate whichever measurements we need in the second round (shifting the result as required by \Cref{lem:multiround_rel_to_samp}).

\addtocontents{toc}{\endgroup}
\end{appendices}

\end{document}